\documentclass[12pt]{article}
\usepackage[utf8]{inputenc}
\usepackage{xcolor}
\definecolor{darkblue}{rgb}{0, 0, 0.699219}
\definecolor{darkgreen}{rgb}{0, 0.699219, 0}
\definecolor{darkred}{rgb}{0.699219, 0, 0}
\definecolor{purple}{rgb}{0.5, 0, 0.25}
\definecolor{webbrown}{rgb}{0.601562, 0, 0}
\definecolor{webgreen}{rgb}{0, 0.398438, 0}
\usepackage{amsmath}
\usepackage{amsthm}
\usepackage{amssymb}
\usepackage{graphicx}
\usepackage{setspace}
\usepackage{esint}
\usepackage[authoryear]{natbib}
\usepackage[bookmarks=false,breaklinks=false,pdfborder={0 0 1},backref=section,colorlinks=false]{hyperref}
\hypersetup{colorlinks,citecolor=darkred,filecolor=black,linkcolor=darkblue,urlcolor=webgreen,pdfstartview=FitH}

\makeatletter
\usepackage{euscript}
\usepackage{amsfonts}
\usepackage{caption}
\usepackage{soul}

\usepackage{comment}
\usepackage{color}
\usepackage{fullpage}

\newtheorem{lemma}{{\bf Lemma}}\newtheorem{proposition}{{\bf Proposition}}\newtheorem{corollary}{{\bf Corollary}}\newtheorem{defn}{{\bf Definition}}\newtheorem{obs}{{\bf Observation}}\newtheorem{claim}{{\bf Claim}}\newtheorem{remark}{{\bf Remark}}
\renewcommand{\thesublemma}{\thelemma\Alph{sublemma}}

\newcommand{\tlb}{\underline\theta}
\newcommand{\tub}{\overline{\theta}}

\newcommand{\qopt}{q^{\rm {{OPT}}}}
\newcommand{\qbm}{q^{\rm {{BM}}}}
\newcommand{\popt}{p^{\rm {{OPT}}}}
\newcommand{\uopt}{u^{\rm {{OPT}}}}

\newcommand{\Lopt}{\Lambda^{\rm {OPT}}}
\newcommand{\muopt}{\mu^{\rm {OPT}}}

\title{{\bf Robust Procurement: 
Bayesian Design under Worst-Case Approval Constraints}\thanks{{An earlier version of the paper was circulated under the title ``Robust Procurement Design."} For comments and useful suggestions, we are grateful to seminar participants and discussants at various conferences and workshops where the paper was presented. We are also thankful to the referees
of the 26th ACM Conference on Economics and Computation (EC’25). Debasis Mishra acknowledges financial support from ANRF through Project ANRF/ARGM/2025/00572/QSS.}}
\author{\href{https://www.isid.ac.in/~dmishra/}{Debasis Mishra}, \href{https://sites.google.com/view/sanketrpatil}{Sanket Patil}, and \href{https://faculty.wcas.northwestern.edu/apa522/}{Alessandro Pavan}\thanks{Mishra: Indian Statistical Institute, Delhi, \texttt{dmishra@isid.ac.in}. Patil: Indian Institute of Management Bangalore, \texttt{sanket.patil@iimb.ac.in}. Pavan: Northwestern University, \texttt{alepavan@northwestern.edu}.}}

\makeatother

\usepackage{nameref}
\begin{document}
\allowdisplaybreaks \maketitle 
\begin{abstract}
We study optimal procurement when a Bayesian designer must obtain
approval from a non-Bayesian authority that shares the designer's
objective but is uncertain about the value of the good and the supplier's
cost. The designer uses a conjectured model to compute expected payoffs
but is constrained to select among mechanisms delivering the largest
payoff guarantee to the authority. This robustness requirement reshapes
the tradeoff between efficiency and rent extraction: it reduces procurement
from sellers with intermediate costs but may increase it from those
with a high cost. When the good is sold in a market, we show that
quantity regulation dominates price regulation if markups under the
conjectured model are large, whereas price regulation dominates when
demand uncertainty is substantial. 
\end{abstract}


\newpage{}


\section{Introduction}

The standard approach to contract design assumes that the designer
has a conjectured model of the environment, derived from data or subjective
beliefs, and selects the mechanism maximizing expected payoff subject
to incentive and participation constraints (the so-called Bayesian
approach). In many applications, however, the designer must obtain
approval from an authority that is skeptical of this model. Such authorities
often require contracts to perform satisfactorily across a range of
scenarios. For example, procurement contracts in the oil and gas sector
must remain viable under conservative price assumptions.

In this paper, we study contract design when a Bayesian designer must
secure approval from a max-min authority. The designer first identifies
mechanisms that maximize the authority’s worst-case payoff over a
set of admissible models and then selects among them using her conjectured
model. We refer to the mechanisms solving this problem as \emph{robustly
optimal}. When the designer and the authority share the same objective,
this formulation also captures a designer who seeks robustness against
her own model mis-specification.

We conduct our analysis in a procurement setting, where buyers (governments
but also private institutions) typically face suppliers with private
information about costs, and where contracts typically require approval
from authorities---regulators, supervisors, or senior management---who
often do not share the designer’s confidence in the model used to
estimate the value of the good and the technology determining the
supplier's cost.\footnote{While the need for approval is particularly salient in procurement,
the robustness approach developed in this paper whereby the designer
uses a conjectured model to select among worst-case optimal mechanisms
applies also to other design problems.}

We show that robustness reshapes the tradeoff between efficiency and
rent extraction. We begin with environments where uncertainty only
concerns the seller’s cost distribution. When the conjectured model
satisfies standard regularity conditions, the robustly optimal mechanism
is as in \citet{BM82} but with a quantity floor: all types must supply
at least the efficient output of the highest-cost type. Relative to
the Bayesian benchmark, robustness increases the procurement from
high-cost sellers and entails efficiency at both extremes of the cost
distribution. These changes protect against the possibility that high-cost
firms occur more frequently than conjectured. Because downward distortions
serve to limit rents for low-cost types, their value diminishes under
cost uncertainty. Procurement from low-cost types, instead, is the
same as under the Bayesian optimum, as welfare when contracting with
these types exceeds the worst-case level even when the conjectured
cost distribution is wrong.

When uncertainty also affects demand, this mechanism remains optimal
under suitable conditions, which we characterize. Otherwise, robustness
leads to reduced procurement from intermediate-cost types, reflecting
the need to guard against demand overestimation. At these cost levels,
the screening motive and the need for approval work in the same direction,
leading to further downward distortions relative to the Bayesian benchmark.

We then study environments with downstream markets, where demand can
be learned by asking sellers to post prices. We characterize robustly
optimal price regulations, under which sellers can choose from a menu
of markups computed under the designer's conjectured model, and where
the final transfer to the seller responds to the realized demand,
making the mechanism ex-post individually rational and incentive compatible.
The markups are the same as under the Bayesian optimum but with a
cap binding for high-cost types. The cap protects the buyer against
the possibility that high-cost sellers are more prevalent than conjectured
by the designer.

Finally, we compare price and quantity regulation. Under standard
Bayesian analysis, price regulation is preferred because it allows
quantities to adjust to realized demand. Under robustness, however,
the two instruments deliver the same worst-case payoff guarantee but
differ in their expected performance under the designer’s model. Quantity
regulation dominates when optimal markups are high, whereas price
regulation performs better when demand uncertainty is substantial.
The intuition is as follows. When markups are high, the designer seeks
to procure little from most sellers, so fixing quantities effectively
limits the buyer’s exposure to the designer's demand mis-specification.
By contrast, price regulation, by capping the seller’s price to protect
the buyer against unexpectedly high costs, forces the designer to
over-procure from high-cost sellers relative to quantity regulation,
thereby reducing welfare.

\medskip{}
 \textbf{Organization}. We conclude the introduction with a brief
discussion of the relevant literature. Section \ref{sec:Model} presents
the environment and the designer's problem. Section \ref{sec:robustly-quantity-mechanisms}
characterizes robustly optimal mechanisms. Section \ref{sec:price_vs_quantity}
extends the analysis to settings with a downstream market, characterizes
robustly optimal price regulations, and identifies conditions under
which quantity regulation dominates (or is dominated by) price regulation.
Section \ref{sec:Conclusions} concludes. Proofs not in the main text
are in the \nameref{Sec:appendix} at the end of the document. \nameref{Sec:OS}
contains extensions of our results to more permissive robustness requirements
and alternative forms of cost uncertainty. \nameref{Sec:additional_supplement}
contains extra material, namely, proof of existence and undomination
of robustly optimal mechanisms, and additional comparative statics
of robustly optimal mechanisms.

\medskip{}

\textbf{Related literature}. A vast literature in information economics
studies optimal mechanisms when agents hold private information and
the designer optimizes over all mechanisms that are incentive compatible
and individually rational under a conjectured model (the Bayesian
approach). The closest work to ours is \citet{BM82}.

Recent work relaxes the key assumptions underlying the Bayesian framework
to develop robust approaches to contract design; \citet{C19} provides
an excellent survey. Closely related contributions include \citet{G14},
\citet{BHM25}, \citet{GS24}, and \citet{K26}. \citet{G14} characterizes
optimal contracts when the designer lacks information about a manager’s
disutility from effort in the \citet{LT86} model. \citet{BHM25}
study robust procurement contracts evaluated by worst-case competitive
ratios. \citet{GS24} analyze minimax-regret design and identify conditions
for the optimality of price-cap regulation under private information
about both demand and costs.\footnote{See also \citet{BS08} for an earlier analysis of minimax-regret design
in monopoly pricing and \citet{S03} for the design of multi-unit
auctions when demand is unknown.} \citet{K26} studies lexicographic worst-case optimality across multiple
beliefs. Our analysis is also lexicographic---restricting the designer's
choice to worst-case optimal mechanisms---but differs in that second-order
beliefs are given by the designer’s conjectured model, a primitive
of the environment. This distinction is important because it implies,
among other things, that the tension between efficiency and rent extraction
persists under robustness concerns. None of the above papers identify
conditions under which quantity regulation outperforms price regulation
under robustness, which is one of our main contributions. Related
are also \citet{BLW25} and \citet{MP25}, who study undominated mechanisms
in auctions and regulation.\footnote{See \nameref{Sec:additional_supplement} for a discussion of the connection
between robust optimality (as defined in this paper) and undomination.} In these papers, the designer treats all admissible models symmetrically,
whereas in our analysis the designer uses the conjectured model to
select among mechanisms that maximize the payoff guarantee. This feature
aligns our approach with \citet{DP22}, who study robust information
design without transfers or screening.

Our paper is also related to the literature on model misspecification
and robust control (see \citet{CVHMM} for an excellent overview).
In that literature, alternative models are weighted by their distance
from the conjectured one.\footnote{See also \citet{BK26} for an alternative approach to robustness where
payoff guarantees are required to be continuous with respect to the
set of admissible models.} In contrast, our designer treats all admissible models symmetrically
when computing the payoff guarantee, but then uses the conjectured
model to select among the worst-case optimal ones. In this respect,
our approach also differs from \citet{MP17}, who study misspecification
in monopolistic screening and show that small errors can generate
large losses.

Finally, in environments with downstream markets, our analysis connects
to \citet{VSP25}, who study markets as instruments for implementation.
\citet{Kang26} compares the worst-case performance of price versus quantity
regulation when the correlation between willingness to pay and marginal
externalities is unknown. By contrast, our comparison hinges on the
equivalence of price and quantity regulation under worst-case optimality
and their divergence under the designer’s conjectured model. See also
\citet{W74} for the classic analysis of price versus quantity regulation
without private information or robustness concerns.

\section{Model}

\label{sec:Model}

\subsection{Environment}

\label{subsec:Environment}

A buyer (a government agency or a private organization) delegates
to an expert the design of a contract for the procurement of a product
or service. The good is supplied by a monopolistic seller who incurs
a cost of $\theta{\rm q}$ for supplying ${\rm q}\ge0$ units of the
good, where $\theta\ge0$ is private information to the seller.\footnote{The results extend to a setting with a fixed cost $c\in(0,\overline{c})$,
where the bound $\overline{c}\in\mathbb{R}_{++}$ ensures that
it is optimal for the buyer to procure a positive quantity even when
the marginal cost is the largest and the value of the good is the
lowest.}

The designer has a model $(V^{\star},F^{\star})$, where the function
$V^{\star}:\mathbb{R}_{+}\rightarrow\mathbb{R}_{+}$ represents the
designer's estimate for the value from procuring the good. The function
$F^{\star}$ represents the estimate for the cost technology, i.e.,
the distribution from which the seller's marginal cost $\theta$ is
drawn; $F^{\star}$ is absolutely continuous over $\mathbb{R}$ with
density $f^{\star}$ strictly positive over $\Theta\equiv[\underline{\theta},\overline{\theta}]\subset\mathbb{R}_{++}$
and $F^{\star}(\underline{\theta})=0$, $F^{\star}(\overline{\theta})=1$,
and the virtual cost $z^{\star}(\theta)\equiv\theta+F^{\star}(\theta)/f^{\star}(\theta)$
is continuous and increasing over $\Theta$.\footnote{Throughout, a function $g$ is said to be increasing (alternatively,
weakly increasing) if $g(x)>g(x')$ (alternatively, $g(x)\ge g(x')$)
whenever $x>x'$. The terms \textit{decreasing} and \textit{weakly
decreasing} are analogously defined. Similarly, when we say that a
function is concave, convex, or quasi-concave, we mean strictly.} The designer is an expected-utility maximizer and evaluates the performance
of any mechanism under the conjectured model $(V^{\star},F^{\star})$.

The buyer does not trust the designer's confidence in the model. She
is concerned that the true value of the good may be given by a function
$V$ different from $V^{\star}$ and the true technology describing
the seller's cost may be given by a distribution $F$ different from
$F^{\star}$. Specifically, the buyer has a set of admissible models
$\mathcal{A}\equiv\mathcal{V}\times\mathcal{F}$, where $\mathcal{V}$
is a set of (gross) value functions and $\mathcal{F}$ is a set of
distributions over $\Theta$. The set $\mathcal{F}$ is the set of
all cdfs whose support is contained in $\Theta$; that is, each $F\in\mathcal{F}$
is a weakly increasing, right-continuous function $F:\mathbb{R}\rightarrow[0,1]$
such that $F(\theta)=0$ for all $\theta<\underline{\theta}$, and
$F(\theta)=1$ for all $\theta\geq\overline{\theta}$. \nameref{Sec:OS}
considers the case in which $\mathcal{F}$ does not coincide with
the set of all cdfs whose support is contained in $\Theta$ and discusses
how the results depend on the lowest element of such a set (in the
main text, the lowest element is a Dirac assigning probability one
to $\overline{\theta}$; in the supplement, we allow the lowest element
of $\mathcal{F}$ to be a non-degenerate cdf).

We assume that, for each $V\in\mathcal{V}$, there is an integrable
function $P:\mathbb{R}_{+}\rightarrow\mathbb{R}_{+}$, interpreted
as the inverse demand function corresponding to $V$, such that, for
all ${\rm q}\ge0$, 
\begin{equation}
V({\rm q})=\int^{{\rm q}}_{0}P(s){\rm d}s.\label{eq:gross-surplus}
\end{equation}
We let $\mathcal{P}$ represent the set of inverse demand functions
and $\mathcal{D}$ the set of direct demand functions corresponding
to the elements of the set $\mathcal{V}$ of value functions. For
each $P\in\mathcal{P}$, there exists a $\mathrm{q}_{P}\in\mathbb{R}_{++}\cup\{+\infty\}$
such that $P$ is decreasing and continuous over $[0,\mathrm{q}_{P})$
and $P(\mathrm{q})=0$ if $\mathrm{q}>\mathrm{q}_{P}$. For each $P\in\mathcal{P}$,
let $D$ denote the corresponding direct demand. For all $\text{p }\in(0,P(0)]$,
$D(\text{p})\equiv P^{-1}(\text{p})$; for all $\mathrm{p}>P(0)$,
$D(\text{p})=0$; and $\lim_{\text{p}\rightarrow0}D(\text{p})=\mathrm{q}_{P}$.

Representing each $V$ as the integral of its inverse demand function
$P$ facilitates the application to monopoly regulation in Section
\ref{sec:price mechanisms}. We assume that $V^{\star}\in\mathcal{V}$
and then denote by $P^{\star}\in\mathcal{P}$ and $D^{\star}\in\mathcal{D}$,
respectively, the inverse and direct demands associated with $V^{\star}$.

The seller can provide any quantity ${\rm q}\in[0,\bar{{\rm q}}]$,
where $\bar{{\rm q}}\in\mathbb{R}_{++}$ is such that $D(\underline{\theta})<\bar{{\rm q}}$
for all $D\in\mathcal{D}$. This assumption ensures it is never optimal
for the buyer to procure more than $\bar{{\rm q}}$ and that the seller's
equilibrium payoff satisfies a familiar envelope-theorem representation.
We also assume that $D^{\star}(z^{\star}(\theta))>0$ for all $\theta\in\Theta$;
the assumption simplifies the exposition by guaranteeing that the
quantity traded under the Bayesian optimal mechanism is strictly positive
for all types.

We assume that there exists a ``smallest'' inverse demand $\underline{P}\in\mathcal{P}$
such that $P({\rm q})\geq\underline{P}({\rm q})$ for all ${\rm q}\geq0$
and all $P\in\mathcal{P}$. The function $\underline{P}$ provides
a lower bound on marginal value and is decreasing and continuous.
The smallest inverse demand function $\underline{P}$ induces the
smallest value function $\underline{V}$. We also assume that $\overline{\theta}<\lim_{{\rm q}\downarrow0}\underline{P}({\rm q})$.
Together with the continuity of $\underline{P}$, this assumption
ensures that there are gains from trade no matter the seller's cost
and the inverse demand.

Lastly, we assume that, once the buyer learns $V$, ex-post adjustments
to output are either infeasible or not worthwhile. Such adjustments
may be prohibitively costly, or the marginal value of additional output
may be too low to justify the extra production cost. In the absence
of these frictions, uncertainty over $V$ is inconsequential.

In summary, $(V^{\star},F^{\star})$ represents the designer's conjectured
model, with $V^{\star}$ capturing the value of procuring output and
$F^{\star}$ capturing the cost technology. The conjectured model
belongs to the admissible set, i.e., $(V^{\star},F^{\star})\in\mathcal{A}$
(this also means that $P^{\star}\in\mathcal{P}$ and $D^{\star}\in\mathcal{D}$).
The buyer fears, however, that the true model may be some alternative
$(V,F)\in\mathcal{A}$.

The designer and the buyer share the same objective of maximizing
the buyer's value from the procured good net of the transfer to the
firm. However, while the designer evaluates any mechanism (described
below) under the conjectured model $(V^{\star},F^{\star})$, the buyer
evaluates it under the worst-case scenario computed over $\mathcal{A}$.
The designer seeks the buyer's approval. She understands that any
mechanism for which the payoff guarantee is not the highest will not
be approved. She thus selects a mechanism that maximizes her expected
payoff (under the model $(V^{\star},F^{\star})$) among those for
which the buyer's payoff guarantee is the highest.

\subsection{Procurement mechanisms}

\label{subsec:Procurement-mechanisms}

We start by considering environments where there are no downstream
markets that can be used to generate information about the value of
the good procured, as in the case of private procurement or the public
procurement of certain types of military equipment not traded in downstream
markets. In these environments, the only viable mechanisms are quantity
mechanisms. In this subsection, we formally define these mechanisms
and then analyze their properties until we introduce price mechanisms
in Section \ref{sec:price mechanisms}.

To elicit the seller's private information and discipline the procurement,
the designer offers a (direct) mechanism $M=(q,t)$. The \textsl{quantity
schedule} $q:\Theta\rightarrow[0,\bar{{\rm q}}]$ specifies the procured
quantity as a function of the reported cost, whereas the \textsl{transfer
schedule} $t:\Theta\rightarrow\mathbb{R}$ specifies the total payment.\footnote{The focus on deterministic mechanisms is without loss of optimality
in our setting. This is because, for any given $\theta$, the buyer's
welfare is concave in $q$. We also adopt a conservative approach
by assuming that Nature's choice of $(V,F)$ can be made contingent
on the designer's choice of the mechanism. As it will become clear
from the analysis below, this assumption, however, is not important.
The mechanisms that solve our design problem remain optimal even if
Nature is expected to choose $(V,F)$ simultaneously with the designer's
choice of $M$. This is because the solution to the designer's problem
is a saddle point.}

The mechanism $M=(q,t)$ is \textsl{incentive compatible} (IC) if,
for all $\theta,\theta'\in\Theta$, 
\[
u(\theta)\equiv t(\theta)-\theta q(\theta)\geq t(\theta')-\theta q(\theta')=u(\theta')+(\theta'-\theta)q(\theta').
\]
It is \textsl{individually rational} (IR) if $u(\theta)\ge0$ for
all $\theta\in\Theta$. Because, given the quantity schedule $q$,
there is a bijection between the transfer schedule $t$ and the utility/rent
schedule $u$, we will often refer to a mechanism by $(q,u)$ instead
of $(q,t)$.

As is standard, $M=(q,u)$ is IC and IR if and only if (a) $q$ is
weakly decreasing, and (b) for all $\theta\in\Theta$, $u(\theta)=u(\overline{\theta})+\intop^{\overline{\theta}}_{\theta}q(y){\rm d}y$,
with $u(\overline{\theta})\ge0$.

Let $\mathcal{M}$ be the set of all IC and IR mechanisms. If the
buyer's gross value is $V\in\mathcal{V}$ and the technology is $F\in\mathcal{F}$,
the buyer's \textsl{ex-ante welfare} under the mechanism $M\equiv(q,u)\in\mathcal{M}$
is 
\[
W(M;V,F)\equiv\int\left(V(q(\theta))-\theta q(\theta)-u(\theta)\right)F(\mathrm{d}\theta),
\]
where $V(q(\theta))-\theta q(\theta)-u(\theta)=V(q(\theta))-t(\theta)$
is the buyer's net value when the cost is $\theta$.

\subsection{Designer's problem}

\label{subsec:Buyer's-problem}

The designer faces a two-step optimization problem.

\noindent\textbf{Step 1 (guarantee maximization)}. The designer first
computes the shortlist of mechanisms for which the guarantee is the
highest. The \textbf{welfare guarantee} of any IC and IR mechanism
$M\in\mathcal{M}$ is given by 
\begin{align*}
G(M)\equiv & \inf_{V\in\mathcal{V},F\in\mathcal{F}}W(M;V,F).
\end{align*}
The \textbf{shortlist} of mechanisms delivering the highest guarantee
is given by 
\[
\mathcal{M}^{{\rm SL}}\equiv\arg\max_{M\in\mathcal{M}}G(M).
\]

Any mechanism $M\in\mathcal{\mathcal{M}^{{\rm SL}}}$ in the shortlist
is thus\textbf{ worst-case optimal}. In \nameref{Sec:OS}, we consider
a more permissive definition of shortlist comprising all mechanisms
for which the guarantee is no smaller than a fraction $\gamma\in[0,1]$
of the maximal attainable guarantee (formally, $\mathcal{M}^{{\rm SL}}(\gamma)\equiv\{M\in\mathcal{M}:G(M)\ge\gamma G(M')~\forall~M'\in\mathcal{M}\}$).\footnote{See also \citet{AC25} for a different problem (without screening)
featuring selection by an expected-utility maximizer over a shortlist
of alternatives defined by a worst-case threshold.}

\noindent\textbf{Step 2 (selection under conjectured model)}. The
designer then selects the mechanism from the shortlist that maximizes
expected welfare under the conjectured model $(V^{\star},F^{\star})$.

\begin{defn}A mechanism $M^{{\rm OPT}}$ is \textbf{robustly optimal}
if 
\begin{align*}
M^{{\rm OPT}}\in\arg\max_{M\in\mathcal{M}^{{\rm SL}}}W(M;V^{\star},F^{\star}).
\end{align*}
\end{defn}

As anticipated in the Introduction, this two-step procedure may either
reflect the need to secure approval from a max-min authority (a regulator,
a policy maker, or a supervisor) facing uncertainty over the set $\mathcal{A}=\mathcal{V}\times\mathcal{F}$
of admissible models, or the designer's response to her own uncertainty
over $\mathcal{A}$, whereby the designer first seeks robustness by
evaluating any mechanism under the worst-case scenario, and then uses
the conjectured model $(V^{\star},F^{\star})$ as a tool to select
from the shortlist of worst-case optimal mechanisms.

\section{Robustly optimal mechanisms}

\label{sec:robustly-quantity-mechanisms}

We begin by deriving the maximal guarantee of an arbitrary IC and
IR mechanism and show that the worst-case welfare need not occur under
the technology that places all probability mass at $\overline{\theta}$,
but always arises under the lowest possible value function $\underline{V}$.\footnote{The optimal mechanisms in this section are also the solution to an
alternative problem in which the designer and the buyer disagree about
the value of the good and their values are equal to $V^{\star}$ and
$\underline{V}$ respectively. The designer is Bayesian with subjective
belief $F^{\star}$ over $\Theta$, whereas the buyer faces uncertainty
over $\Theta$ and has max-min preferences.}

Let 
\[
{\rm q}_{\ell}\equiv\arg\max_{\text{{\rm q}\ensuremath{\in}[0,\ensuremath{\bar{{\rm q}}}]}}\left\{ \underline{V}({\rm q})-\overline{\theta}{\rm q}\right\} =\underline{D}(\overline{\theta})
\]
be the unique quantity that maximizes total surplus when $V=\underline{V}$
and $\theta=\overline{\theta}$; i.e., ${\rm q}_{\ell}$ is the efficient
quantity at the lowest demand and highest cost. Define 
\[
G^{*}\equiv\underline{V}({\rm q}_{\ell})-\overline{\theta}{\rm q}_{\ell},
\]
the total surplus when gross value is lowest, cost is highest, and
the buyer procures ${\rm q}_{\ell}$.

\begin{lemma}[welfare guarantee] \label{lemma-guarantee} For any
IC and IR mechanism $M\equiv(q,u)\in\mathcal{M}$, 
\begin{align}
G(M) & =\inf_{\theta\in\Theta}\left\{ \underline{V}(q(\theta))-\theta q(\theta)-u(\theta)\right\} ~~~~\textrm{and}~~~~G(M)\le G^{*}.\label{eq:wg1}
\end{align}
\end{lemma}

The first part of Lemma \ref{lemma-guarantee} highlights that Nature
can reduce the buyer's welfare more effectively by selecting a cost
$\theta<\overline{\theta}$. Since $q$ is weakly decreasing, the
buyer procures more from lower-cost types; thus, when Nature selects
an inverse demand below $P^{\star}$, the welfare loss from over-procurement
can be larger at low $\theta$.

The second part states that the welfare guarantee of any IC and IR
mechanism cannot exceed the total surplus obtained from procuring
the efficient output ${\rm q}_{\ell}$ when demand is lowest and cost
is highest. This follows because Nature can always choose the lowest
demand and the Dirac distribution that puts probability mass one at
$\tub$, in which case the buyer obtains the maximal welfare by procuring
the quantity ${\rm q}_{\ell}$.

The next lemma shows that the upper bound on the maximal welfare guarantee
is tight and fully characterizes the shortlist $\mathcal{M}^{{\rm SL}}$.

\begin{lemma}[shortlist characterization] \label{lemma-SL-characterization}Take
any IC and IR mechanism $M\equiv(q,u)\in\mathcal{M}$. Then, $M\in\mathcal{M}^{{\rm SL}}$
if and only if (a) $u(\overline{\theta})=0$, and (b), for all $\theta\in\Theta$,
\begin{align}
\underline{V}(q(\theta))-\theta q(\theta)-\int\limits^{\overline{\theta}}_{\theta}q(y){\rm d}y\ge G^{*}.\label{SL-constraint}
\end{align}
\end{lemma}

Worst-case optimality therefore imposes two additional constraints
beyond IC and IR. First, the highest-cost type $\overline{\theta}$
must earn zero rent: otherwise, Nature can select $\underline{V}$
and the Dirac distribution that puts probability mass one at $\tub$,
reducing welfare strictly below $G^{*}$ regardless of $q(\overline{\theta})$.
Consequently, for any $M\in\mathcal{M}^{{\rm SL}}$, the rent schedule
$u$ is pinned down by the output schedule $q$. Second, ex-post welfare
under the lowest possible demand function and zero rent for $\overline{\theta}$
must be weakly above the maximal guarantee $G^{*}$, \emph{for any
cost} $\theta$. Sufficiency follows from Lemma \ref{lemma-guarantee};
necessity is shown by constructing a simple constant mechanism that
attains $G^{*}$, namely one that procures ${\rm q}_{\ell}$ from
all types. This constant mechanism, however, is only one element of
a continuum in the shortlist. To see this, it is useful to reformulate
the robustness constraints in terms of majorization inequalities:

\noindent\begin{lemma}[majorizations] \label{lemma:majorizations}
Take any weakly decreasing function $q:\Theta\rightarrow\mathbb{R}_{+}$.
(1) For every $\theta\in\Theta$, constraint (\ref{SL-constraint})
is equivalent to 
\begin{align}
\int\limits^{\overline{\theta}}_{\theta}q(y){\rm d}y\leq\int\limits^{\overline{\theta}}_{\theta}\underline{D}(y){\rm d}y-\underset{\mathrm{\underline{DWL}}(\theta,q(\theta))\geq0}{\underbrace{\int\limits^{\underline{P}(q(\theta))}_{\theta}\left(\underline{D}(y)-q(\theta)\right){\rm d}y}},\label{eq:strong-major}
\end{align}
where $\mathrm{\underline{DWL}}(\theta,q(\theta))$ is the deadweight
loss incurred under $\underline{V}$ when output $q(\theta)$ is procured
instead of $\underline{D}(\theta)$. (2) The following statements
are equivalent: (a) the inequality in (\ref{eq:strong-major}) holds
for all $\theta\in\Theta$; (b) the inequality in (\ref{eq:strong-major})
holds for $\theta\in\{\underline{\theta},\overline{\theta}\}$ and,
for all $\theta\in(\underline{\theta},\overline{\theta})$, 
\begin{align}
\int\limits^{\overline{\theta}}_{\theta}q(y){\rm d}y\leq\int\limits^{\overline{\theta}}_{\theta}\underline{D}(y){\rm d}y.\label{eq:weak-major}
\end{align}
\end{lemma}

Because $\mathrm{\underline{DWL}}(\theta,q(\theta))\geq0$, the weaker
majorization in (\ref{eq:weak-major}) is implied by the stronger
one in (\ref{eq:strong-major}). However, when (\ref{eq:strong-major})
holds for the lowest and the highest type, and, in addition, the weaker
constraints are satisfied for all the remaining types, all the robustness
constraints are met. The intuition is the following. Because rents
$u(\theta)$ are decreasing in type, when the robustness constraint
is met at the extremes of the type distribution and the quantity procured
from intermediate types is small relative to $\underline{D}$ (in
the sense of (\ref{eq:strong-major})), welfare under the lowest demand
$\underline{V}(q(\theta))-\theta q(\theta)-\int\limits^{\overline{\theta}}_{\theta}q(y){\rm d}y$
exceeds the guarantee $G^{*}$ for all types, even if it is not necessarily
monotone in $\theta$.

It is easy to see that the schedule $q=\underline{D}$ satisfies the
robustness constraint in (\ref{eq:strong-major}). The following observation
is also useful:

\begin{obs}[quantity bound] \label{obs:qell} If $(q,u)\in\mathcal{M}^{{\rm SL}}$,
then $q(\theta)\ge{\rm q}_{\ell}$ for all $\theta\in\Theta$, with
equality at $\overline{\theta}$. \end{obs}

Clearly, for the robustness constraint (\ref{SL-constraint}) to hold
at $\overline{\theta}$, it must be that $q(\overline{\theta})={\rm q}_{\ell}$.
The observation then follows from this property along with the fact
that $q$ must be weakly decreasing.

Thus, any mechanism $(q,u)\in\mathcal{M}^{{\rm SL}}$ procures no
less than ${\rm q}_{\ell}$ from any type. For any weakly decreasing
schedule $q$ that procures $q(\theta)\in[{\rm {q}_{\ell}},\underline{D}(\theta)]$
from each type $\theta$, the robustness constraint in (\ref{eq:strong-major})
reduces to\footnote{To see this, note that the RHS in (\ref{eq:strong-major}) reduces
to $\int^{\underline{P}(q(\theta))}_{\theta}q(\theta){\rm d}y$ over
the interval $(\theta,\underline{P}(q(\theta)))$ and to $\int^{\tub}_{\underline{P}(q(\theta))}\underline{D}(y){\rm d}y$
over the interval $(\underline{P}(q(\theta)),\tub)$.} 
\[
\int\limits^{\overline{\theta}}_{\theta}q(y){\rm d}y\le\int\limits^{\overline{\theta}}_{\theta}\min\{\underline{D}(y),q(\theta)\}{\rm d}y,
\]
which is always satisfied, and, consequently, such a $q$ corresponds
to a mechanism in the shortlist. Moreover, $\mathcal{M}^{{\rm SL}}$
also includes mechanisms that procure more than the efficient output
$\underline{D}(\theta)$ for some types---a feature that plays a
key role in the analysis below.

Using Lemma \ref{lemma-SL-characterization} and the standard ``virtual
surplus'' representation of total welfare under the conjectured model
$(V^{\star},F^{\star})$ (see \citet{BM82}), the designer's problem
reduces to\footnote{Existence of a solution to the program in (\ref{opt:ropt}) follows
from the extreme-value theorem after establishing that the set of
quantity schedules $q$ satisfying the constraints is sequentially
compact and the objective function in (\ref{opt:ropt}) is sequentially
continuous under the point-wise convergence topology. See Section
\ref{sec:existence} in \nameref{Sec:additional_supplement}.} 
\begin{align}
 & ~~~~\max_{q}\int\limits^{\overline{\theta}}_{\underline{\theta}}\Big[V^{\star}(q(\theta))-z^{\star}(\theta)q(\theta)\Big]F^{\star}(\text{d}\theta)~~~~\tag{\textbf{ROPT}}\label{opt:ropt}\\
\textrm{} & \textrm{\textrm{subject to~\ensuremath{q}~} weakly decreasing,}~\mathrm{\textrm{and}}\nonumber \\
 & \underline{V}(q(\theta))-\theta q(\theta)-\int\limits^{\overline{\theta}}_{\theta}q(y){\rm d}y\ge G^{*}\qquad\forall~\theta\in\Theta.\label{eq:robustness}
\end{align}
Hereafter, we will often use Lemma \ref{lemma:majorizations} to express
the robustness constraints in (\ref{eq:robustness}) as majorization
constraints. Relative to \citet{BM82}, this program incorporates
the additional robustness constraint requiring that, for each $\theta\in\Theta$,
ex-post welfare under the lowest possible demand exceeds the guarantee
$G^{*}$. This constraint is not standard because it combines the
quantity procured from type $\theta$ with the integral of the quantity
procured from all higher types.

We denote a quantity schedule solving (\ref{opt:ropt}) as $\qopt$
and let $M^{{\rm OPT}}\equiv(\qopt,\uopt)$ be the mechanism corresponding
to it, with $\uopt$ given by $\uopt(\theta)=\int\limits^{\overline{\theta}}_{\theta}\qopt(y){\rm d}y$
for all $\theta$.

\subsection{Baron-Myerson-with-quantity-floor}

\noindent We start with a natural candidate solution of (\ref{opt:ropt})
and provide necessary and sufficient conditions under which it is
robustly optimal. Let $\qbm$ be the \emph{Baron-Myerson} quantity
schedule, defined, for all $\theta$, by 
\[
\qbm(\theta)\equiv\arg\max_{{\rm q}\in[0,\bar{{\rm q}}]}\Big[V^{\star}({\rm q})-z^{\star}(\theta){\rm q}\Big]=D^{\star}(z^{\star}(\theta)).
\]

\begin{defn} The \textbf{Baron-Myerson-with-quantity-floor} mechanism
$M^{\star}\equiv(q^{\star},u^{\star})$ is such that, for all $\theta$,\footnote{\noindent Observe that $\qbm(\underline{\theta})=D^{\star}(\underline{\theta})\ge\underline{D}(\underline{\theta})>\underline{D}(\overline{\theta})={\rm q}_{\ell}$.
On the other hand, $\qbm(\overline{\theta})=D^{\star}(z^{\star}(\overline{\theta}))$
can be smaller than $\underline{D}(\overline{\theta})={\rm q}_{\ell}$
given that $z^{\star}(\overline{\theta})>\overline{\theta}$.} 
\begin{equation}
q^{\star}(\theta)=\max\{q^{{\rm {BM}}}(\theta),{\rm q}_{\ell}\}\label{eq:BM-floor}
\end{equation}
and $u^{\star}(\theta)=\intop^{\overline{\theta}}_{\theta}q^{\star}(y){\rm d}y$.
\end{defn}

\begin{proposition}[optimality of Baron-Myerson-with-quantity-floor]
\label{prop:BM-floor} The mechanism $M^{\star}\equiv(q^{\star},u^{\star})$
is robustly optimal if and only if 
\begin{align}
\int\limits^{\overline{\theta}}_{\underline{\theta}}q^{\star}(y){\rm d}y & \leq\int\limits^{\overline{\theta}}_{\underline{\theta}}\underline{D}(y){\rm d}y-\int\limits^{\underline{P}(q^{\star}(\underline{\theta}))}_{\underline{\theta}}\Big[\underline{D}(y)-q^{\star}(\underline{\theta})\Big]{\rm d}y,\label{eq:robustness-bottom}\\
\int\limits^{\overline{\theta}}_{\theta}q^{\star}(y){\rm d}y & \le\int\limits^{\overline{\theta}}_{\theta}\underline{D}(y){\rm d}y~\qquad~\forall~\theta\in\Theta.\label{eq:weak-majorization}
\end{align}
\end{proposition}

The two conditions in the proposition are equivalent to the requirement
that the quantity schedule $q^{\star}$ satisfies all the robustness
constraints in (\ref{SL-constraint}). Because $q^{\star}$ maximizes
virtual surplus over all weakly decreasing schedules $q$ satisfying
$q(\theta)\geq{\rm q}_{\ell}$ for all $\theta$, Baron-Myerson-with-quantity-floor
is robustly optimal when, and only when, these conditions hold.

Note that the majorization constraints in (\ref{eq:weak-majorization})
are easier to verify than the constraints in (\ref{SL-constraint}).
In practice, it suffices to verify them at those points in which the
schedule $q^{\star}$ crosses $\underline{D}$ from below.\footnote{In fact, because $q^{\star}$ is weakly decreasing, Lemma \ref{Lem-Mono}
in the \nameref{Sec:appendix} implies that the function $\underline{W}(\theta,q^{\star})=\underline{V}(q^{\star}(\theta))-\theta q^{\star}(\theta)-\int\limits^{\overline{\theta}}_{\theta}q^{\star}(y){\rm d}y$
is weakly decreasing (alternatively, weakly increasing) over an interval
of types $I\subset\Theta$ if $0<q^{\star}(\theta)\le\underline{D}(\theta)$
(alternatively, $q^{\star}(\theta)>\underline{D}(\theta)$) for all
$\theta\in I$. Because, for all $\theta\in\Theta$, $q^{\star}(\theta)>0$,
when Condition (\ref{eq:robustness-bottom}) holds, to verify that
Condition (\ref{eq:weak-majorization}) also holds, it suffices to
check that it holds at points in which the schedule $q^{\star}$ crosses
$\underline{D}$ from below. If there are no such points, and (\ref{eq:robustness-bottom})
holds, then all the robustness constraints in (\ref{SL-constraint})
hold and $M^{\star}$ is robustly optimal.} For example, in Figure \ref{fig:BM-floor}, it suffices to verify
that $q^{\star}$, in addition to satisfying Condition (\ref{eq:robustness-bottom}),
satisfies Condition (\ref{eq:weak-majorization}) at $\theta_{1}$.

\begin{figure}
\centering \includegraphics[width=0.6\linewidth]{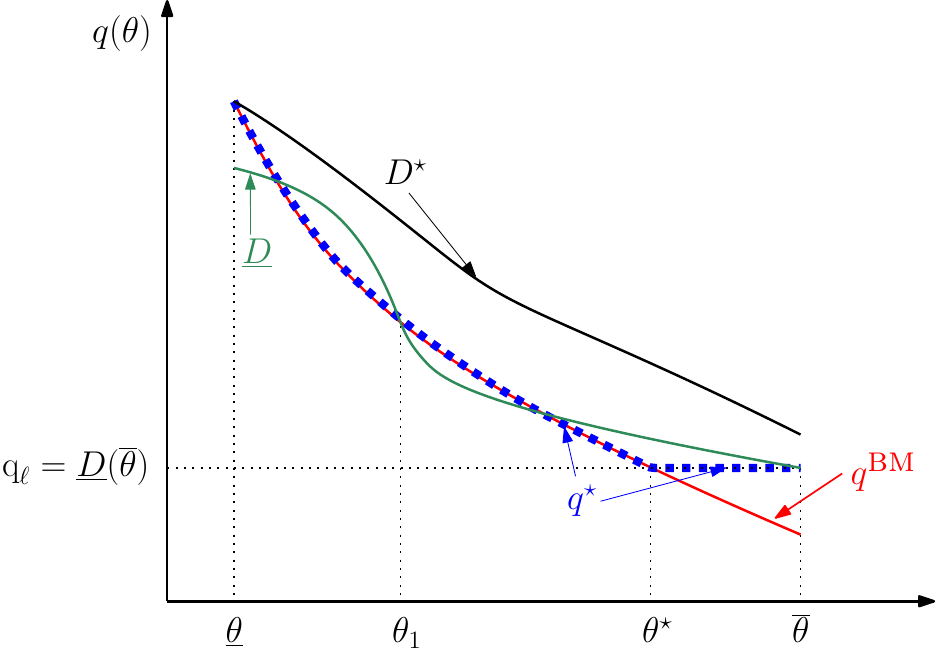}
\caption{Baron-Myerson-with-quantity-floor.}
\label{fig:BM-floor} 
\end{figure}

\begin{corollary}[uniqueness of output schedule] \label{cor:qstaropt}
If $M^{\star}=(q^{\star},u^{\star})$ is robustly optimal, then any
robustly optimal mechanism $M^{{\rm OPT}}=(\qopt,\uopt)$ satisfies
$\qopt(\theta)=q^{\star}(\theta)$ for all $\theta>\underline{\theta}$.
\end{corollary}

\noindent \emph{Proof}: If $q^{\star}$ solves (\ref{opt:ropt}), then it also
solves the relaxed program where the objective function is the same
as in (\ref{opt:ropt}) and the constraints in (\ref{opt:ropt}) are
replaced by $q(\theta)\ge{\rm q}_{\ell}$ for all $\theta\in\Theta$
(see Observation \ref{obs:qell}). Any solution to this relaxed program
must coincide with $q^{\star}$ at almost all $\theta$. Because $q^{\star}$
is continuous, it must in fact coincide at all $\theta>\underline{\theta}$.
$\hfill Q.E.D.$

Proposition \ref{prop:BM-floor} can also be used to obtain easy-to-verify
sufficient conditions for the Baron-Myerson-with-quantity-floor mechanism
to be robustly optimal. For every $\theta\in\Theta$, let 
\begin{align}
d(\theta) & \equiv P^{\star}(\underline{D}(\theta))-\theta\label{eq:def-d}
\end{align}
denote the distance between the two \emph{inverse demand} functions
$P^{\star}$ and $\underline{P}$ at ${\rm q}=\underline{D}(\theta)$.\footnote{Note that $d(\theta)=P^{\star}(\underline{D}(\theta))-\underline{P}(\underline{D}(\theta))$.
Equivalently, $d(\theta)$ is the unique real solution to $D^{\star}(\theta+x)=\underline{D}(\theta)$.}

\begin{corollary}[large downward distortions] \label{cor:rotate0}
Suppose $F^{\star}(\theta)/f^{\star}(\theta)\ge d(\theta)$ for all
$\theta\in\Theta$. Then, $M^{\star}=(q^{\star},u^{\star})$ is robustly
optimal. \end{corollary}

\noindent \emph{Proof}: By definition, $d(\theta)\ge0$ for all $\theta\in\Theta$.
Note that $F^{\star}(\underline{\theta})/f^{\star}(\underline{\theta})=0$.
The condition in the corollary thus implies that $d(\underline{\theta})=0$.
Hence, $P^{\star}(\underline{D}(\underline{\theta}))=\underline{\theta}$,
and therefore, $\underline{D}(\underline{\theta})=D^{\star}(\underline{\theta})$.
Because $q^{\star}(\underline{\theta})=\qbm(\underline{\theta})=D^{\star}(\underline{\theta})=\underline{D}(\underline{\theta})$,
we have that $\underline{P}(q^{\star}(\underline{\theta}))=\underline{\theta}$,
and the two conditions in Proposition \ref{prop:BM-floor} reduce
to the requirement that $\int\limits^{\overline{\theta}}_{\theta}[\underline{D}(y)-q^{\star}(y)]{\rm d}y\geq0$
for all $\theta\in\Theta$. Next, observe that, for all $\theta\in\Theta$,
\[
\qbm(\theta)=D^{\star}\Big(\theta+\frac{F^{\star}(\theta)}{f^{\star}(\theta)}\Big)\le D^{\star}\Big(\theta+d(\theta)\Big)=\underline{D}(\theta),
\]
where the inequality follows because $D^{\star}$ is decreasing and
$F^{\star}(\theta)/f^{\star}(\theta)\ge d(\theta)$, whereas the last
equality follows from the definition of $d(\theta)$. Thus, for every
$\theta\in\Theta$, we have that $q^{\star}(\theta)=\max\{\qbm(\theta),{\rm q}_{\ell}\}=\max\{\qbm(\theta),\underline{D}(\overline{\theta})\}\le\underline{D}(\theta)$.
Hence, $q^{\star}$ satisfies the conditions in Proposition \ref{prop:BM-floor}.
$\hfill Q.E.D.$

The following is also an immediate implication of the above results:

\begin{corollary}[no demand uncertainty] \label{cor:onlycost}
If $V^{\star}=\underline{V}$ (which is the case when there is no
demand uncertainty, i.e., when $\mathcal{V}=\{V^{\star}\}$), then
$M^{\star}=(q^{\star},u^{\star})$ is robustly optimal. In this case,
the optimal mechanism features efficiency both at ``the top'' $(\underline{\theta})$
and ``the bottom'' $(\overline{\theta})$ of the type distribution.
\end{corollary}

\subsection{Beyond Baron-Myerson-with-quantity-floor}

\label{subsec:Beyond-Baron-Myerson-with-quanti}

We now characterize the qualitative properties of robustly optimal
mechanisms when they differ from the Baron-Myerson-with-quantity-floor
(the precise characterization is in \nameref{Sec:OS}).

Let $\theta^{\star}$ be the unique solution to $\qbm(\theta^{\star})={\rm q}_{\ell}$
when such a solution exists (which is the case when, and only when,
$\qbm(\overline{\theta})<{\rm q}_{\ell}$); else, let $\theta^{\star}\equiv\overline{\theta}$.
Next, let $\underline{W}(\cdot,q^{\star})$ be the function defined,
for all $\theta\in\Theta$, by 
\begin{equation}
\underline{W}(\theta,q^{\star})\equiv\underline{V}(q^{\star}(\theta))-\theta q^{\star}(\theta)-\int\limits^{\overline{\theta}}_{\theta}q^{\star}(y)dy\label{eq:W-function-q-star}
\end{equation}
and let $\theta^{m}\equiv\max\{\theta:\theta\in\arg\min_{y\in\Theta}\underline{W}(y,q^{\star})\}.$
Under regularity, $\theta^{m}$ is well defined and is the largest
cost at which the function $\underline{W}(\cdot,q^{\star})$ attains
a minimum.\footnote{Under regularity, $q^{\star}$ is continuous. Because $\Theta$ is
compact and $\underline{W}(\cdot,q^{\star})$ is continuous over $\Theta$,
the set $\{\theta:\underline{W}(\theta,q^{\star})\le\underline{W}(\theta',q^{\star})~\forall~\theta'\in\Theta\}$
is non-empty and compact.} Type $\theta^{m}$ plays an important role in our analysis. Indeed,
Lemma \ref{lemm-theta_m-star} in the \nameref{Sec:appendix} shows
that the Baron-Myerson-with-quantity-floor mechanism is robustly
optimal if and only if $\theta^{m}=\overline{\theta}$. Proposition
\ref{prop:robust-quantity-mechanism-general} below uses $\theta^{m}$
and $\theta^{\star}$ to characterize the main qualitative properties
of robustly optimal mechanisms. Section \ref{subsec:Stronger-Prop-2}
in \nameref{Sec:OS} provides a sharper characterization.

\begin{proposition}[robust optimality: general case]\label{prop:robust-quantity-mechanism-general}
Suppose $M^{\star}$ is not robustly optimal. (1) Then, $\theta^{m}<\theta^{\star}$
and 
\begin{align*}
\qopt(\theta) & =\qbm(\theta)~\qquad~\forall~\theta\in(\tlb,\theta^{m}),\\
\qopt(\theta) & \le\qbm(\theta)~\qquad~\forall~\theta\in(\theta^{m},\theta^{\star}),\\
\qopt(\theta) & ={\rm q}_{\ell}~\qquad~\forall~\theta\in[\theta^{\star},\tub].
\end{align*}
(2) Furthermore, if $\theta^{\star}<\tub$, then $\qopt(\theta)<\qbm(\theta)$
for a subset of $(\theta^{m},\theta^{\star})$ of positive Lebesgue
measure. If, instead, $\theta^{\star}=\tub$, there exists $\theta^{\star\star}\in[\theta^{m},\theta^{\star}]$
such that $\qopt(\theta)<\qbm(\theta)$ for a subset of $(\theta^{m},\theta^{\star\star})$
of positive Lebesgue measure and $\qopt(\theta)=\underline{D}(\theta)$
for all $\theta\in[\theta^{\star\star},\tub]$. \end{proposition}



Figure \ref{fig:alpha0robust} illustrates the structure of robustly
optimal quantity schedules when they differ from the quantity schedule
in Baron-Myerson-with-quantity-floor and $\theta^{\star}<\tub$.
In this case, robustness entails an increase in procurement from high-cost
sellers and a reduction in procurement from intermediate-cost sellers,
relative to the Bayesian optimum. The increase (from $\qbm(\theta)$
to ${\rm q}_{\ell}$) for high-cost sellers prevents welfare losses
that arise if Nature selects high-cost sellers with a higher probability
than what is conjectured by the designer. Recall that the downward
distortions in the Baron-Myerson schedule $\qbm$ serve to limit
the rents of low-cost types. When high-cost types are more likely,
these distortions lose their rationale. In contrast, the reduction
in procurement (from $\qbm(\theta)$ to $\qopt(\theta)<\qbm(\theta)$)
for intermediate-cost sellers limits welfare losses when actual demand
is lower than conjectured. Section \ref{subsec:Stronger-Prop-2} in
\nameref{Sec:OS} extends the result in part (2) by showing that the
strict inequality $\qopt(\theta)<\qbm(\theta)$ applies to all $\theta\in(\theta^{m},\min\{\theta^{\star},\theta^{\star\star}\})$.
The proof uses Lagrangian methods for infinite-dimensional constraints
along with ironing techniques to establish that, for $\theta\in(\theta^{m},\min\{\theta^{\star},\theta^{\star\star}\})$,
$\qopt(\theta)=D^{\star}(\bar{z}^{M}(\theta))$, where $\bar{z}^{M}$
is an ironed virtual cost that accounts for robustness.

\begin{figure}
\centering \includegraphics[width=0.5\linewidth]{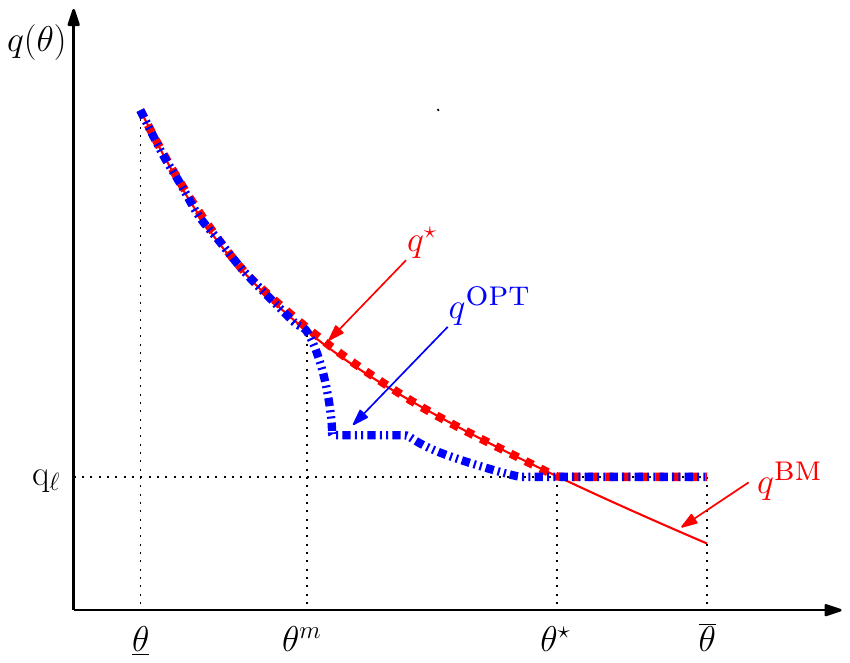}
\caption{Illustration of Proposition \ref{prop:robust-quantity-mechanism-general}.}
\label{fig:alpha0robust} 
\end{figure}

\section{Downstream markets and monopoly regulation}

\label{sec:price_vs_quantity}

Now suppose there exists a downstream market for the good supplied
by the seller, and interpret $D$ as the demand curve in such a market
(with inverse $P$ and associated value $V$). Recall that $D$ is
continuous over $\mathbb{R}_{++}$, decreasing on $(0,P(0)]$, and
satisfies $D({\rm p})=0$ for all ${\rm p}>P(0)$. The mechanisms
in the previous section can then be interpreted as \emph{quantity
regulation}, whereby the buyer (acting as a regulator maximizing consumer
surplus) directly controls the quantity supplied to consumers.\footnote{If the revenue $P(q(\theta))q(\theta)$ is accrued to the seller,
the mechanism $(q,u)$ is implemented through a transfer schedule
$\tilde{t}(\theta,P)=t(\theta)-P(q(\theta))q(\theta)$, where $t$
is the total transfer schedule corresponding to the mechanism $(q,u)$,
and where $P(q(\theta))$ is the observed market price. Such a transfer
schedule fully insulates the seller from any uncertainty about the
demand $D$ (alternatively, the price function $P$).}

The regulator can alternatively influence market outcomes by regulating
prices rather than quantities. Price regulation has the advantage
of allowing the quantity to adjust to the realized demand. However,
as shown in Subsection \ref{subsec:Price-vs-quantity} below, it also
entails disadvantages that make its overall desirability vis-\`a-vis
quantity regulation unclear.\footnote{Under Bayesian analysis (i.e., in the absence of worst-case approval
constraints), price regulation always dominates.} Below, we first characterize robust price regulation, and then identify
conditions under which each form of regulation dominates the other.
Consistently with the analysis in the previous section, we assume
that the regulator delegates to a Bayesian designer the choice of
the price mechanism but requires that the welfare guarantee (over
the admissible set $\mathcal{D}\times\mathcal{F}$) be the highest.

\subsection{Price regulation}

\label{sec:price mechanisms}

Price regulation consists of a price function $p$ and a transfer
function $t$. The price function $p:\Theta\rightarrow\mathbb{R}$
specifies the price the seller must set for each reported cost. The
seller must then supply any quantity demanded by consumers at that
price. Because demand is uncertain, price mechanisms expose the seller
to uncertainty in profits. To ensure that the seller participates
and reports truthfully regardless of its beliefs or attitude toward
uncertainty, the designer conditions the total transfer $t$ to the
seller on the realized demand $D\in\mathcal{D}$, which we assume
becomes known ex post.\footnote{In some cases, (e.g., when the set $\mathcal{P}$ of inverse demand
functions is fully ordered so that for all $P,P'\in\mathcal{P}$,
we have that $P(\mathrm{q})\neq P'(\mathrm{q})$ for all $\mathrm{q}\in[0,\bar{\mathrm{q}}]$)
the demand $D$ can be learned by observing the quantity $D(p(\theta))$
traded at the posted price $p(\theta)$. The assumption that $D$
is learned ex-post gives price regulation the ``best shot'' at beating
quantity regulation. As we show in Subsection \ref{subsec:Price-vs-quantity},
even under such a favorable scenario, price regulation may be dominated
by quantity regulation.} We maintain that, once the quantity $D(p(\theta))$ demanded at price
$p(\theta)$ is observed by the designer and the seller (which can
happen either concurrently with or before learning the full demand
$D$), it is too late to make adjustments to the price or the output
produced. As explained in the previous section, in the absence of
these natural frictions, demand uncertainty is inconsequential.

\begin{defn} A \textbf{price regulation} $\widetilde{M}=(p,t)$ is
a pair of mappings $p:\Theta\rightarrow\mathbb{R}_{+}~\mathrm{and}~~t:\Theta\times\mathcal{D}\rightarrow\mathbb{R}_{+}$
where $p(\theta)$ is the price the regulator asks the seller to post
and $t(\theta,D)$ is the total transfer to the seller when the reported
cost is $\theta$ and the realized demand is $D$.\footnote{Again, if the revenue $p(\theta)D(p(\theta))$ is accrued to the seller,
the regulator commits to transfer an additional amount $t^{\dag}(\theta,D)=t(\theta,D)-p(\theta)D(p(\theta))$
to the seller. This transfer is funded with consumers' money. The
total monetary cost to the consumers is thus $t(\theta,D)$. } \end{defn}

The price regulation $\widetilde{M}=(p,t)$ is \emph{ex-post} \textsl{incentive
compatible }(EPIC) if, for all $\theta,\theta'\in\Theta$ and $D\in\mathcal{D}$,
$t(\theta,D)-\theta D(p(\theta))\ge t(\theta',D)-\theta D(p(\theta'))$.
It is \textsl{ex-post individually rational }(EPIR) if, for all $\theta\in\Theta$
and $D\in\mathcal{D}$, $\tilde{u}(\theta,D)\equiv t(\theta,D)-\theta D(p(\theta))\ge0$.

Let $\mathcal{\widetilde{M}}$ be the set of all EPIC and EPIR price
regulations. For any $\widetilde{M}\in\mathcal{\widetilde{M}}$, demand
$D\in\mathcal{D}$, and technology $F\in\mathcal{F}$, welfare (consumer
surplus) is given by 
\[
\widetilde{W}(\widetilde{M};D,F)\equiv\int\limits^{\overline{\theta}}_{\underline{\theta}}\widetilde{w}(\theta,\widetilde{M};D)F(\mathrm{d}\theta),
\]
where, for all $\theta\in\Theta$, $\widetilde{M}\in\mathcal{\widetilde{M}}$,
and $D\in\mathcal{D}$, 
\[
\widetilde{w}(\theta,\widetilde{M};D)\equiv\intop^{D(p(\theta))}_{0}D^{-1}(y){\rm d}y-\theta D(p(\theta))-\tilde{u}(\theta,D).
\]
The \textbf{welfare guarantee} of any price regulation $\widetilde{M}\in\mathcal{\widetilde{M}}$
is given by 
\[
G(\widetilde{M})\equiv\inf_{D\in\mathcal{D},F\in\mathcal{F}}\widetilde{W}(\widetilde{M};D,F).
\]
The \textbf{shortlist} of approved price regulations is given by 
\[
\mathcal{\widetilde{M}}^{{\rm SL}}\equiv\arg\max_{\widetilde{M}\in\mathcal{\widetilde{M}}}G(\widetilde{M}).
\]
Recall that the maximal welfare guarantee for quantity regulations
is $G^{*}\equiv\underline{V}({\rm q}_{\ell})-\overline{\theta}{\rm q}_{\ell}$,
as shown in Lemma \ref{lemma-guarantee}, where ${\rm q}_{\ell}\equiv\underline{D}(\overline{\theta})$
is the efficient quantity when $\theta=\overline{\theta}$ and $D=\underline{D}$
(equivalently, when $P=\underline{P}$ and $V=\underline{V}$). This
same guarantee applies to price regulations:

\begin{lemma}[shortlist of price regulations] \label{lemma-SL-price-mechanisms}
If $\widetilde{M}\in\widetilde{\mathcal{M}}^{{\rm SL}}$, then $G(\widetilde{M})=G^{*}$.
Moreover, $\widetilde{M}\equiv(p,t)\in\widetilde{\mathcal{M}}^{{\rm SL}}$
if and only if (a) $p$ is weakly increasing, (b) for all $\theta\in\Theta$
and $D\in\mathcal{D}$, 
\begin{equation}
\tilde{u}(\theta,D)=\tilde{u}(\overline{\theta},D)+\int\limits^{\overline{\theta}}_{\theta}D(p(y)){\rm d}y\label{eq:rent-price-mechanism}
\end{equation}
with $\tilde{u}(\overline{\theta},D)\geq0$ and $\tilde{u}(\overline{\theta},\underline{D})=0$,
(c) $p(\overline{\theta})=\overline{\theta}$, and (d) for all $\theta\in\Theta$
and $D\in\mathcal{D}$, $\widetilde{w}(\theta,\widetilde{M};D)\geq G^{*}$.
\end{lemma}

The following is an immediate implication of the previous lemma: 

\begin{corollary}[equivalence under worst-case optimality]\label{Cor-equivalence-p-and-q}The
maximal welfare guarantee over all price regulations equals that under
quantity regulations.\end{corollary}

The result follows from the fact that Nature can always select the
lowest possible demand $\underline{D}$ and highest cost $\overline{\theta}$.
The maximal welfare at $(\underline{D},\overline{\theta})$ is attained
when the monopolist supplies the efficient output ${\rm q}_{\ell}$.
Whether the regulator enforces ${\rm q}_{\ell}$ by fixing the price
at marginal cost or by directly imposing that the monopolist supply
that quantity is immaterial. However, as shown in Subsection \ref{subsec:Price-vs-quantity},
identical guarantees do not imply that the two mechanisms yield the
same welfare under the designer's conjectured model.

A simple way to attain $G^{*}$ is to set $p(\theta)=\overline{\theta}$
for all $\theta\in\Theta$ and $\tilde{u}(\overline{\theta},D)=0$
for all $D\in\mathcal{D}$. Yet many other price regulations yield
the same guarantee, allowing the designer to use her conjectured model
$(D^{\star},F^{\star})$ to select among them. Notably, as indicated
in the last lemma, any approved price regulation must set $p(\overline{\theta})=\overline{\theta}$
and leave zero rent when $\theta=\overline{\theta}$ and $D=\underline{D}$,
as only these choices achieve $G^{*}$ in the worst case. Unlike quantity
regulation, however, the price function does not uniquely determine
the transfers, because $\tilde{u}(\overline{\theta},D)$ may be positive
for $D\neq\underline{D}$.

\begin{defn}\label{def:BM-with-price-cap} A \textbf{Baron-Myerson-with-price-cap}
regulation is a pair $\widetilde{M}=(p,t)$ such that 
\begin{equation}
p(\theta)=\min\{z^{\star}(\theta),\overline{\theta}\}\label{p-optimal}
\end{equation}
for all $\theta\in\Theta$, and the induced rent schedule $\tilde{u}(\theta,D)\equiv t(\theta,D)-\theta D(p(\theta))$
satisfies $\tilde{u}(\theta,D)=\tilde{u}(\overline{\theta},D)+\intop^{\overline{\theta}}_{\theta}D(p(y))dy$
for all $\theta\in\Theta$ and $D\in\mathcal{D}$, with $\tilde{u}(\overline{\theta},\underline{D})=\tilde{u}(\overline{\theta},D^{\star})=0$
and, for all $D\in\mathcal{D}\setminus\{\underline{D},D^{\star}\}$,
\begin{align}
0\le\tilde{u}(\overline{\theta},D) & \le\int\limits^{D(\overline{\theta})}_{0}D^{-1}(y)dy-\overline{\theta}D(\overline{\theta})-G^{*}.\label{eq:eee1}
\end{align}
\end{defn}

All such regulations share the same price schedule (for $\theta>\underline{\theta}$)
but differ in their transfer schedules.

\begin{proposition}[optimality of Baron-Myerson-with-price-cap]\label{prop:robustly-optimal-price-mechanism}
Every Baron-Myerson-with-price-cap regulation is robustly optimal.
Moreover, in every robustly optimal price regulation $\widetilde{M}^{\textsc{OPT}}=(p^{\textsc{OPT}},t^{{\rm OPT}})$,
for any $\theta>\underline{\theta}$, $p^{\textsc{OPT}}(\theta)=\min\{z^{\star}(\theta),\overline{\theta}\}$,
and $\widetilde{u}^{\textsc{OPT}}(\overline{\theta},D^{\star})=0$.\end{proposition}

Under the conjectured model $(D^{\star},F^{\star})$ with gross value
function $V^{\star}$, the virtual surplus $V^{\star}(D^{\star}(\text{p}))-z^{\star}(\theta)D^{\star}(\text{p})$
is quasi-concave in $\text{p}$ and maximized at $\mathrm{p}=z^{\star}(\theta)$
for all $\theta\in\Theta$. The constants $\{{\tilde{u}}^{{\rm OPT}}(\overline{\theta},D)\}$
can then be chosen to satisfy all the constraints in Lemma \ref{lemma-SL-price-mechanisms}
(for instance, by choosing $\widetilde{u}(\tub,D)=0$ for all $D\in\mathcal{D}$).

\begin{corollary}[robustly optimal price schedule]\label{Corollary:optimal-price}In
the interval $(\tlb,\tub]$, the robustly optimal price schedule is
unique and is invariant to both the designer's conjectured demand
$D^{\star}$, and the admissible set $\mathcal{D}\times\mathcal{F}$
of demands and cost technologies the regulator considers plausible.
It sets a markup of $F^{\star}(\theta)/f^{\star}(\theta)$ at each
cost $\theta>\tlb$, capped at $(\overline{\theta}-\theta)$. \end{corollary}

Figure \ref{fig:rob_opt_price} illustrates the result. By committing
to rent payments contingent on realized demand and setting a cost-dependent
markup---based solely on the conjectured cost distribution $F^{\star}$---capped
at $(\overline{\theta}-\theta)$, the designer maximizes welfare (under
the conjectured model $(D^{\star},F^{\star})$) regardless of the
regulator's uncertainty over demand or technology.

\begin{figure}
\centering \includegraphics[width=0.75\linewidth]{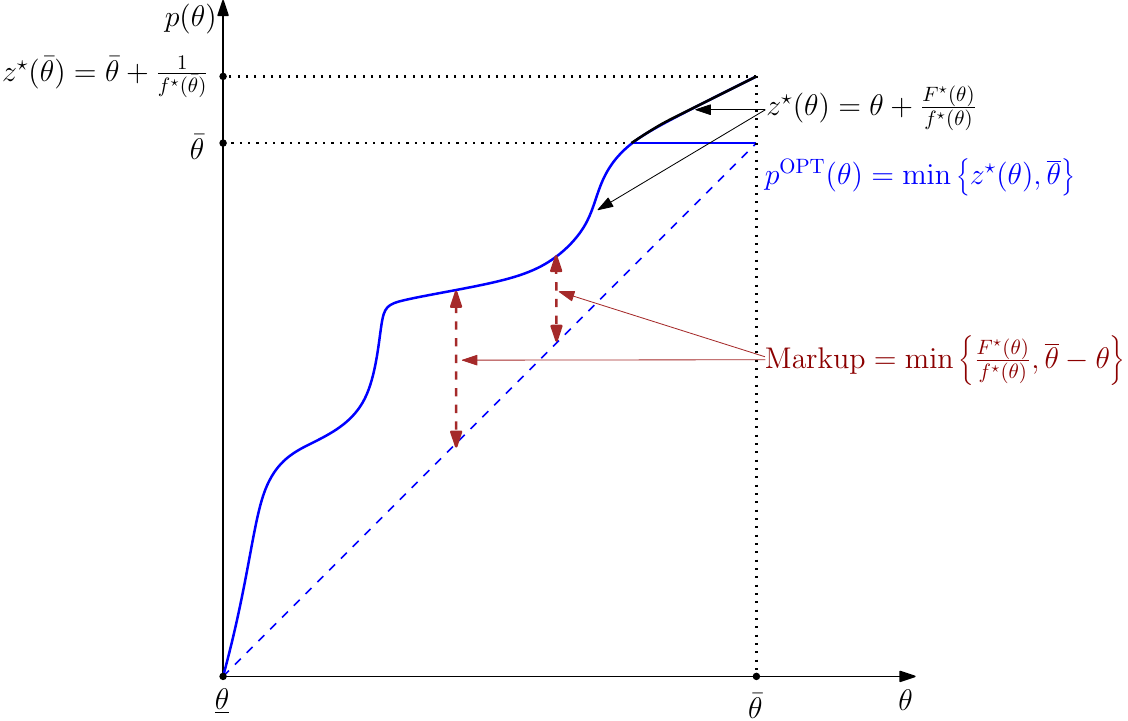}
\caption{Robustly optimal price schedule.}
\label{fig:rob_opt_price} 
\end{figure}

\subsection{Quantity vs price regulation}

\label{subsec:Price-vs-quantity}

Corollary \ref{Cor-equivalence-p-and-q} establishes that the maximal
welfare guarantee---the highest welfare achievable under the worst-case
scenario---is the same for both price and quantity regulation. However,
the maximal welfare attainable under the designer's conjectured model
$(D^{\star},F^{\star})$ over the shortlist of approved worst-case-optimal
regulations may differ between the two types of regulation. The following
definition formalizes what it means for one type of regulation to
dominate the other, from the perspective of the designer:

\begin{defn} Price regulation dominates quantity regulation if 
\begin{equation}
\widetilde{W}(\widetilde{M}^{\textsc{OPT}};D^{\star},F^{\star})\ge W(M^{\textsc{OPT}};V^{\star},F^{\star})\label{ranking-p-vs-q}
\end{equation}
(strictly if the inequality is strict). Conversely, quantity regulation
dominates price regulation if 
\begin{equation}
W(M^{\textsc{OPT}};V^{\star},F^{\star})\geq\widetilde{W}(\widetilde{M}^{\textsc{OPT}};D^{\star},F^{\star})\label{eq:quantity-wins}
\end{equation}
(strictly if the inequality is strict). Price and quantity regulation
are equivalent if (\ref{ranking-p-vs-q}) and (\ref{eq:quantity-wins})
jointly hold. \end{defn}

\begin{proposition}[quantity vs price regulation]\label{prop:comp_reg}
(1) If the Baron-Myerson-with-quantity-floor mechanism is robustly
optimal, quantity regulation dominates price regulation (strictly
if $D^{\star}(\overline{\theta})>\underline{D}(\overline{\theta})$).
(2) If $D^{\star}(\overline{\theta})=\underline{D}(\overline{\theta})$,
price regulation dominates quantity regulation. Furthermore, the domination
is strict if the Baron-Myerson-with-quantity-floor mechanism is not
robustly optimal. \end{proposition}

Panel A of Figure \ref{fig:comp_reg} illustrates the induced quantity
schedules $\qopt$ and $D^{\star}(\popt)$ under the designer's conjectured
demand $D^{\star}$ for the first part of Proposition \ref{prop:comp_reg},
when $D^{\star}(\overline{\theta})>\underline{D}(\overline{\theta})={\rm q}_{\ell}$.
Panel B illustrates the second part, when $D^{\star}(\overline{\theta})=\underline{D}(\overline{\theta})={\rm q}_{\ell}$.

\begin{figure}
\centering %
\begin{minipage}[b]{0.49\linewidth}%
\includegraphics[scale=0.53]{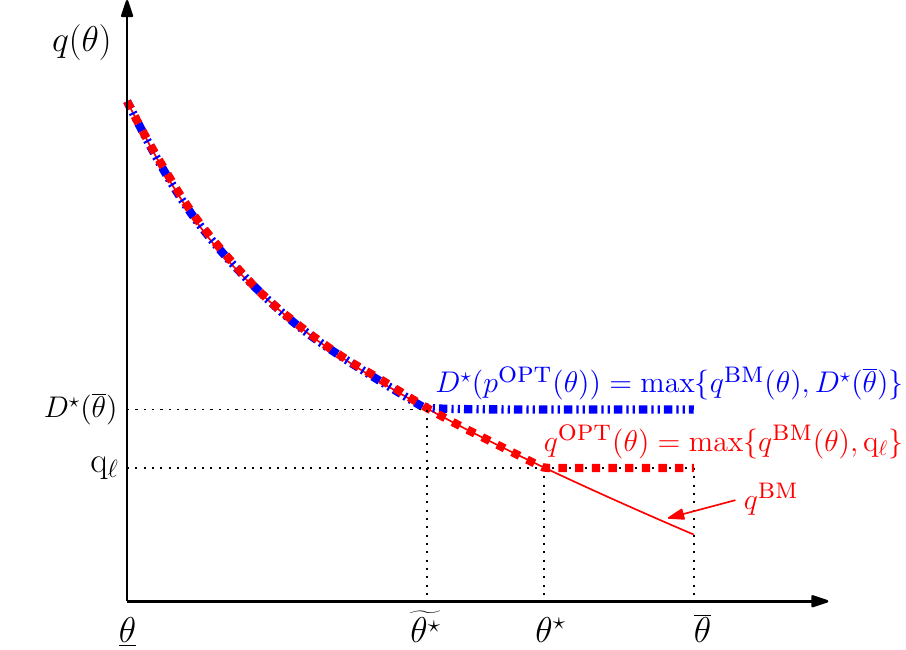} \caption*{\centering A: Part 1.}
\end{minipage}%
\begin{minipage}[b]{0.49\linewidth}%
\includegraphics[scale=0.54]{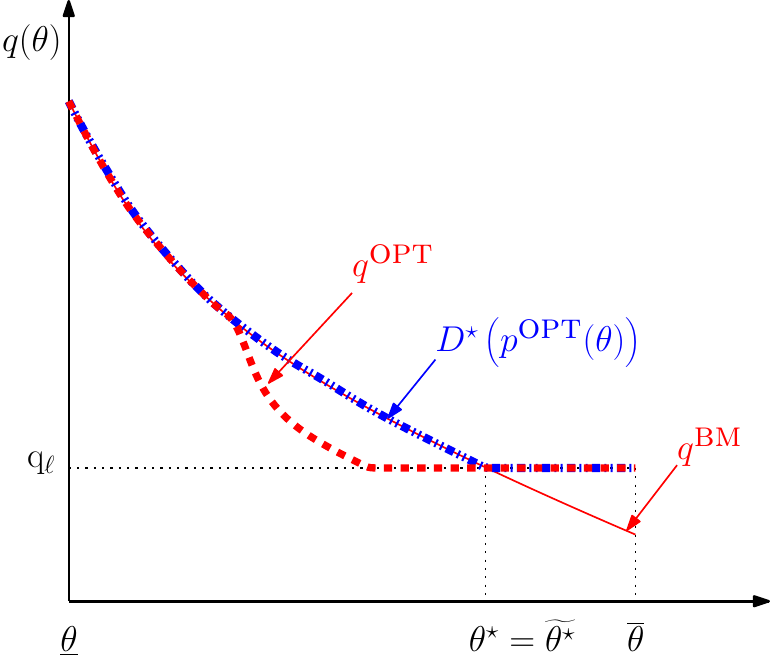} \caption*{\centering B: Part 2.}
\end{minipage}\caption{Graphical illustration of Proposition \ref{prop:comp_reg}.}
\label{fig:comp_reg} 
\end{figure}

The intuition for the first part is as follows. When the Baron-Myerson-with-quantity-floor
mechanism is robustly optimal, both regulations induce the same volume
of trade at low costs under the conjectured model. For high costs,
however, the output traded under the optimal price regulation can
exceed that under the optimal quantity regulation. This occurs because
the price is capped at $\overline{\theta}$, leading the seller to
supply $D^{\star}(\overline{\theta})$ under the conjectured demand,
which exceeds ${\rm {q}_{\ell}}$, the quantity supplied under the
optimal quantity regulation, when $D^{\star}(\overline{\theta})>\underline{D}(\overline{\theta})={\rm q}_{\ell}$.
This oversupply increases the distance from $D^{\star}(z^{\star}(\theta))$,
the Bayesian optimal quantity under the conjectured model $(D^{\star},F^{\star})$,
reducing welfare relative to quantity regulation (see Panel A of Figure
\ref{fig:comp_reg}).

For the second part, when the Baron-Myerson-with-quantity-floor mechanism
is not robustly optimal, quantity regulation requires downward adjustments
in output for intermediate costs to avoid excessive trading under
low realized demand (see Proposition \ref{prop:robust-quantity-mechanism-general}).
These adjustments reduce welfare under the conjectured model. Price
regulation, by contrast, automatically protects the regulator by adjusting
output to the realized demand, sparing the designer from the need
to reduce trade for intermediate-cost sellers, thus avoiding the welfare
reduction. Provided the price cap ${\rm p}=\overline{\theta}$ does
not induce over-procurement for high costs---which is the case when
$D^{\star}(\overline{\theta})-\underline{D}(\overline{\theta})$ is
small---price regulation strictly dominates.

The following is an implication of the previous results and can be
seen directly by combining Corollary \ref{cor:onlycost} and Proposition
\ref{prop:robustly-optimal-price-mechanism}.

\begin{corollary}[equivalence in the absence of downward demand
uncertainty]\label{Corollary:equivalence-price-quantity} If the
regulator faces no downward demand uncertainty (i.e., $\underline{D}=D^{\star}$),
price and quantity regulations are equivalent. \end{corollary}

In the absence of downward demand uncertainty, it does not matter
whether the designer induces the desired output (under the conjectured
model $(D^{\star},F^{\star})$) by fixing the price or specifying
the quantity.

Note that Corollary \ref{cor:rotate0} provides a condition on primitives
($F^{\star}(\theta)/f^{\star}(\theta)\ge d(\theta)$ for all $\theta\in\Theta$,
where $d(\theta)=P^{\star}(\underline{D}(\theta))-\theta$), under
which the Baron-Myerson-with-quantity-floor mechanism is robustly
optimal. Proposition \ref{prop:comp_reg} then implies that, under
the same condition, quantity regulation dominates price regulation.
The next result provides further primitive conditions under which
each regulation dominates the other.

For every $\theta\in\Theta$, let $\delta(\theta)\equiv D^{\star}(\theta)-\underline{D}(\theta)$
denote the distance between the conjectured and the lowest \emph{direct
demands}, and then let $\underline{\delta}\equiv\inf\{\delta(\theta):\theta\in\Theta\}$,
$\bar{\delta}\equiv\sup\{\delta(\theta):\theta\in\Theta\}$, $\underline{f}\equiv\inf\{f^{\star}(\theta):\theta\in\Theta\}$
and $\bar{f}\equiv\sup\{f^{\star}(\theta):\theta\in\Theta\}$.

\begin{proposition}[quantity vs price regulation: markup-based conditions]\label{prop:primitives}
Suppose $D^{\star}$ is bi-Lipschitz continuous over $[P^{\star}({\rm \bar{q})},P^{\star}(0)]$
with constants $0<k\le K$.\footnote{That is, for any $p,p'\in[P^{\star}({\rm \bar{q})},P^{\star}(0)]$,
with $p'>p$, $D^{\star}(p)-D^{\star}(p')$$\in[k(p'-p),K(p'-p)]$.} 
\begin{enumerate}
\item[(a)] Quantity regulation dominates price regulation if, for all $\theta\in\Theta$,
$F^{\star}(\theta)/f^{\star}(\theta)\geq\delta(\theta)/k$, with the
domination strict if $\delta(\overline{\theta})>0$. 
\item[(b)] When, instead, for all $\theta\in\Theta$, $F^{\star}(\theta)/f^{\star}(\theta)\leq\delta(\theta)/K[1+\frac{\bar{\delta}}{\underline{\delta}}\sqrt{2\bar{f}/\underline{f}}]$
and $\underline{\delta}>0$, then price regulation strictly dominates
quantity regulation. 
\end{enumerate}
\end{proposition}

The proposition thus establishes that the type of regulation that
dominates depends on the size of the markup that the designer would
like to apply to each cost, relative to the difference between the
two demands. Quantity regulation dominates when the markups $F^{\star}(\theta)/f^{\star}(\theta)$
are large relative to the downstream demand uncertainty $\delta(\theta)$
faced by the regulator. This is because the quantity that the designer
would like to procure from each type is small and hence, in the eyes
of the regulator, the danger of over-procurement in case the designer's
model is wrong is also small. 

The second part of the proposition, in turn, establishes that price
regulation strictly dominates when the markups are small relative
to the difference between the two demands. When this is the case,
in the eyes of the regulator, there is significant danger of over-procurement
from intermediate-cost types under quantity regulation. As a consequence,
if the designer were to opt for quantity regulation, the regulator
would require a cut to the quantity procured from intermediate types,
as established in the previous section. The designer can avoid these
cuts through price regulation (which guards the regulator by letting
the quantity adjust to the realized demand). As explained above, the
drawback of this regulation, in the eyes of the designer, is that
it may result in over-procurement from high-cost types. However, the
conditions in part (b) imply that the cost of such over-procurement
to the designer is small compared to the cost of reducing the procurement
from intermediate-cost sellers under quantity regulation, thus making
price regulation superior.

\section{Conclusions}

\label{sec:Conclusions}

We study the problem of a Bayesian designer who must obtain approval
from a max-min authority that does not share the designer’s confidence
in the model and requires satisfactory performance across alternative
models.

We conduct the analysis in a procurement setting, where approval by
conservative regulators, policymakers, or supervisors is common. We
show that robustness reshapes the efficiency-rent-extraction trade-off,
calling for increased procurement from high-cost sellers and reduced
procurement from intermediate-cost sellers relative to the Bayesian
optimum. Quantity floors protect against higher-than-expected costs,
while reducing procurement from intermediate-cost sellers guards against
demand overestimation. When the good is sold in a downstream market,
our approach also delivers a novel comparison of price and quantity regulation:
the two instruments provide identical welfare guarantees but yield
different payoffs under the designer’s conjectured model. Which instrument
dominates depends on the desired markups and the extent of demand
uncertainty.

We expect the approach to be relevant in other applications. In auction
markets, for example, sellers often rely on experts to design mechanisms
(e.g., for spectrum allocation), yet rarely share the designers’ confidence
in the underlying model. This lack of trust may stem from limited
understanding of the estimation methods or from concerns that the
environment has changed. Similarly, in public good provision, local
governments often need approval from state or federal authorities
that do not share their confidence in the model used to design the
relevant contracts. The same applies to taxation, where reforms of
fiscal rules at the local level typically require approval by conservative
central authorities. Extending the approach developed here to these
settings is a promising direction for future research.

\begin{appendix}

\section{Appendix}
\label{Sec:appendix}

The Appendix is divided into two section. Sections \ref{sec:proofsec3} and \ref{sec:proofsec4} contain the missing proofs for the results from Sections \ref{sec:robustly-quantity-mechanisms} and \ref{sec:price_vs_quantity} respectively.

\subsection{Proofs of Section \ref{sec:robustly-quantity-mechanisms}}
\label{sec:proofsec3}

\textbf{Proof of Lemma \ref{lemma-guarantee}}. For any $(V,F)\in\mathcal{V}\times\mathcal{F}$,
\begin{align*}
\int\left[V(q(\theta))-\theta q(\theta)-u(\theta)\right]F(\mathrm{d}\theta)\ge\inf_{\theta\in\Theta}\{\underline{V}(q(\theta))-\theta q(\theta)-u(\theta)\},
\end{align*}
where the inequality follows from the definition of $\underline{V}$.
Hence, $G(M)\ge\inf_{\theta\in\Theta}\{\underline{V}(q(\theta))-\theta q(\theta)-u(\theta)\}.$
Because $\underline{V}\in\mathcal{V}$ and, for each $\theta$, the
Dirac distribution that puts probability mass one at $\theta$ is
in $\mathcal{F}$, $G(M)\leq\inf_{\theta\in\Theta}\{\underline{V}(q(\theta))-\theta q(\theta)-u(\theta)\}$.
Combining the last two inequalities, we obtain the first part of Condition
(\ref{eq:wg1}). To see the second part, observe that $G(M)\le\underline{V}(q(\overline{\theta}))-\overline{\theta}q(\overline{\theta})-u(\overline{\theta})\le G^{*}$,
where the first inequality holds because Nature can always choose
$\underline{V}\in\mathcal{V}$ and the Dirac distribution that puts
probability mass one at $\overline{\theta}$. The second inequality
follows from the fact that $u(\overline{\theta})\ge0$ along with
the definition of $G^{*}$.\hfill{}Q.E.D.\smallskip{}

\noindent\textbf{Proof of Lemma \ref{lemma-SL-characterization}}.
First, we show that there exists an IC and IR mechanism that delivers
the guarantee upper bound in (\ref{eq:wg1}). Consider the constant
mechanism $M_{L}\equiv(q_{L},u_{L})$ in which $q_{L}(\theta)={\rm q}_{\ell}$
and $t_{L}(\theta)=\overline{\theta}{\rm q}_{\ell}$, for every $\theta\in\Theta$
(yielding a rent $u_{L}(\theta)=(\overline{\theta}-\theta){\rm q}_{\ell}$
to each $\theta$). The mechanism $M_{L}$ is clearly IC and IR. Furthermore,
$\inf_{\theta}\left\{ \underline{V}(q(\theta))-\theta q(\theta)-u(\theta)\right\} =G^{*}$.
The first part of Condition (\ref{eq:wg1}) in Lemma \ref{lemma-guarantee}
then implies that $G(M_{L})=G^{*}$. The second part of Condition
(\ref{eq:wg1}) in Lemma \ref{lemma-guarantee} in turn implies that
$M_{L}\in\mathcal{M}^{{\rm SL}}$. Therefore, for any $M\in\mathcal{M}^{{\rm SL}}$,
$G(M)=G^{*}$.

For a mechanism $M=(q,u)$ to be IC and IR, it must be that $q$ is
weakly decreasing and, for all $\theta$, $u(\theta)=u(\overline{\theta})+\int\limits^{\overline{\theta}}_{\theta}q(y)dy$,
with $u(\overline{\theta})\geq0$. Moreover, if $M\in\mathcal{M}^{{\rm SL}}$,
then, for all $\theta$, $\underline{V}(q(\theta))-\theta q(\theta)-u(\theta)\geq G^{*}$.
This is possible only if $u(\overline{\theta})=0$ (else, the constraint
is violated at $\overline{\theta}$) and, for any $\theta$, constraint
(\ref{SL-constraint}) holds. We conclude that any $M\in\mathcal{M}^{{\rm SL}}$
must satisfy the constraint (\ref{SL-constraint}).

\noindent Now, let $M$ be any IC and IR mechanism with $u(\overline{\theta})=0$
satisfying constraint (\ref{SL-constraint}). By the first part of
Condition (\ref{eq:wg1}) in Lemma \ref{lemma-guarantee}, $G(M)\geq G^{*}$.
Because every mechanism in $\mathcal{M}^{{\rm SL}}$ has a welfare
guarantee of $G^{*}$, we conclude that $M\in\mathcal{M}^{{\rm SL}}$.
\hfill{} Q.E.D.

\smallskip{}

\noindent\textbf{Proof of Lemma \ref{lemma:majorizations}}. \textbf{Part
(1)}. Observe that, for any $\theta\in\Theta$, 
\begin{align*}
\underline{V}(q(\theta))-\theta q(\theta) & =\int^{\infty}_{\theta}\underline{D}(y){\rm d}y-\underline{{\rm DWL}}(\theta,q(\theta)),
\end{align*}
where, for any ${\rm q}$, $\underline{{\rm DWL}}(\theta,{\rm q})\equiv\int^{\underline{P}({\rm q})}_{\theta}\big(\underline{D}(y)-{\rm q}\big){\rm d}y\ge0$.
The equivalence between the constraints in (\ref{SL-constraint})
and (\ref{eq:strong-major}) follows from this observation together
with the fact that $G^{*}\equiv\underline{V}({\rm q}_{\ell})-\overline{\theta}{\rm q}_{\ell}=\underline{V}(\underline{D}(\overline{\theta}))-\overline{\theta}\underline{D}(\overline{\theta})=\int\limits^{\infty}_{\overline{\theta}}\underline{D}(y){\rm d}y.$

\noindent\textbf{Part (2).} That (a) implies (b) is immediate given
that, for all $\theta$, $\underline{{\rm DWL}}(\theta,q(\theta))\ge0$.
To complete the proof, it thus suffices to show that, when Condition
(\ref{eq:strong-major}) holds for $\theta\in\{\underline{\theta},\overline{\theta}\}$
and, in addition, Condition (\ref{eq:weak-major}) holds for all $\theta\in(\underline{\theta},\overline{\theta})$,
the following inequality 
\begin{equation}
\int\limits^{\overline{\theta}}_{\theta}\underline{D}(y)dy-\int\limits^{\overline{\theta}}_{\theta}q(y)dy-\int\limits^{\underline{P}(q(\theta))}_{\theta}\big(\underline{D}(y)-q(\theta)\big)dy\geq0\label{eq:robut-ineq-manipulated}
\end{equation}
holds for all $\theta\in(\underline{\theta},\overline{\theta})$.
We consider two cases, which are covered in Claims \ref{cl:case1}
and \ref{cl:case2} below.

\begin{claim} \label{cl:case1} Suppose that $q(\theta)\le\underline{D}(\theta)$
and, for all $\theta'\in\Theta$, $\int\limits^{\overline{\theta}}_{\theta'}q(y){\rm d}y\leq\int\limits^{\overline{\theta}}_{\theta'}\underline{D}(y){\rm d}y$.
Then the inequality in (\ref{eq:robut-ineq-manipulated}) holds for
$\theta$.\end{claim}

\noindent\emph{Proof of Claim \ref{cl:case1}}. Observe that $\theta\le\underline{P}(q(\theta))\le\overline{\theta}$.
The first inequality follows because $q(\theta)\le\underline{D}(\theta)$.
The second inequality follows from Condition (\ref{eq:strong-major})
applied to $\theta=\overline{\theta}$ which gives $q(\overline{\theta})={\rm q}_{\ell}=\underline{D}(\overline{\theta})$;
because $q$ is weakly decreasing, we then have that $q(\theta)\geq\underline{D}(\overline{\theta})$.
The left-hand-side of (\ref{eq:robut-ineq-manipulated}) is thus equivalent
to 
\begin{equation}
\int\limits^{\overline{\theta}}_{\underline{P}(q(\theta))}(\underline{D}(y)-q(y)){\rm d}y+\int\limits^{\underline{P}(q(\theta))}_{\theta}(q(\theta)-q(y)){\rm d}y.\label{eq:inequality-claim-1}
\end{equation}
The first integral in (\ref{eq:inequality-claim-1}) is non-negative;
this follows from the second condition in Claim \ref{cl:case1} applied
to $\theta'=\underline{P}(q(\theta))$. The second integral is non-negative
because $q$ is weakly decreasing. This completes the proof of Claim
\ref{cl:case1}. \hfill{} $\blacksquare$

\begin{claim} \label{cl:case2} Suppose that $q(\theta)>\underline{D}(\theta)$,
and, for all $\theta'\in\Theta$, $\int\limits^{\overline{\theta}}_{\theta'}q(y){\rm d}y\leq\int\limits^{\overline{\theta}}_{\theta'}\underline{D}(y){\rm d}y$.
Then the inequality in (\ref{eq:robut-ineq-manipulated}) holds for
$\theta$. \end{claim}

\noindent\emph{Proof of Claim \ref{cl:case2}.} Because $q(\theta)>\underline{D}(\theta)$,
$\underline{P}(q(\theta))<\theta$, which implies that the left-hand-side
of the inequality in (\ref{eq:robut-ineq-manipulated}) is equal to
$\int\limits^{\overline{\theta}}_{\theta}(\underline{D}(y)-q(y)){\rm d}y-\int\limits^{\theta}_{\underline{P}(q(\theta))}(q(\theta)-\underline{D}(y)){\rm d}y$.
Let $\theta^{\sharp}\equiv\inf\{y\le\theta:q(s)\geq\underline{D}(s)\:\mathrm{for\:all}\:s\in(y,\theta]\}$
and note that $\theta^{\sharp}<\theta$. Suppose that $\theta^{\sharp}>\underline{\theta}$,
which implies that $q(\theta^{\sharp})=\underline{D}(\theta^{\sharp})$.
Applying the second condition in the statement of Claim \ref{cl:case2}
to $\theta'=\theta^{\sharp}$, we have that $\int\limits^{\overline{\theta}}_{\theta^{\sharp}}\left(\underline{D}(y)-q(y)\right)dy\ge0$.
Hence, $\int\limits^{\overline{\theta}}_{\theta}\left(\underline{D}(y)-q(y)\right)dy\ge\int\limits^{\theta}_{\theta^{\sharp}}\left(q(y)-\underline{D}(y)\right)dy.$
It follows that 
\begin{align}
\int\limits^{\overline{\theta}}_{\theta}\left(\underline{D}(y)-q(y)\right)dy-\int\limits^{\theta}_{\underline{P}(q(\theta))}\big(q(\theta)-\underline{D}(y)\big)dy\ge\int\limits^{\theta}_{\theta^{\sharp}}\left(q(y)-\underline{D}(y)\right)dy-\int\limits^{\theta}_{\underline{P}(q(\theta))}\big(q(\theta)-\underline{D}(y)\big)dy.\label{eq:new-ineq-claim-2}
\end{align}
Because $q(\theta)\le q(\theta^{\sharp})=\underline{D}(\theta^{\sharp})$,
we have that $\underline{P}(q(\theta))\ge\theta^{\sharp}$. This property,
together with the fact that $q(y)>\underline{D}(y)$ for all $y\in(\theta^{\sharp},\theta)$,
implies that $\int\limits^{\theta}_{\theta^{\sharp}}(q(y)-\underline{D}(y))dy\geq\int\limits^{\theta}_{\underline{P}(q(\theta))}(q(y)-\underline{D}(y))dy$.
In turn, this means that the right-hand-side of the inequality in
(\ref{eq:new-ineq-claim-2}) is greater than 
\begin{align*}
\int\limits^{\theta}_{\underline{P}(q(\theta))}\left(q(y)-\underline{D}(y)\right)dy-\int\limits^{\theta}_{\underline{P}(q(\theta))}\big(q(\theta)-\underline{D}(y)\big)dy=\int\limits^{\theta}_{\underline{P}(q(\theta))}\left(q(y)-q(\theta)\right)dy\ge0,
\end{align*}
where the last inequality follows from the fact that $q$ is weakly
decreasing. We thus conclude that the inequality in (\ref{eq:robut-ineq-manipulated})
holds, as claimed.

Next, suppose that $\theta^{\sharp}=\underline{\theta}$. That Condition
(\ref{eq:strong-major}) holds for $\theta=\underline{\theta}$ means
that $\underline{W}(\underline{\theta},q)\ge G^{*}$, where, for any
$\theta\in\Theta$, 
\begin{equation}
\underline{W}(\theta,q)\equiv\underline{V}(q(\theta))-\theta q(\theta)-\int\limits^{\overline{\theta}}_{\theta}q(y)dy.\label{eq:W-fn}
\end{equation}
Because $q(y)\geq\underline{D}(y)$ for all $y\in[\underline{\theta},\theta]$,
the function $\underline{W}(\cdot,q)$ is weakly increasing over $[\underline{\theta},\theta]$
(see Lemma \ref{Lem-Mono} below) and hence $\underline{W}(\theta,q)\ge\underline{W}(\underline{\theta},q)\ge G^{*}$.
As shown above, this means that (\ref{eq:robut-ineq-manipulated})
holds. This completes the proof of Claim \ref{cl:case2}. \hfill{}
$\blacksquare$

Claims \ref{cl:case1} and \ref{cl:case2} cover all cases and jointly
establish Lemma \ref{lemma:majorizations}. \hfill{}Q.E.D.\smallskip{}

\begin{lemma}[monotonicity]\label{Lem-Mono} Suppose $M\equiv(q,u)$
is an IC mechanism and $I\subseteq\Theta$ is an interval. Let $\underline{W}(\cdot,q)$
be the function defined, for all $\theta\in\Theta$, by (\ref{eq:W-fn}).

\noindent{[}A{]}. Suppose $0<q(\theta)\le\underline{D}(\theta)$
for all $\theta\in I$. Then $\underline{W}(\cdot,q)$ is weakly decreasing
over $I$. If, in addition, $q$ is decreasing with $q(\theta)<\underline{D}(\theta)$
for all $\theta\in I$, then $\underline{W}(\cdot,q)$ is decreasing
over $I$.

\noindent{[}B{]}. Suppose $q(\theta)>\underline{D}(\theta)$ for
all $\theta\in I$. Then, $\underline{W}(\cdot,q)$ is weakly increasing
over $I$. If, in addition, $q$ is decreasing over $I$, then $\underline{W}(\cdot,q)$
is increasing over $I$. \end{lemma}

\noindent\textbf{Proof of Lemma \ref{Lem-Mono}.} Pick $\theta,\theta'\in I$,
with $\theta'<\theta$. Note that 
\begin{align}
\underline{W}(\theta',q)-\underline{W}(\theta,q) & =\int^{q(\theta')}_{q(\theta)}\underline{P}(y)dy-\theta'q(\theta')+\theta q(\theta)-\int^{\theta}_{\theta'}q(y)dy.\label{eq:d1}
\end{align}
\textbf{\textsc{Proof of Part (A)}}. We consider two cases.

\noindent\textbf{Case 1}: $\underline{P}(q(\theta'))\ge\theta>\theta'$.
Note that the right-hand-side of (\ref{eq:d1}) is equal to
\begin{equation}
\int^{q(\theta')}_{q(\theta)}(\underline{P}(y)-\theta)dy+\left\{ (\theta-\theta')q(\theta')-\int^{\theta}_{\theta'}q(y)dy\right\} .\label{eq:d1-BIS}
\end{equation}
The first term in (\ref{eq:d1-BIS}) is non-negative because, for
all $y\in(q(\theta),q(\theta'))$, $\underline{P}(y)>\underline{P}(q(\theta'))\ge\theta$,
which follows from $\underline{P}$ being decreasing. Furthermore,
if $q(\theta')>q(\theta)$, then the first term in (\ref{eq:d1-BIS})
is positive. Next, observe that, because $q$ is weakly decreasing,
the expression in curly brackets in (\ref{eq:d1-BIS}) is non-negative.
Thus, $\underline{W}(\theta',q)\ge\underline{W}(\theta,q)$, i.e.,
$\underline{W}(\cdot,q)$ is weakly decreasing over $I$ (decreasing
when $q$ is decreasing over $I$).

\noindent\textbf{Case 2}: $\theta>\underline{P}(q(\theta'))\ge\theta'$.
Use Panel A in Figure \ref{fig:mono} to observe that the sum of the
first four terms in (\ref{eq:d1}) is equal to 
\begin{align}
\int^{\underline{P}(q(\theta'))}_{\theta'}\Big(q(\theta')-q(y)\Big)dy+\int^{\theta}_{\underline{P}(q(\theta'))}\Big(\underline{D}(y)-q(y)\Big)dy+\int^{\underline{P}(q(\theta))}_{\theta}\Big(\underline{D}(y)-q(\theta)\Big)dy.\label{eq:d2}
\end{align}

\begin{figure}
\centering %
\begin{minipage}[b]{0.49\linewidth}%
\includegraphics[scale=0.4]{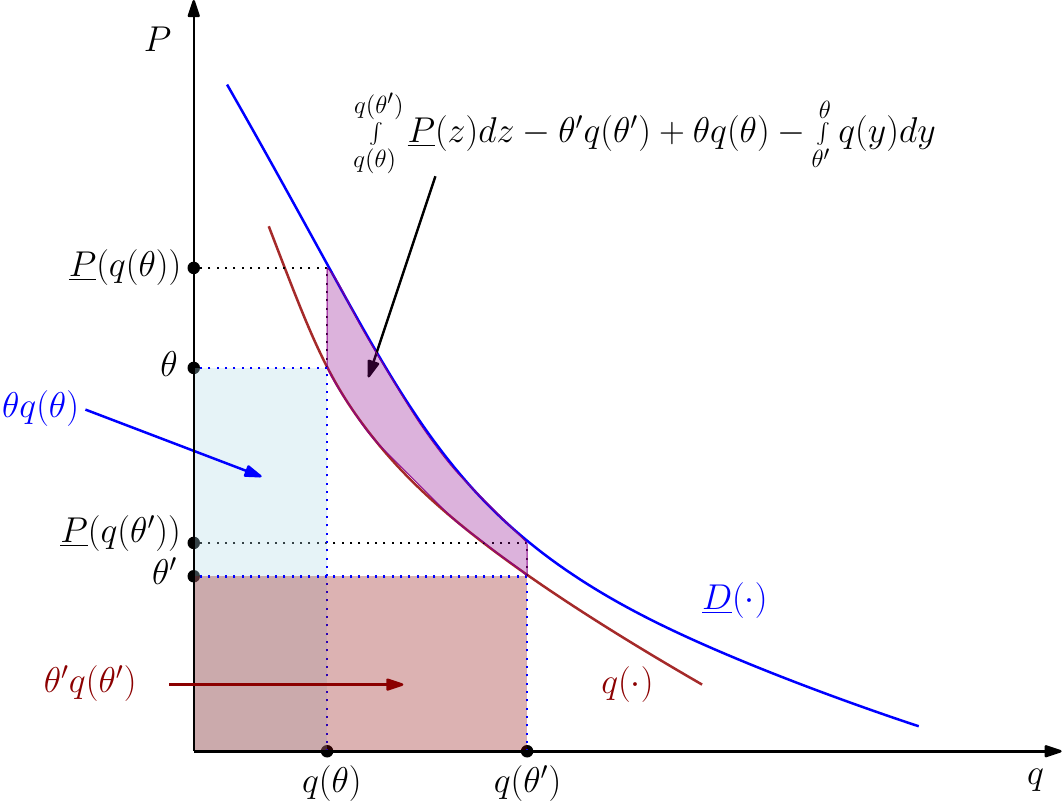} \caption*{\centering A: Case 2 in Part A.}
\end{minipage}%
\begin{minipage}[b]{0.49\linewidth}%
\includegraphics[scale=0.4]{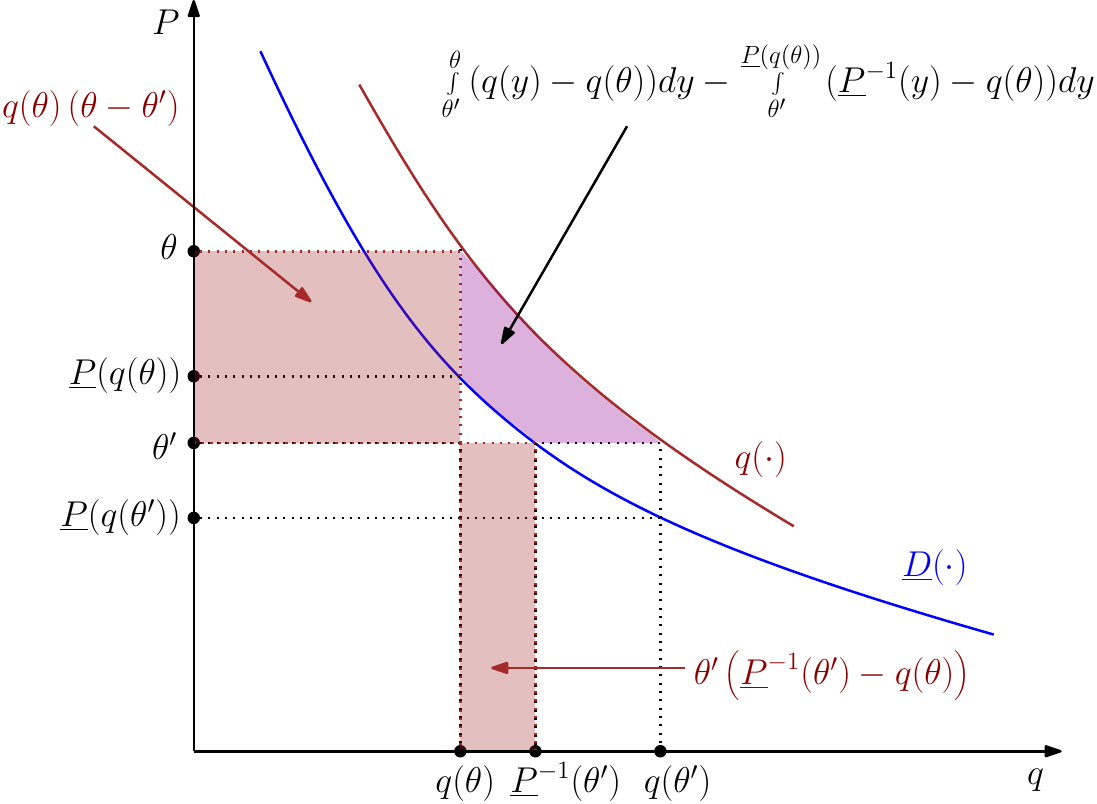} \caption*{\centering B: Case 2 in Part B.}
\end{minipage}\caption{Illustration for Lemma \ref{Lem-Mono} }
\label{fig:mono} 
\end{figure}


The first term in (\ref{eq:d2}) is non-negative because $q$ is weakly
decreasing. Next observe that, for all $y\in(\underline{P}(q(\theta')),\theta),$
$\underline{D}(y)\geq q(y)$. Hence, the second term in (\ref{eq:d2})
is also non-negative. Finally, the last term in (\ref{eq:d2}) is
also non-negative because, for any $y\in(\theta,\underline{P}(q(\theta)))$,
$\underline{D}(y)\ge q(\theta)$, which follows from $\underline{P}$
being decreasing. We conclude that $\underline{W}(\cdot,q)$ is weakly
decreasing over $I$ (decreasing when $q$ is decreasing and such
that $q(y)<\underline{D}(y)$ for all $y\in I$).

\noindent\textbf{\textsc{Proof of Part (B)}}: The difference in welfare
$\underline{W}(\theta,q)-\underline{W}(\theta',q)$ across the two
types is given by the (negative of) expression in (\ref{eq:d1}),
which can be rewritten as 
\begin{align}
\underline{W}(\theta,q)-\underline{W}(\theta',q) & =\int^{\theta}_{\theta'}\Big(q(y)-q(\theta)\Big)dy+\int^{q(\theta')}_{q(\theta)}\Big(\theta'-\underline{P}(z)\Big)dz.\label{eq:d111}
\end{align}
We consider two cases.

\noindent\textbf{Case 1}: $\underline{P}(q(\theta))\leq\theta'$.
In this case, $\underline{P}(z)<\theta'$ for all $z>q(\theta)$.
This implies that the second integral in (\ref{eq:d111}) is non-negative.
The first integral is non-negative because $q$ is weakly decreasing.
If $q$ is decreasing, both integrals are positive.

\noindent\textbf{Case 2}: $\theta'<\underline{P}(q(\theta))<\theta$.
We then have that $q(\theta')>\underline{D}(\theta')>q(\theta)$.
Hence, using (\ref{eq:d111}), we have that 
\begin{align}
\underline{W}(\theta,q)-\underline{W}(\theta',q) & \geq\int^{\theta}_{\theta'}\Big(q(y)-q(\theta)\Big)dy-\int^{\underline{D}(\theta')}_{q(\theta)}\Big(\underline{P}(z)-\theta'\Big)dz.\label{eq:diff_final}
\end{align}
Changing the variable of integration, the second integral can be written
as 
\[
\int^{\underline{D}(\theta')}_{q(\theta)}\Big(\underline{P}(z)-\theta'\Big)dz=\int^{\underline{P}(q(\theta))}_{\theta'}\Big(\underline{D}(y)-q(\theta)\Big)dy.
\]
Thus, the right-hand-side of (\ref{eq:diff_final}) reduces to (see
Panel B in Figure \ref{fig:mono} for an illustration) 
\begin{equation}
\int^{\theta}_{\theta'}\Big(q(y)-q(\theta)\Big)dy-\int^{\underline{P}(q(\theta))}_{\theta'}\Big(\underline{D}(y)-q(\theta)\Big)dy,\label{eq:laststep-M}
\end{equation}
which is non-negative because $\underline{P}(q(\theta))<\theta$ and
$\underline{D}(y)<q(y)$ for all $y\in I$. 

Thus, $\underline{W}(\theta,q)\ge\underline{W}(\theta',q)$, i.e.,
$\underline{W}(\cdot,q)$ is weakly increasing over $I$. The inequality
above also reveals that, when $q$ is decreasing, the expression in
(\ref{eq:laststep-M}) is positive, implying that $\underline{W}(\cdot,q)$
is increasing over $I$. This completes the proof of the lemma. \hfill{}Q.E.D.\smallskip{}

\noindent\textbf{Proof of Proposition \ref{prop:BM-floor}.} By Observation
\ref{obs:qell}, if $M=(q,u)\in\mathcal{M}^{{\rm SL}}$, then $q(\theta)\ge{\rm q}_{\ell}$
for all $\theta\in\Theta$. The following is thus a relaxation of
problem (\ref{opt:ropt}): 
\begin{align}
 & ~~~~\max_{q}\int\limits^{\overline{\theta}}_{\underline{\theta}}\Big[V^{\star}(q(\theta))-z^{\star}(\theta)q(\theta)\Big]F^{\star}(\text{d}\theta)~~~~\tag{\textbf{RP}}\label{opt:rel}\\
\textrm{subject to} & ~q~\textrm{ weakly decreasing}~\textrm{and}~q(\theta)\ge{\rm q}_{\ell}~\forall~\theta\in\Theta.\nonumber 
\end{align}
That the quantity schedule $q^{\star}$ satisfies the constraints
in the relaxed problem (\ref{opt:rel}) follows from the fact that,
when $F^{\star}$ is regular, $\qbm$ is decreasing and hence $q^{\star}$
is weakly decreasing. That $q^{\star}$ also satisfies the other constraint
of the relaxed problem follows directly from its definition. Next,
observe that, for any $\theta$, the function $V^{\star}({\rm q})-z^{\star}(\theta){\rm q}$
is concave in ${\rm q}$ and attains a maximum at $\qbm(\theta)$.
Therefore, the quantity schedule $q^{\star}$ maximizes the objective
function in the relaxed program over all weakly decreasing functions
$q$ satisfying $q(\theta)\ge{\rm q}_{\ell}$ for all $\theta$.

To complete the proof, it suffices to show that $q^{\star}$ satisfies
the robustness constraints in (\ref{SL-constraint}) if and only if
Conditions (\ref{eq:robustness-bottom}) and (\ref{eq:weak-majorization})
in the proposition hold. This follows from Lemma \ref{lemma:majorizations}
along with the continuity of $q^{\star}$ and $\underline{D}$ which
implies that, if $q^{\star}$ satisfies Condition (\ref{eq:weak-majorization}),
then $q^{\star}(\overline{\theta})={\rm {q}_{\ell}}$.  \hfill{}Q.E.D.\smallskip{}

\noindent\textbf{Proof of Proposition \ref{prop:robust-quantity-mechanism-general}}.\textsc{
Part 1}. The result follows from Lemmas \ref{lem:weakqoptbound}-\ref{lem:leftofthetam}
below.

\noindent\begin{lemma} \label{lem:weakqoptbound} Every $M^{\textsc{OPT}}\equiv(\qopt,\uopt)$
is such that $\qopt(\theta)\le q^{\star}(\theta)$ for all $\theta\in(\tlb,\tub]$
and $\qopt(\theta)={\rm q}_{\ell}$ for all $\theta\in[\theta^{\star},\tub]$.\end{lemma}

\noindent\textbf{Proof of Lemma \ref{lem:weakqoptbound}.} By Observation
\ref{obs:qell}, $\qopt(\tub)={\rm q}_{\ell}$. Assume, towards a
contradiction, that $\qopt(\theta)>q^{\star}(\theta)$ for some $\theta\in(\tlb,\tub)$.
The continuity of $q^{\star}$ over $(\tlb,\tub)$ along with the
fact that $\qopt$ is weakly decreasing, imply that there is a positive
measure set of types $I\subseteq(\tlb,\tub)$ such that $\qopt(y)>q^{\star}(y)$
for all $y\in I$. Then consider the quantity schedule $\tilde{q}(y)=\min\{\qopt(y),q^{\star}(y)\}$
for all $y\in(\tlb,\tub)$ with $\tilde{q}(y)=\qopt(y)$ for all $y\in\{\tlb,\tub\}$.
Clearly, $\tilde{q}$ is weakly decreasing. Furthermore, because $\qopt$
satisfies constraint (\ref{eq:strong-major}) for $\theta\in\{\tlb,\tub\}$,
and constraint (\ref{eq:weak-major}) for all $\theta\in(\tlb,\tub)$,
and because $\tilde{q}(y)\le\qopt(y)$ for all $y$, with equality
at $\tlb$ and $\tub$, $\tilde{q}$ also satisfies constraint (\ref{eq:strong-major})
for $\theta\in\{\tlb,\tub\}$, and constraint (\ref{eq:weak-major})
for all $\theta\in(\tlb,\tub)$. By Lemma \ref{lemma:majorizations},
this means that $\tilde{q}$ satisfies the robustness constraint in
(\ref{opt:ropt}) for all $\theta\in\Theta$.

Next observe that (a) $\qbm(y)\le q^{\star}(y)$ for all $y\in\Theta$
and (b) $\qbm$ maximizes the objective function in (\ref{opt:ropt})
point-wise. The concavity of $V^{\star}({\rm q})-z^{\star}(\theta){\rm q}$
in ${\rm q}$ then implies that the objective function in (\ref{opt:ropt})
is strictly higher under $\tilde{q}$ than under $\qopt$, where the
strict inequality follows because $I$ has positive Lebesgue measure.
This contradicts the optimality of $\qopt$. We conclude that $\qopt(\theta)\le q^{\star}(\theta)$
for all $\theta\in(\tlb,\tub]$. Because $q^{\star}(\theta)={\rm q}_{\ell}$
for all $\theta\in[\theta^{\star},\tub]$, we thus have that $\qopt(\theta)\le{\rm q}_{\ell}$
for all $\theta\in[\theta^{\star},\tub]$. However, by Observation
\ref{obs:qell}, $\qopt(\tub)={\rm q}_{\ell}$. Because $\qopt$ is
weakly decreasing, it must be that $\qopt(\theta)={\rm q}_{\ell}$
for all $\theta\in[\theta^{\star},\tub]$. $\hfill\blacksquare{}$

The next lemma compares $\theta^{m}$ and $\theta^{\star}$ when $M^{\star}$
is not robustly optimal.

\begin{lemma} \label{lemm-theta_m-star} The mechanism $M^{\star}$
is robustly optimal if and only if $\theta^{m}=\tub$. Furthermore,
if $M^{\star}$ is not robustly optimal, then $\theta^{m}<\theta^{\star}$.
\end{lemma}

\noindent\textbf{Proof of Lemma \ref{lemm-theta_m-star}.} Suppose
$M^{\star}$ is robustly optimal. Then, $M^{\star}\in\mathcal{M}^{{\rm SL}}$.
Observation \ref{obs:qell} then implies that $q^{\star}(\tub)={\rm q}_{\ell}$.
This implies that $\underline{W}(\tub,q^{\star})=G^{*}$. Lemma \ref{lemma-SL-characterization}
in turn implies that $\underline{W}(\theta,q^{\star})\ge G^{*}$ for
all $\theta\in\Theta$. Hence, $\underline{W}(\theta,q^{\star})\ge\underline{W}(\tub,q^{\star})$
for all $\theta.$ This implies that $\theta^{m}=\tub$.

\noindent For the converse, suppose that $\theta^{m}=\tub$. Then,
by definition of $\theta^{m}$, $\underline{W}(\theta,q^{\star})\ge\underline{W}(\tub,q^{\star})$
for all $\theta\in\Theta$. This implies that $q^{\star}(\tub)={\rm q}_{\ell}$.
To see this, suppose, towards a contradiction, that $q^{\star}(\tub)>{\rm q}_{\ell}$.
This means that $q^{\star}(\tub)=\qbm(\tub)>{\rm q}_{\ell}=\underline{D}(\tub)$.
By continuity of $\qbm$ and $\underline{D}$, there exists $\epsilon>0$
such that $q^{\star}(\theta)>\underline{D}(\theta)$ for all $\theta\in(\tub-\epsilon,\tub]$.
Because $\qbm$ is decreasing over $(\tub-\epsilon,\tub]$, Lemma
\ref{Lem-Mono} implies that $\underline{W}(\cdot,q^{\star})$ is
increasing over $(\tub-\epsilon,\tub]$, a contradiction to the assumption
that $\theta^{m}=\tub$. Thus, $q^{\star}(\tub)={\rm q}_{\ell}$.
But then $\underline{W}(\tub,q^{\star})=G^{*}$. Together with the
definition of $\theta^{m}=\tub$, this implies that $\underline{W}(\theta,q^{\star})\ge G^{*}$
for all $\theta\in\Theta$. Hence, $M^{\star}\in\mathcal{M}^{{\rm SL}}$.
Proposition \ref{prop:BM-floor} then implies that $M^{\star}$ is
robustly optimal.

Finally, to see that, when $M^{\star}$ is not robustly optimal, $\theta^{m}<\theta^{\star}$,
observe that, when $\theta^{\star}=\tub$, the result follows from
the fact that $\theta^{m}<\tub$. When, instead, $\theta^{\star}<\tub$,
then $\underline{W}(\theta,q^{\star})=G^{*}$ for all $\theta\in[\theta^{\star},\tub]$.
Because $M^{\star}$ is not robustly optimal, it must be that $\underline{W}(\theta^{m},q^{\star})<G^{*}$.
Hence, $\theta^{m}<\theta^{\star}$. $\hfill\blacksquare{}$

The next Lemma compares $q^{\star}(\theta^{m})$ with $\underline{D}(\theta^{m})$.

\begin{lemma} \label{lem:qstardlowint} Suppose $M^{\star}$ is not
robustly optimal. Then, $q^{\star}(\theta^{m})\ge\underline{D}(\theta^{m})$,
with the inequality holding as equality if $\theta^{m}>\tlb$. \end{lemma}

\noindent\textbf{Proof of Lemma \ref{lem:qstardlowint}.} By Lemma
\ref{lemm-theta_m-star}, $\theta^{m}<\theta^{\star}$. This means
that $q^{\star}(\theta^{m})=\qbm(\theta^{m})$. Assume, towards a
contradiction, that $\qbm(\theta^{m})<\underline{D}(\theta^{m})$.
Because $\qbm$ and $\underline{D}$ are continuous and $\theta^{m}<\theta^{\star}$,
there exists $\epsilon>0$ such that $\theta^{m}+\epsilon<\theta^{\star}$
and $\qbm(\theta)<\underline{D}(\theta)$ for all $\theta\in[\theta^{m},\theta^{m}+\epsilon]$.
Because $\qbm$ (and thus $q^{\star}$) is decreasing on this interval,
Lemma \ref{Lem-Mono} implies that $\underline{W}(\cdot,q^{\star})$
is decreasing on $[\theta^{m},\theta^{m}+\epsilon]$. This contradicts
the definition of $\theta^{m}$. Hence, $q^{\star}(\theta^{m})\ge\underline{D}(\theta^{m})$.

By the same argument, if $\theta^{m}>\tlb$ and $q^{\star}(\theta^{m})>\underline{D}(\theta^{m})$,
there exists $\epsilon>0$ such that $\underline{W}(\cdot,q^{\star})$
is increasing on $[\theta^{m}-\epsilon,\theta^{m}]$, contradicting
the definition of $\theta^{m}$. Hence, $q^{\star}(\theta^{m})=\underline{D}(\theta^{m})$
if $\theta^{m}>\tlb$.$\hfill\blacksquare{}$

The next lemma uses the property in Lemma \ref{lem:qstardlowint}
to complete the proof of Part 1 by establishing that $\qopt$ and
$\qbm$ agree on $(\tlb,\theta^{m})$.

\begin{lemma} \label{lem:leftofthetam} Suppose $M^{\star}$ is not
robustly optimal. If $\theta^{m}>\tlb$, then in every robustly optimal
mechanism, $\qopt(\theta)=\qbm(\theta)$ for all $\theta\in(\tlb,\theta^{m})$.\end{lemma}

\noindent\textbf{Proof of Lemma \ref{lem:leftofthetam}.} Because
$M^{\star}$ is not robustly optimal, by Lemma \ref{lemm-theta_m-star},
$\theta^{m}<\theta^{\star}$. Hence, $q^{\star}(\theta)=\qbm(\theta)$
for all $\theta\in(\tlb,\theta^{m})$. Furthermore, by Lemma \ref{lem:weakqoptbound},
$\qopt(\theta)\le\qbm(\theta)$ for all $\theta\in(\tlb,\theta^{m})$.
Now, assume, towards a contradiction, that $\qopt(\theta)<\qbm(\theta)$
for some $\theta\in(\tlb,\theta^{m})$. By continuity of $\qbm$ and
the fact that $\qopt$ is weakly decreasing, there exists a set of
types $I\subseteq(\tlb,\theta^{m})$ of positive Lebesgue measure
such that $\qopt(y)<\qbm(y)$ for all $y\in I$. Let $\widetilde{M}\equiv(\tilde{q},\tilde{u})$
be the mechanism with quantity schedule given by 
\begin{align*}
\tilde{q}(\theta)=\begin{cases}
\qbm(\theta) & \textrm{if}~\theta\in[\underline{\theta},\theta^{m})\\
\qopt(\theta) & \textrm{otherwise}
\end{cases}
\end{align*}
and $\tilde{u}(\theta)=\int\limits^{\tub}_{\theta}\tilde{q}(y){\rm d}y$
for all $\theta\in\Theta$. By Lemma \ref{lem:weakqoptbound}, $\qopt(\theta^{m})\le\qbm(\theta^{m})$,
and hence $\tilde{q}$ is weakly decreasing. Furthermore, because
$\tilde{u}$ satisfies the above envelope formula with $\tilde{u}(\tub)=0$,
$\widetilde{M}$ is IC and IR. Because $\qbm$ maximizes the objective
function of (\ref{opt:ropt}) point-wise and, for any $\theta$, $V^{\star}({\rm q})-z^{\star}(\theta){\rm q}$
in strictly concave in ${\rm q}$, and $I\subseteq(\tlb,\theta^{m})$
has positive Lebesgue measure, the objective function in (\ref{opt:ropt})
is strictly higher under $\tilde{q}$ than under $\qopt$. Below we
show that $\tilde{q}$ satisfies the robustness constraints in (\ref{opt:ropt}),
thus contradicting the robust optimality of $\qopt$. Clearly, $\tilde{q}$
satisfies the robustness constraints in (\ref{opt:ropt}) for all
$\theta\in[\theta^{m},\tub)$ because $\tilde{q}$ agrees with $\qopt$
on $[\theta^{m},\tub)$ and $\qopt$ satisfies these constraints by
its definition. Now, pick $\theta\in[\tlb,\theta^{m})$. Then, $\tilde{q}(\theta)=\qbm(\theta)=q^{\star}(\theta)$.
Hence, 
\begin{align*}
\underline{W}(\theta,\tilde{q})-\underline{W}(\theta,q^{\star}) & =\int\limits^{\tub}_{\theta^{m}}q^{\star}(y){\rm d}y-\int\limits^{\tub}_{\theta^{m}}\qopt(y){\rm d}y\\
 & \ge\Big[\underline{V}(\qopt(\theta^{m}))-\theta^{m}\qopt(\theta^{m})\Big]-\Big[\underline{V}(q^{\star}(\theta^{m}))-\theta^{m}q^{\star}(\theta^{m})\Big]\\
 & +\int\limits^{\tub}_{\theta^{m}}q^{\star}(y){\rm d}y-\int\limits^{\tub}_{\theta^{m}}\qopt(y){\rm d}y\\
 & =\underline{W}(\theta^{m},\qopt)-\underline{W}(\theta^{m},q^{\star}).
\end{align*}
The first equality follows from the definition of $\tilde{q}$. The
inequality follows because (a) $q^{\star}(\theta^{m})=\underline{D}(\theta^{m})$
(by Lemma \ref{lem:qstardlowint}) and (b) $\underline{V}({\rm q})-\theta^{m}{\rm q}$
is maximized at $\underline{D}(\theta^{m})$. The last equality is
by definition of the $\underline{W}$ function. As a consequence,
$\underline{W}(\theta,\tilde{q})-\underline{W}(\theta^{m},\qopt)\ge\underline{W}(\theta,q^{\star})-\underline{W}(\theta^{m},q^{\star})\ge0$,
where the inequality follows from the definition of $\theta^{m}$.
Hence, $\underline{W}(\theta,\tilde{q})\ge\underline{W}(\theta^{m},\qopt)\ge G^{*}$,
where the last inequality follows from the fact that $\qopt$ is a
solution to (\ref{opt:ropt}) and hence satisfies the robustness constraints
in (\ref{opt:ropt}) for all $\theta$. We conclude that $\underline{W}(\theta,\tilde{q})\ge G^{*}$
for all $\theta\in[\tlb,\tub]$, implying that $\tilde{q}$ satisfies
the robustness constraints in (\ref{opt:ropt}) for all $\theta\in\Theta$,
thus contradicting the robust optimality of $\qopt$. $\hfill\blacksquare{}$

\noindent\textsc{Part 2}. Let 
\begin{align*}
\theta^{\star\star} & =\inf\{\theta\in[\theta^{m},\tub]:\qopt(y)=\underline{D}(y)~\forall~y\in[\theta,\tub]\}.
\end{align*}
By Observation \ref{obs:qell}, $\qopt(\tub)={\rm q}_{\ell}=\underline{D}(\tub)$
and therefore $\theta^{\star\star}$ is well defined. Furthermore,
when $\theta^{\star}<\tub$, Lemma \ref{lem:weakqoptbound} along
with the assumption that $\underline{D}$ is decreasing implies that
$\theta^{\star\star}=\tub$. If $\theta^{\star\star}=\tlb$, then
by the definition of $\theta^{\star\star}$, $\theta^{m}=\tlb$ and
there is nothing to prove. Hence, we will assume that $\theta^{\star\star}>\tlb$.

\begin{lemma} \label{lem:along-D-low} If $\theta^{\star\star}>\underline{\theta}$,
then $\qopt(\theta^{\star\star})=\underline{D}(\theta^{\star\star})$.\end{lemma}

\noindent\textbf{Proof of Lemma \ref{lem:along-D-low}.} If $\theta^{\star\star}=\overline{\theta}$,
the result follows from the fact that $\qopt(\overline{\theta})=\underline{D}(\overline{\theta})$.
Thus suppose $\theta^{\star\star}<\overline{\theta}$. Because $\qopt$
is weakly decreasing and $\underline{D}$ is continuous, it cannot
be that $\qopt(\theta^{\star\star})<\underline{D}(\theta^{\star\star})$.
The same properties imply that, if $\qopt(\theta^{\star\star})>\underline{D}(\theta^{\star\star})$,
there exists $\epsilon>0$ such that $\qopt(\theta)>\underline{D}(\theta)$
for all $\theta\in[\theta^{\star\star}-\epsilon,\theta^{\star\star}]$.
This however implies that, for any $\theta\in[\theta^{\star\star}-\epsilon,\theta^{\star\star})$,
$\int\limits^{\tub}_{\theta}\qopt(y){\rm d}y>\int\limits^{\tub}_{\theta}\underline{D}(y){\rm d}y$
which, by virtue of Lemma \ref{lemma:majorizations}, contradicts
the assumption that $\qopt$ satisfies the robustness constraints.
Hence, $\qopt(\theta^{\star\star})=\underline{D}(\theta^{\star\star})$.
$\hfill\blacksquare{}$

If $\theta^{\star\star}=\theta^{m}$, the result in Part 2 thus follows
directly from Lemma \ref{lem:along-D-low}. If, instead, $\theta^{\star\star}>\theta^{m}$,
the result follows from the next lemma.

\begin{lemma} \label{lem:generictheta2} Suppose $M^{\star}$ is
not robustly optimal and $\theta^{\star\star}>\theta^{m}$. Then,
there exists a set of types $I\subseteq(\theta^{m},\min\{\theta^{\star},\theta^{\star\star}\})$
of positive Lebesgue measure such that $\qopt(\theta)<\qbm(\theta)$
for all $\theta\in I$. 
\end{lemma}

\noindent\textbf{Proof of Lemma \ref{lem:generictheta2}.} We consider
two cases.

\noindent\textsc{Case 1}: $\theta^{\star}<\overline{\theta}$. As
explained above, in this case, $\min\{\theta^{\star},\theta^{\star\star}\}=\theta^{\star}$.
Furthermore, by Lemma \ref{lemm-theta_m-star}, $\theta^{m}<\theta^{\star}$.
Lemma \ref{lem:weakqoptbound} in turn implies that, for all $\theta\in(\theta^{m},\theta^{\star})$,
$\qopt(\theta)\le q^{\star}(\theta)=\qbm(\theta)$, whereas, from
Lemma \ref{lem:leftofthetam}, $\qopt(\theta)=\qbm(\theta)$ for all
$\theta\in(\tlb,\theta^{m})$. Now suppose, towards a contradiction,
that $\qopt(\theta)=\qbm(\theta)$ for almost all $\theta\in(\tlb,\theta^{\star})$.
Because $\qbm$ is continuous and $\qopt$ and $\qbm$ are weakly
decreasing, it must be that $\qopt(\theta)=\qbm(\theta)$ for all
$\theta\in(\tlb,\theta^{\star})$. Lemma \ref{lem:weakqoptbound}
then implies that $\qopt(\theta)=q^{\star}(\theta)$ for all $\theta>\tlb$.
Furthermore, because $M^{\star}$ is not robustly optimal, $\underline{W}(\theta^{m},q^{\star})<G^{*}$.
Because $\theta^{m}<\theta^{\star}$ and $\underline{W}(\cdot,q^{\star})$
is continuous, there must exist some $\theta'>\theta^{m}$ such that
$\underline{W}(\theta',q^{\star})<G^{*}$. Because $\qopt(\theta)=q^{\star}(\theta)$
for all $\theta>\tlb$, $\underline{W}(\theta',\qopt)=\underline{W}(\theta',q^{\star})<G^{*}$,
which contradicts the assumption that $\qopt$ satisfies the robustness
constraints of (\ref{opt:ropt}). We thus conclude that there must
exist a set of types $I\subseteq(\theta^{m},\theta^{\star})$ of positive
Lebesgue measure such that $\qopt(\theta)<\qbm(\theta)$ for all $\theta\in I$,
as claimed.

\noindent\textsc{Case 2}: $\theta^{\star}=\overline{\theta}$. In
this case, $\min\{\theta^{\star},\theta^{\star\star}\}=\theta^{\star\star}$.
Again, assume towards a contradiction that $\qopt(\theta)=\qbm(\theta)$
for almost all $\theta\in(\theta^{m},\theta^{\star\star})$. 

\noindent First, suppose that $\theta^{m}>\tlb$. Lemma \ref{lem:leftofthetam}
then implies that $\qopt(\theta)=\qbm(\theta)$ for all $\theta\in(\tlb,\theta^{m})$.
Because $\qbm$ is continuous and $\qopt$ is weakly decreasing, the
assumption that $\qopt(\theta)=\qbm(\theta)$ for almost all $\theta\in(\theta^{m},\theta^{\star\star})$
then implies that $\qopt(\theta^{m})=\qbm(\theta^{m})$. Next, suppose
that $\theta^{m}=\underline{\theta}$. Again, because $\qbm$ is continuous
and $\qopt$ is weakly decreasing, the assumption that $\qopt(\theta)=\qbm(\theta)$
for almost all $\theta\in(\theta^{m},\theta^{\star\star})$ implies
that $\qopt(\tlb)\ge\qbm(\tlb)$. Lastly, recall that the assumption
that $\theta^{\star}=\tub$ implies that $\qbm(\theta)=q^{\star}(\theta)$
for all $\theta\in(\tlb,\theta^{\star\star})$. The assumption that
$\qopt(\theta)=\qbm(\theta)$ for almost all $\theta\in(\theta^{m},\theta^{\star\star})$
then implies that 
\begin{align}
\underline{W}(\theta^{m},\qopt)-\underline{W}(\theta^{m},q^{\star}) & =\Big[\underline{V}(\qopt(\theta^{m}))-\theta^{m}\qopt(\theta^{m})\Big]-\Big[\underline{V}(\qbm(\theta^{m}))-\theta^{m}\qbm(\theta^{m})\Big]\label{eq:diff}\\
 & +\int\limits^{\tub}_{\theta^{m}}\big(q^{\star}(y)-\qopt(y)\big){\rm d}y\nonumber \\
 & =\Big[\underline{V}(\qopt(\theta^{m}))-\theta^{m}\qopt(\theta^{m})\Big]-\Big[\underline{V}(\qbm(\theta^{m}))-\theta^{m}\qbm(\theta^{m})\Big]\nonumber \\
 & +\int\limits^{\tub}_{\theta^{\star\star}}\big(q^{\star}(y)-\qopt(y)\big){\rm d}y.\nonumber 
\end{align}
If $\theta^{m}>\tlb$, the first two bracketed terms in (\ref{eq:diff})
cancel out because, in this case, $\qbm(\theta^{m})=\qopt(\theta^{m})$,
as established above. If, instead, $\theta^{m}=\tlb$, because $\qbm(\tlb)=D^{\star}(\tlb)\ge\underline{D}(\tlb)$,
and because, as established above, $\qopt(\tlb)\ge\qbm(\tlb)$, by
the concavity of $\underline{V}({\rm q})-\theta^{m}{\rm q}$ in ${\rm q}$,
the sum of the two square brackets in (\ref{eq:diff}) is non-positive.
Hence, no matter whether $\theta^{m}>\tlb$ or $\theta^{m}=\tlb$,
in either case, we have that 
\begin{align}
\underline{W}(\theta^{m},\qopt)-\underline{W}(\theta^{m},q^{\star}) & \le\int\limits^{\tub}_{\theta^{\star\star}}\big(q^{\star}(y)-\qopt(y)\big){\rm d}y.\label{eq:temp1}
\end{align}
By Lemma \ref{lem:along-D-low}, $\qopt(\theta^{\star\star})=\underline{D}(\theta^{\star\star})$.
Hence, 
\begin{align}
\Big[\underline{V}(\qopt(\theta^{\star\star}))-\theta^{\star\star}\qopt(\theta^{\star\star})\Big]-\Big[\underline{V}(\qbm(\theta^{\star\star}))-\theta^{\star\star}\qbm(\theta^{\star\star})\Big] & \ge0,\label{eq:new-ineq}
\end{align}
where the inequality follows from the fact that $\underline{D}(\theta^{\star\star})$
maximizes $\underline{V}({\rm q})-\theta^{\star\star}{\rm q}$ over
$[0,{\rm \bar{q}}]$. Combining (\ref{eq:new-ineq}) with (\ref{eq:temp1}),
we thus have that 
\begin{align*}
\underline{W}(\theta^{m},\qopt)-\underline{W}(\theta^{m},q^{\star}) & \le\Big[\underline{V}(\qopt(\theta^{\star\star}))-\theta^{\star\star}\qopt(\theta^{\star\star})-\int\limits^{\tub}_{\theta^{\star\star}}\qopt(y){\rm d}y\Big]\\
 & -\Big[\underline{V}(\qbm(\theta^{\star\star}))-\theta^{\star\star}\qbm(\theta^{\star\star})-\int\limits^{\tub}_{\theta^{\star\star}}q^{\star}(y){\rm d}y\Big]\\
 & =\underline{W}(\theta^{\star\star},\qopt)-\underline{W}(\theta^{\star\star},q^{\star}).
\end{align*}
Because $\theta^{m}<\theta^{\star\star}$, $\underline{W}(\theta^{m},q^{\star})<\underline{W}(\theta^{\star\star},q^{\star})$.
We thus have that $\underline{W}(\theta^{m},\qopt)<\underline{W}(\theta^{\star\star},\qopt)$.
The definition of $\theta^{\star\star}$, along with Lemma \ref{lem:along-D-low},
in turn imply that $\qopt(\theta)=\underline{D}(\theta)$ for all
$\theta\ge\theta^{\star\star}$. Because $\underline{W}(\tub,\qopt)=G^{*}$,
this last property in turn implies that $\underline{W}(\theta^{\star\star},\qopt)=G^{*}$.
Hence, $\underline{W}(\theta^{m},\qopt)<G^{*}$, a contradiction to
the robust optimality of $\qopt$. We thus conclude that there must
exist a set of types $I\subseteq(\theta^{m},\theta^{\star\star})$
of positive Lebesgue measure such that $\qopt(\theta)<\qbm(\theta)$
for all $\theta\in I$. $\hfill\blacksquare{}$

\noindent This completes the proof of Proposition \ref{prop:robust-quantity-mechanism-general}.
\hfill{}Q.E.D. \smallskip{}

\subsection{Proofs of Section \ref{sec:price_vs_quantity}}
\label{sec:proofsec4}

\noindent\textbf{Proof of Lemma }\ref{lemma-SL-price-mechanisms}.
First, we show that $G(\widetilde{M})\le G^{*}$ for any $\widetilde{M}\equiv(p,t)\in\widetilde{\mathcal{M}}$.
To establish this, we represent the mechanism as $\widetilde{M}\equiv(p,\tilde{u})$,
where $\tilde{u}(\theta,D)\equiv t(\theta,D)-\theta D(p(\theta))$
for all $\theta\in\Theta$ and $D\in\mathcal{D}$. Observe that, for
any $\widetilde{M}\in\mathcal{\widetilde{M}}$, 
\begin{align*}
G(\widetilde{M})\le_{(i)}\underline{V}(\underline{D}(p(\overline{\theta})))-\overline{\theta}\underline{D}(p(\overline{\theta}))-\tilde{u}(\overline{\theta},\underline{D}) & \le_{(ii)}\underline{V}(\underline{D}(p(\overline{\theta})))-\overline{\theta}\underline{D}(p(\overline{\theta}))\\
 & \le_{(iii)}\underline{V}(\underline{D}(\overline{\theta}))-\overline{\theta}\underline{D}(\overline{\theta})=G^{*}.
\end{align*}
Inequality $(i)$ holds because the right-hand-side is just expected
welfare under a Dirac distribution that puts probability one at $\overline{\theta}$,
when the demand is $\underline{D}$. Inequality $(ii)$ holds because
$\tilde{u}(\overline{\theta},\underline{D})\ge0$ as $\widetilde{M}$
is EPIR. Inequality $(iii)$ holds because $\overline{\theta}=\arg\max_{\text{p}}\left\{ \underline{V}(\underline{D}(\text{p}))-\overline{\theta}\underline{D}(\text{p})\right\} $.

Next, to prove that $G(\widetilde{M})=G^{*}$ for any $\widetilde{M}\in\mathcal{\widetilde{M}}^{{\rm SL}}$,
it suffices to note that there exists a mechanism $\underline{\widetilde{M}}\in\mathcal{\widetilde{M}}$
such that $G(\underline{\widetilde{M}})=G^{*}$. Let $\underline{\widetilde{M}}\equiv(p,t)$
be the price mechanism such that $p(\theta)=\overline{\theta}$ for
all $\theta\in\Theta$, and where $t(\theta,D)$ satisfies Condition
(\ref{eq:rent-price-mechanism}) with $\tilde{u}(\overline{\theta},D)=0$
for all $\theta\in\Theta$, all $D\in\mathcal{D}$. Under such a mechanism,
for all $\theta\in\Theta$, all $D\in\mathcal{D}$, welfare is equal
to 
\begin{align*}
V(D(\overline{\theta}))-\theta D(\overline{\theta})-\int\limits^{\overline{\theta}}_{\theta}D(\overline{\theta})dy & =V(D(\overline{\theta}))-\overline{\theta}D(\overline{\theta})\ge\underline{V}(\underline{D}(\overline{\theta}))-\overline{\theta}\underline{D}(\overline{\theta})=G^{*}.
\end{align*}
This means that, for all $F\in\mathcal{F}$ and $D\in\mathcal{D}$,
$\widetilde{W}(\underline{\widetilde{M}};D,F)\ge G^{*}$, and hence,
$G(\underline{\widetilde{M}})=G^{*}$.

\noindent We now prove that, if $\widetilde{M}\equiv(p,\tilde{u})\in\mathcal{\widetilde{M}}^{{\rm SL}}$,
Conditions (a)-(d) in the lemma must hold. Because $\widetilde{M}$
is EPIC and EPIR, $p$ is weakly increasing and $\tilde{u}$ satisfies
Condition (\ref{eq:rent-price-mechanism}) for all $\theta\in\Theta$
and $D\in\mathcal{D}$, with $\tilde{u}(\overline{\theta},D)\ge0$.
That $p(\overline{\theta})=\overline{\theta}$ and $\tilde{u}(\overline{\theta},\underline{D})=0$
follows from the fact that the only way welfare can be equal to $G^{*}$
when Nature selects $D=\underline{D}$ and a technology that selects
$\theta=\overline{\theta}$ with probability one is by inducing efficient
output by setting a price $p(\overline{\theta})=\overline{\theta}$
and giving no rent to the monopolist, which amounts to setting $\tilde{u}(\overline{\theta},\underline{D})=0$.
That $\tilde{w}(\theta,\widetilde{M};D)\ge G^{*}$ must hold for all
$\theta\in\Theta$ and all $D\in\mathcal{D}$ follows from the fact that,
if this is not the case, then $G(\widetilde{M})<G^{*}$, a contradiction
to $\widetilde{M}\in\mathcal{\widetilde{M}}^{{\rm SL}}$.

\noindent Finally, to see that, jointly, Conditions (a)-(d) imply
that $\widetilde{M}\in\mathcal{\widetilde{M}}^{{\rm SL}}$, observe
that Conditions (a)-(c) imply that $\widetilde{M}$ is EPIC and EPIR.
Furthermore, Condition (d) implies that $G(\widetilde{M})\ge G^{*}$.
Since $G(\widetilde{M})\le G^{*}$, we conclude $G(\widetilde{M})=G^{*}$.
The first part of the lemma then implies that $\widetilde{M}\in\mathcal{\widetilde{M}}^{{\rm SL}}$.\hfill{}Q.E.D.
\smallskip{}

\noindent\textbf{Proof of Proposition \ref{prop:robustly-optimal-price-mechanism}.}
By standard arguments, for any EPIC and EPIR price regulation $\widetilde{M}\equiv(p,t)$,
expected welfare under the conjectured model $(D^{\star},F^{\star})$
is equal to 
\begin{align*}
\int\limits^{\overline{\theta}}_{\underline{\theta}}\Big[V^{\star}(D^{\star}(p(\theta)))-z^{\star}(\theta)D^{\star}(p(\theta))\Big]F^{\star}(\textrm{d}\theta)-\widetilde{u}(\overline{\theta},D^{\star}).
\end{align*}
Any robustly optimal price regulation thus maximizes this expression
subject to Conditions (a)-(d) in Lemma \ref{lemma-SL-price-mechanisms}.

\noindent Consider a relaxed program where (1) Conditions (a) and
(c) in Lemma \ref{lemma-SL-price-mechanisms} are replaced by the
requirement that $p(\theta)\le\overline{\theta}$ for all $\theta\in\Theta$,
(2) the constraints on ${\{\tilde{u}(\theta,D)\}}_{D\in\mathcal{D}}$
in Condition (b) in Lemma \ref{lemma-SL-price-mechanisms} are replaced
by the requirement that $\widetilde{u}(\overline{\theta},D^{\star})\ge0$;
and (3) the constraints in Condition (d) in Lemma \ref{lemma-SL-price-mechanisms}
are dropped. Any solution to this relaxed program is such that $\widetilde{u}(\overline{\theta},D^{\star})=0$
and, for any $\theta>\underline{\theta}$, $p(\theta)$ is as in the
Baron-Myerson-with-price-cap regulation. This is because, for any
$\theta\in\Theta$, the expression $V^{\star}(D^{\star}(\text{p}))-z^{\star}(\theta)D^{\star}(\text{p})$
is quasi-concave in $\text{p}$. The unique maximizer of this expression
over $[0,\overline{\theta}]$ is thus $p(\theta)=\min\{z^{\star}(\theta),\overline{\theta}\}$.
Because the function $p:\Theta\rightarrow\mathbb{R}$ given by $p(\theta)=\min\{z^{\star}(\theta),\overline{\theta}\}$
for all $\theta\in\Theta$ is weakly increasing, any solution to the
relaxed program is such that $p(\theta)=\min\{z^{\star}(\theta),\overline{\theta}\}$
for all $\theta>\underline{\theta}$.

Equipped with this result, we now show that any mechanism $\widetilde{M}\equiv(p,\tilde{u})$
in which (A) $p$ is weakly increasing and such that $p(\theta)=\min\{z^{\star}(\theta),\overline{\theta}\}$
for all $\theta>\underline{\theta}$, and (B) $\tilde{u}$ satisfies
the conditions in Definition \ref{def:BM-with-price-cap} satisfies
all the conditions of Lemma \ref{lemma-SL-price-mechanisms}, implying
that $\widetilde{M}\in\mathcal{\widetilde{M}}^{{\rm SL}}$. First
note that Condition (b) in Lemma \ref{lemma-SL-price-mechanisms}
trivially holds. Because $z^{\star}$ is increasing, Condition (a)
is also satisfied. Because $z^{\star}(\overline{\theta})>\overline{\theta}$,
Condition (c) also holds. To complete the proof, it thus suffices
to show that, for any $D\in\mathcal{D}$, Condition (d) holds. That
$\tilde{u}$ satisfies the conditions in Definition \ref{def:BM-with-price-cap}
implies that $\widetilde{w}(\overline{\theta},\widetilde{M}^{\textsc{}};D)\ge G^{*}$.
Now, pick any $\theta<\overline{\theta}$. If $z^{\star}(\theta)\ge\overline{\theta}$,
then $\widetilde{w}(\theta,\widetilde{M}^{\textsc{}};D)=\widetilde{w}(\overline{\theta},\widetilde{M}^{\textsc{}};D)\ge G^{*}$.
If, instead, $z^{\star}(\theta)<\overline{\theta}$, then $p(\theta)=z^{\star}(\theta)<\overline{\theta}$.
That $z^{\star}(\theta)\geq\theta$ implies that $D(z^{\star}(\theta))\le D(\theta)$.
Arguments identical to those in the proof of Lemma \ref{Lem-Mono}
but applied to the quantity schedule $q=D(p^{\star}(\cdot))$ and
with welfare evaluated under the demand $D$ instead of the demand
$\underline{D}$ then imply that the welfare function $\widetilde{w}(\cdot,\widetilde{M}^{\textsc{}};D)$
is weakly decreasing in $\theta$. Thus, $\widetilde{w}(\theta,\widetilde{M}^{\textsc{}};D)\ge\widetilde{w}(\overline{\theta},\widetilde{M};D)\ge G^{*}$.
This means that Condition (d) in Lemma \ref{lemma-SL-price-mechanisms}
also holds. 

Jointly the above properties imply the result in the proposition.\hfill{}Q.E.D.
\smallskip{}

\noindent\textbf{Proof of Proposition \ref{prop:comp_reg}. }The
proof is in two parts, each establishing the corresponding claim in
the proposition.

\noindent\textbf{Part 1.} Suppose $M^{\star}$ is robustly optimal.
Corollary \ref{cor:qstaropt} then implies that in any robustly optimal
quantity regulation $M^{\textsc{OPT}}$, $\qopt(\theta)=\max\{\qbm(\theta),{\rm q}_{\ell}\}$
for all $\theta>\underline{\theta}$. Under any robustly optimal price
regulation $\widetilde{M}^{\textsc{OPT}}=(\popt,\widetilde{u}^{\textsc{OPT}})$,
when the demand is $D^{\star}$ and the cost is $\theta>\underline{\theta}$,
the monopolist sells a quantity $D^{\star}(\popt(\theta))=\max\{\qbm(\theta),\tilde{{\rm q}}_{\ell}\}$,
where $\tilde{{\rm q}}_{\ell}\equiv D^{\star}(\overline{\theta})\ge\underline{D}(\overline{\theta})={\rm q}_{\ell}$.
Hence, there exists $\tilde{\theta}^{\star}\le\theta^{\star}$ such
that $D^{\star}(\popt(\theta))=\tilde{{\rm q}}_{\ell}$ if $\theta\ge\tilde{\theta}^{\star}$
and $D^{\star}(\popt(\theta))=\qbm(\theta)$ if $\theta\in(\underline{\theta},\tilde{\theta}^{\star})$.
See Panel A of Figure \ref{fig:comp_reg} for the illustration.

For $\theta\in(\underline{\theta},\tilde{\theta}^{\star})$, we have
that $D^{\star}(\popt(\theta))=\qopt(\theta)=\qbm(\theta)$. However,
for $\theta\ge\tilde{\theta}^{\star}$, we have that $D^{\star}(\popt(\theta))=\tilde{{\rm q}}_{\ell}\ge\qopt(\theta)\ge\qbm(\theta)$.
Because, for any $\theta$, virtual surplus $V^{\star}({\rm q})-z^{\star}(\theta){\rm q}$
is quasi-concave in ${\rm q}$, reaching a maximum at $\qbm(\theta)$,
we thus have that, for any $\theta\ge\tilde{\theta}^{\star}$, $V^{\star}(D^{\star}(\popt(\theta)))-z^{\star}(\theta)D^{\star}(\popt(\theta))\leq V^{\star}(\qopt(\theta))-z^{\star}(\theta)\qopt(\theta)$.
Because $F^{\star}$ does not have an atom at $\theta=\underline{\theta}$,
and because $\widetilde{u}^{\textsc{OPT}}(\overline{\theta},D^{\star})=u^{\textsc{OPT}}(\overline{\theta})=0$,
we conclude that $\widetilde{W}(\widetilde{M}^{\textsc{OPT}};D^{\star},F^{\star})\le W(M^{\textsc{OPT}};V^{\star},F^{\star})$,
with the inequality strict if, and only if, $D^{\star}(\overline{\theta})>\underline{D}(\overline{\theta})$.

\noindent\textbf{Part 2.} If $D^{\star}(\overline{\theta})=\underline{D}(\overline{\theta})$,
then ${\rm q}_{\ell}=\tilde{{\rm q}}_{\ell}$. In this case, for any
$\theta>\underline{\theta}$, the quantity traded under the conjectured
model $(D^{\star},F^{\star})$ when running any robustly optimal price
regulation is $D^{\star}(\popt(\theta))=\max\{\qbm(\theta),{\rm q}_{\ell}\}=q^{\star}(\theta)$.
Again, because $F^{\star}$ does not have any atom at $\underline{\theta}$,
this means that, by running any robustly optimal price regulation
$\widetilde{M}^{\textsc{OPT}}$, the regulator obtains the same welfare
as by running the Baron-Myerson-with-quantity-floor regulation $M^{\star}$,
i.e., 
\begin{align}
\widetilde{W}(\widetilde{M}^{\textsc{OPT}};D^{\star},F^{\star})=W(M^{\star};V^{\star},F^{\star}).\label{eq:welfare_comparison-BMF-vs-Price}
\end{align}
As shown in the proof of Proposition \ref{prop:BM-floor}, $M^{\star}$
is the solution to a relaxation of the full program yielding the robustly
optimal quantity regulation, implying that 
\begin{equation}
W(M^{\star};V^{\star},F^{\star})\geq W(M^{\textsc{OPT}};V^{\star},F^{\star}).\label{BMF-vs-quantity-robust-mech}
\end{equation}
When $M^{\star}$ is not robustly optimal, the inequality in (\ref{BMF-vs-quantity-robust-mech})
is strict. 

\noindent Together, (\ref{eq:welfare_comparison-BMF-vs-Price}) and
(\ref{BMF-vs-quantity-robust-mech}) imply that, when $D^{\star}(\overline{\theta})=\underline{D}(\overline{\theta})$,
price regulation dominates quantity regulation, strictly if $M^{\star}$
is not robustly optimal. \hfill{}Q.E.D.\smallskip{}

\noindent\textbf{Proof of Proposition \ref{prop:primitives}.} The
proof is in two parts, each establishing the corresponding claim in
the proposition.

\noindent\textbf{Part (a)}. The bi-Lipschitz continuity of $D^{\star}$
implies that, for any $y\in\Theta$, 
\begin{align*}
D^{\star}(y)-\qbm(y)=D^{\star}(y)-D^{\star}\left(y+\frac{F^{\star}(y)}{f^{\star}(y)}\right) & \geq k\frac{F^{\star}(y)}{f^{\star}(y)}\ge\delta(y)=D^{\star}(y)-\underline{D}(y).
\end{align*}
This implies that $\underline{D}(y)\geq\qbm(y)$ for all $y\in\Theta$.
Decreasingness of $\underline{D}$ in turn implies that $\underline{D}(y)\ge\underline{D}(\overline{\theta})={\rm q}_{\ell}$.
Hence, for all $y\in\Theta$, $q^{\star}(y)\leq\underline{D}(y)$.
Condition (\ref{eq:weak-majorization}) in Proposition \ref{prop:BM-floor}
thus holds for all $\theta\in\Theta$. Furthermore, because, for any
$\theta\in\Theta$, $\delta(\theta)\geq0$, the assumption in part
(a) of the proposition that $F^{\star}(\theta)/f^{\star}(\theta)\geq\delta(\theta)/k$
implies that $\delta(\underline{\theta})=0$. Hence, $q^{\star}(\underline{\theta})=\qbm(\underline{\theta})=D^{\star}(\underline{\theta})=\underline{D}(\underline{\theta})$.
In turn, this implies that Condition (\ref{eq:robustness-bottom})
in Proposition \ref{prop:BM-floor} reduces to Condition (\ref{eq:weak-majorization})
evaluated at $\theta=\underline{\theta}$. We conclude that $M^{\star}$
is robustly optimal. The result in part (a) then follows from Proposition
\ref{prop:comp_reg}.

\noindent\textbf{Part (b)}. The proof is in three steps.

\noindent\textsc{Step 1.} Let $\Delta\equiv\int\limits^{\overline{\theta}}_{\underline{\theta}}\Big(\delta(\theta)-K\frac{F^{\star}(\theta)}{f^{\star}(\theta)}\Big){\rm d}\theta$.
The bi-Lipschitz continuity of $D^{\star}$ implies that, for any
$\theta\in\Theta$, 
\begin{align}
\qbm(\theta)=D^{\star}\left(\theta+\frac{F^{\star}(\theta)}{f^{\star}(\theta)}\right)\ge D^{\star}(\theta)-K\frac{F^{\star}(\theta)}{f^{\star}(\theta)}=\underline{D}(\theta)+\delta(\theta)-K\frac{F^{\star}(\theta)}{f^{\star}(\theta)}.\label{eq:comp_lemma_1}
\end{align}
Hence, 
\begin{align}
\int^{\overline{\theta}}_{\underline{\theta}}\bigl(\qbm(\theta)-\qopt(\theta)\bigr){\rm d}\theta & \ge\int^{\overline{\theta}}_{\underline{\theta}}\bigl(\underline{D}(\theta)-\qopt(\theta)\bigr){\rm d}\theta+\Delta\geq\Delta,\label{eq:LB-1}
\end{align}
where the last inequality holds because of Lemma \ref{lemma:majorizations}.

\noindent\textsc{Step 2.} We now use the lower bound in Step 1 to
establish a lower bound on $W(M^{{\rm BM}};V^{\star},F^{\star})-W(M^{{\rm OPT}};V^{\star},F^{\star})$.
Before proceeding with the proof of this Step, we establish a technical
lemma that we use below.

\noindent\begin{lemma}[strong concavity]\label{lem:strong-concavity}
If $D^{\star}$ is $K$-Lipschitz continuous on $[P^{\star}({\rm \bar{q})},P^{\star}(0)]$,
then $V^{\star}$ is $(1/K)$-strongly concave on $[0,{\rm \bar{q}}]$,
in the sense that, for all ${\rm q}_{1},{\rm q}_{2}\in[0,{\rm \bar{q}}]$,
$V^{\star}({\rm q}_{2})-V^{\star}({\rm q}_{1})+P^{\star}({\rm q}_{2})({\rm q}_{1}-{\rm q}_{2})\ge\frac{1}{2K}({\rm q}_{1}-{\rm q}_{2})^{2}$.
\end{lemma}

\noindent\textbf{Proof of Lemma \ref{lem:strong-concavity}}. By
the $K$-Lipschitz continuity of $D^{\star}$, for any ${\rm q}_{1},{\rm q}_{2}\in[0,{\rm \bar{q}}]$,
we have that $|{\rm q}_{1}-{\rm q}_{2}|=|D^{\star}(P^{\star}({\rm q}_{1}))-D^{\star}(P^{\star}({\rm q}_{2}))|\le K|P^{\star}({\rm q}_{1})-P^{\star}({\rm q}_{2})|$.
Because $P^{\star}$ is decreasing, irrespective of the sign of $({\rm q}_{1}-{\rm q}_{2})$,
\begin{equation}
\bigl(P^{\star}({\rm q}_{2})-P^{\star}({\rm q}_{1})\bigr)({\rm q}_{1}-{\rm q}_{2})\ge\frac{1}{K}({\rm q}_{1}-{\rm q}_{2})^{2}.\label{eq:strong-monotonicity-P}
\end{equation}
Fix ${\rm q}_{1},{\rm q}_{2}\in[0,\bar{{\rm q}}]$, and let $\phi(t)\equiv V^{\star}({\rm q}_{2}+t({\rm q}_{1}-{\rm q}_{2}))$
for all $t\in[0,1]$. Then, for all $t\in[0,1]$, $\phi'(t)=P^{\star}({\rm q}_{2}+t({\rm q}_{1}-{\rm q}_{2}))({\rm q}_{1}-{\rm q}_{2}).$
Furthermore, using (\ref{eq:strong-monotonicity-P}), we have that
\[
\Big(P^{\star}({\rm q}_{2})-P^{\star}({\rm q}_{2}+t({\rm q}_{1}-{\rm q}_{2}))\Big)t({\rm q}_{1}-{\rm q}_{2})\ge\frac{1}{K}t^{2}({\rm q}_{1}-{\rm q}_{2})^{2}
\]
which implies that 
\begin{align}
\phi'(t)=P^{\star}({\rm q}_{2}+t({\rm q}_{1}-{\rm q}_{2}))({\rm q}_{1}-{\rm q}_{2}) & \le P^{\star}({\rm q}_{2})({\rm q}_{1}-{\rm q}_{2})-\frac{1}{K}t({\rm q}_{1}-{\rm q}_{2})^{2}.\label{eq:eee1-1}
\end{align}
Using (\ref{eq:eee1-1}), we thus have that 
\begin{align*}
V^{\star}({\rm q}_{1})-V^{\star}({\rm q}_{2})=\phi(1)-\phi(0)=\int^{1}_{0}\phi'(t){\rm d}t & \le P^{\star}({\rm q}_{2})({\rm q}_{1}-{\rm q}_{2})-\frac{1}{K}({\rm q}_{1}-{\rm q}_{2})^{2}\int^{1}_{0}t{\rm d}t\\
 & =P^{\star}({\rm q}_{2})({\rm q}_{1}-{\rm q}_{2})-\frac{1}{2K}({\rm q}_{1}-{\rm q}_{2})^{2},
\end{align*}
from which we obtain that $V^{\star}({\rm q}_{2})-V^{\star}({\rm q}_{1})+P^{\star}({\rm q}_{2})({\rm q}_{1}-{\rm q}_{2})\ge\frac{1}{2K}({\rm q}_{1}-{\rm q}_{2})^{2}$.
\hfill{} $\blacksquare$

Now, for each $\theta$ and each ${\rm q}$, let 
\[
\psi({\rm q};\theta)\equiv V^{\star}({\rm q})-\left(\theta+\frac{F^{\star}(\theta)}{f^{\star}(\theta)}\right){\rm q}.
\]
Observe that, for any $\theta\in\Theta$, $\psi(\cdot;\theta)$ reaches
a maximum at $\qbm(\theta)$. Furthermore, 
\begin{align}
\psi(\qbm(\theta);\theta)-\psi(\qopt(\theta);\theta) & =V^{\star}(\qbm(\theta))-V^{\star}(\qopt(\theta))+\left(\theta+\frac{F^{\star}(\theta)}{f^{\star}(\theta)}\right)\left(\qopt(\theta)-\qbm(\theta)\right)\nonumber \\
 & \ge\frac{1}{2K}\bigl(\qbm(\theta)-\qopt(\theta)\bigr)^{2},\label{eq:intermezzo}
\end{align}
where the inequality follows from the fact that $P^{\star}(\qbm(\theta))=\theta+\frac{F^{\star}(\theta)}{f^{\star}(\theta)}$
along with Lemma \ref{lem:strong-concavity}. Hence, 
\begin{align}
 & W(M^{{\rm BM}};V^{\star},F^{\star})-W(M^{{\rm OPT}};V^{\star},F^{\star})=\int\limits^{\overline{\theta}}_{\underline{\theta}}\left(\psi(\qbm(\theta);\theta)-\psi(\qopt(\theta);\theta)\right)f^{\star}(\theta){\rm d}\theta\nonumber \\
 & \ge\frac{\underline{f}}{2K}\int\limits^{\overline{\theta}}_{\underline{\theta}}\bigl(\qbm(\theta)-\qopt(\theta)\bigr)^{2}{\rm d}\theta\ge\frac{\underline{f}}{2K(\overline{\theta}-\underline{\theta})}\Big(\int\limits^{\overline{\theta}}_{\underline{\theta}}\bigl(\qbm(\theta)-\qopt(\theta)\bigr){\rm d}\theta\Big)^{2}\geq\frac{\Delta^{2}\underline{f}}{2K(\overline{\theta}-\underline{\theta})},\label{eq:LB-2}
\end{align}
where the first inequality follows from (\ref{eq:intermezzo}) and
the definition of $\underline{f}$, the second inequality follows
from Jensen's inequality, and the third inequality follows from (\ref{eq:LB-1}).

\noindent\textsc{Step 3.} Finally, we compare the welfare generated
by any price regulation $\widetilde{M}^{{\rm OPT}}\equiv(\popt,\widetilde{u}^{\textsc{OPT}})$
to that generated by the quantity regulation $M^{{\rm BM}}\equiv(\qbm,u^{{\rm BM}})$,
under the designer's model. For any $\theta>\underline{\theta}$,
let $\tilde{q}(\theta)=D^{\star}\bigl(\min\{z^{\star}(\theta),\overline{\theta}\}\bigr)$
be the quantity traded under any robustly optimal price regulation
$\widetilde{M}^{{\rm OPT}}$, when the demand is the one in the designer's
model. For every $\theta>\underline{\theta}$ such that $z^{\star}(\theta)\le\overline{\theta}$,
$\tilde{q}(\theta)=\qbm(\theta)$, implying that 
\begin{align}
\psi(\qbm(\theta);\theta)-\psi(\tilde{q}(\theta);\theta)=0.\label{eq:t1}
\end{align}

\noindent On the other hand, for any $\theta>\underline{\theta}$
such that $z^{\star}(\theta)>\overline{\theta}$, $\tilde{q}(\theta)=D^{\star}(\overline{\theta})\geq\qbm(\theta)$
and hence 
\begin{align}
\psi(\qbm(\theta);\theta)-\psi(\tilde{q}(\theta);\theta) & =V^{\star}(\qbm(\theta))-V^{\star}(D^{\star}(\overline{\theta}))+z^{\star}(\theta)[D^{\star}(\overline{\theta})-\qbm(\theta)].\label{eq:ea11}
\end{align}
Because $V^{\star}({\rm q})-\overline{\theta}{\rm q}$ is maximized
at $D^{\star}(\overline{\theta})$, for any $\theta$, 
\begin{align}
V^{\star}(D^{\star}(\overline{\theta}))-\overline{\theta}D^{\star}(\overline{\theta}) & \ge V^{\star}(\qbm(\theta))-\overline{\theta}\qbm(\theta).\label{eq:ea22}
\end{align}
Combining (\ref{eq:ea11}) with (\ref{eq:ea22}), we have that, for
any $\theta>\underline{\theta}$ such that $z^{\star}(\theta)>\overline{\theta}$,
\begin{align*}
 & \psi(\qbm(\theta);\theta)-\psi(\tilde{q}(\theta);\theta)\le(z^{\star}(\theta)-\overline{\theta})[D^{\star}(\overline{\theta})-\qbm(\theta)]\le(z^{\star}(\theta)-\overline{\theta})[D^{\star}(\theta)-\qbm(\theta)]\\
 & \le(z^{\star}(\theta)-\theta)\left[D^{\star}(\theta)-D^{\star}\left(\theta+\frac{F^{\star}(\theta)}{f^{\star}(\theta)}\right)\right]\le K\Big(\frac{F^{\star}(\theta)}{f^{\star}(\theta)}\Big)^{2}\leq\frac{\delta(\theta)^{2}}{KN^{2}},
\end{align*}
where $N\equiv1+\frac{\bar{\delta}}{\underline{\delta}}\sqrt{\frac{2\bar{f}}{\underline{f}}}$.
The penultimate inequality follows from the bi-Lipschitz continuity
of $D^{\star}$ whereas the last inequality follows from the fact
that, by assumption, $F^{\star}(\theta)/f^{\star}(\theta)\leq\delta(\theta)/(KN)$.
By (\ref{eq:t1}), this inequality trivially holds when $z^{\star}(\theta)\le\tub$.
As a consequence for any $\theta>\underline{\theta}$,
\begin{align*}
\psi(\qbm(\theta);\theta)-\psi(\tilde{q}(\theta);\theta)\le\frac{\delta(\theta)^{2}}{KN^{2}}.
\end{align*}
Integrating over all types, we thus have that 
\begin{equation}
\begin{array}{c}
W(M^{{\rm BM}};V^{\star},F^{\star})-\widetilde{W}(\widetilde{M}^{{\rm OPT}};D^{\star},F^{\star})=\int\limits^{\overline{\theta}}_{\underline{\theta}}\Big(\psi(\qbm(\theta);\theta)-\psi(\tilde{q}(\theta);\theta)\Big)f^{\star}(\theta){\rm d}\theta\\
\le\frac{1}{KN^{2}}\int\limits^{\overline{\theta}}_{\underline{\theta}}\delta(\theta)^{2}f^{\star}(\theta){\rm d}\theta\le\frac{\bar{f}(\overline{\theta}-\underline{\theta})\bar{\delta}^{2}}{KN^{2}}.
\end{array}\label{eq:ub2}
\end{equation}
Combining (\ref{eq:LB-2}) with (\ref{eq:ub2}), we have that 
\begin{align}
\widetilde{W}(\widetilde{M}^{{\rm OPT}};D^{\star},F^{\star})-W(M^{{\rm OPT}};V^{\star},F^{\star}) & \ge\frac{\Delta^{2}\underline{f}}{2K(\overline{\theta}-\underline{\theta})}-\frac{\overline{f}(\overline{\theta}-\underline{\theta})\bar{\delta}^{2}}{K{N^{2}}}.\label{eq:inequality-primitives-price-dominates}
\end{align}
To establish that price regulation strictly dominates quantity regulation,
it thus suffices to show that, under the condition in part (b) of
the proposition, $\Delta>\sqrt{\frac{2\bar{f}}{\underline{f}}}\left(\frac{\bar{\delta}(\overline{\theta}-\underline{\theta})}{N}\right)$.
To see this, note that
\begin{align}
\Delta\equiv\int\limits^{\overline{\theta}}_{\underline{\theta}}\Big(\delta(\theta)-K\frac{F^{\star}(\theta)}{f^{\star}(\theta)}\Big){\rm d}\theta>\int\limits^{\overline{\theta}}_{\underline{\theta}}\delta(\theta)\frac{N-1}{N}{\rm d}\theta\ge\underline{\delta}(\overline{\theta}-\underline{\theta})\frac{N-1}{N}\ge\sqrt{\frac{2\bar{f}}{\underline{f}}}\left(\frac{\bar{\delta}(\overline{\theta}-\underline{\theta})}{N}\right).\label{eq:delta-inequality}
\end{align}
The first inequality follows from the condition in the proposition
along with the fact that $F^{\star}(\tlb)/f^{\star}(\tlb)=0<\underline{\delta}/(KN)$
and the continuity of $F^{\star}/f^{\star}$. Jointly these properties
imply that $KF^{\star}(\theta)/f^{\star}(\theta)<\delta(\theta)/N$
on a right-neighborhood of $\underline{\theta}$. The second inequality
follows from the definition of $\underline{\delta}$. The last inequality
is obtained by using the fact that $N-1=\frac{\bar{\delta}}{\underline{\delta}}\sqrt{\frac{2\bar{f}}{\underline{f}}}$.
\hfill{}Q.E.D.
\end{appendix}

{\small{}  \bibliographystyle{ecta}
\bibliography{references}
 }{\small\par}

\newpage{} 
\appendix
{
\setcounter{section}{18}
\section{Supplement I}\label{Sec:OS}
\global\long\def\thesublemma{\thelemma\Alph{sublemma}}%
\global\long\def\thelemma{S.\arabic{lemma}}%
\global\long\def\thesection{S.\arabic{section}}%
\global\long\def\thesubsection{S.\arabic{subsection}}%
\global\long\def\theproposition{S.\arabic{proposition}}%
\global\long\def\thedefn{S.\arabic{defn}}%
\global\long\def\thetheorem{S.\arabic{theorem}}%
\global\long\def\thecorollary{S.\arabic{corollary}}%
\global\long\def\theremark{S.\arabic{remark}}%
\global\long\def\theclaim{S.\arabic{claim}}%
\global\long\def\thefigure{S.\arabic{figure}}%
\global\long\def\theequation{S.\arabic{equation}}%
\global\long\def\theobs{S.\arabic{obs}}%

\renewcommand{\thesubsection}{S.\arabic{subsection}}
\setcounter{subsection}{0}

The notation in this supplement is the same as in the original article.
All sections, definitions, displayed conditions, and results specific
to this document have the prefix ``S'' to avoid confusion with the
corresponding parts in the main text. Section \ref{subsec:Stronger-Prop-2}
contains the proof of a stronger version of the result in Part 2 of
Proposition \ref{prop:robust-quantity-mechanism-general} in the main
text. Section \ref{sec:gamma} considers more permissive robustness
requirements. Section \ref{sec:Arbitrary-Cost-Uncertainty} considers
alternative forms of cost uncertainty.

\subsection{Stronger version of Part 2 of Proposition \ref{prop:robust-quantity-mechanism-general}.}

\label{subsec:Stronger-Prop-2} Part (2) of Proposition \ref{prop:robust-quantity-mechanism-general}
in the main text establishes that there exists $\theta^{\star\star}\in[\theta^{m},\tub]$
such that $\qopt(\theta)\leq\qbm(\theta)$ for all $\theta\in(\theta^{m},\min\{\theta^{\star},\theta^{\star\star}\})$
with the inequality strict for a subset of $(\theta^{m},\min\{\theta^{\star},\theta^{\star\star}\})$
of positive Lebesgue measure. Proposition \ref{prop:part2_strong}
below establishes a stronger version of this result. The proof of
this proposition further sharpens the characterization of the robustly
optimal mechanism by showing that the optimal quantity schedule in
the range $(\theta^{m},\min\{\theta^{\star},\theta^{\star\star}\})$
satisfies $\qopt(\theta)=D^{\star}(\bar{z}^{M}(\theta))$, where $\bar{z}^{M}$
is an ironed virtual cost that accounts for robustness.

\begin{proposition}\label{prop:part2_strong} Suppose $M^{\star}$
is not robustly optimal. Then, in any robustly optimal mechanism $M^{{\rm OPT}}=(\qopt,\uopt)$,
\begin{align}
\qopt(\theta)<\qbm(\theta)~\qquad~\forall~\theta\in(\theta^{m},\min\{\theta^{\star},\theta^{\star\star}\})\label{eq:strong-prop2}
\end{align}
where $\theta^{\star\star}=\inf\{\theta\in[\theta^{m},\tub]:\qopt(y)=\underline{D}(y)~\forall~y\in[\theta,\tub]\}$.
\end{proposition} Recall that, when $\theta^{\star}<\overline{\theta}$,
$\theta^{\star\star}=\tub$, and therefore $\min\{\theta^{\star},\theta^{\star\star}\}=\theta^{\star}$.
On the other hand, when $\theta^{\star}=\overline{\theta}$, $\min\{\theta^{\star},\theta^{\star\star}\}=\theta^{\star\star}$.

\noindent\textbf{Proof of Proposition \ref{prop:part2_strong}.}
The proof proceeds in five steps. Step 1 sets up a sub-program that
permits us to use Lagrangian methods with infinite-dimensional constraints
to identify necessary conditions for the optimal policy. Specifically,
it considers an auxiliary problem in which the policy is held fixed
over an interval of types $[\theta^{\dag},\bar{\theta}]$, with $\theta^{\dag}\in(\theta^{m},\min\{\theta^{\star},\theta^{\star\star}\}]$
and the optimization is over the shape of the policy over the interval
$[\tlb,\theta^{\dag})$. Step 2 identifies a few properties of the
solution to the sub-program of Step 1. Step 3 uses ironing techniques
to sharpen the characterization in Step 2. Step 4 shows how to select
$\theta^{\dag}$ to guarantee constraint qualification thus validating
the conditions identified in the previous steps, and then establishes
the result in the proposition by continuity. Finally, Step 5 proves
two technical lemmas used in the ironing arguments of Step 3.


\noindent\textbf{Step 1. Lagrangian method.}

\noindent Let $\theta^{\dag}\in(\theta^{m},\min\{\theta^{\star},\theta^{\star\star}\}]$
and define $\mathcal{Q}^{\dag}\equiv\{q:[\tlb,\tub]\rightarrow[0,{\rm \overline{q}}]:~q~\textrm{is weakly decreasing and}~q(\theta)=\qopt(\theta)~\forall~\theta\in[\theta^{\dag},\tub]\}.$
Note that $\mathcal{Q}^{\dag}$ is a convex set. Clearly, $\qopt$
must solve the following sub-program:\footnote{Recall from Lemma \ref{lemma:majorizations} in the main text that
$q$ satisfies the robustness constraints in (\ref{SL-constraint})
for all $\theta\in\Theta$ if and only if (a) $R(q)\geq0$, (b) $S(\theta;q)\geq0$
for all $\theta\in(\tlb,\tub)$, and (c) $q(\tub)={\rm q}_{\ell}$.
Imposing that $S(\theta;q)\geq0$ also for $\theta=\tlb$ is inconsequential
because this constraint is implied by $R(q)\geq0$. Omitting the constraints
$S(\theta;q)\geq0$ for $\theta>\theta^{\dag}$ and the constraint
$q(\tub)={\rm q}_{\ell}$ is motivated by the fact that $q$ is required
to agree with $\qopt$ over $[\theta^{\dag},\tub]$ along with the
fact that $\qopt$ satisfies these constraints because it is robustly
optimal.} 
\begin{align}
\qopt\in\arg\max_{q\in\mathcal{Q}^{\dag}}\int\limits^{\theta^{\dag}}_{\tlb}\Big[V^{\star}(q(\theta))-z^{\star}(\theta)q(\theta)\Big]f^{\star}(\theta){\rm d}\theta & \tag{\textbf{FP}\ensuremath{^{\dag}}}\label{opt:fpdag}\\
\textrm{subject to}~~~~S(\theta;q)\equiv\int\limits^{\tub}_{\theta}(\underline{D}(y)-q(y){\rm )d}y & \ge0~\qquad~\forall~\theta\in[\tlb,\theta^{\dag}],\nonumber \\
R(q)\equiv\underline{V}(q(\tlb))-\tlb q(\tlb)-\int\limits^{\tub}_{\tlb}q(y){\rm d}y-G^{*} & \ge0.\nonumber 
\end{align}

\begin{defn}\label{def:interior} The sub-program (\ref{opt:fpdag})
satisfies the \textbf{Interiority Condition} if there exists $q^{\dag}\in\mathcal{Q}^{\dag}$
such that $R(q^{\dag})>0$ and $\underset{\theta\in[\underline{\theta},\theta^{\dag}]}{\inf}S(\theta;q^{\dag})>0$.
\end{defn}

\noindent As anticipated above, in Step 4 we show how to select a
sequence of $\theta^{\dag}$ satisfying the Interiority Condition
and converging to $\min\{\theta^{\star},\theta^{\star\star}\}$. Along
the sequence, the properties identified in Steps 2-3 are necessary
for optimality. Taking the limit for $\theta^{\dag}\rightarrow\min\{\theta^{\star},\theta^{\star\star}\}$
then permits us to establish (\ref{eq:strong-prop2}). In Steps 2-4,
we focus on arbitrary $\theta^{\dag}$ satisfying the Interiority
Condition and use Lagrangian methods with infinite-dimensional constraints
to solve the sub-program (\ref{opt:fpdag}).

When the sub-program (\ref{opt:fpdag}) satisfies the Interiority
Condition of Definition \ref{def:interior}, Theorem 1 of \citet{Lu97}
(page 217) implies existence of a weakly increasing, right-continuous
function $\Lopt$, satisfying $\Lopt(\tlb)=0$, and a non-negative
scalar $\muopt\geq0$ such that 
\begin{align}
\qopt & \in\arg\max_{q\in\mathcal{Q}^{\dag}}\mathcal{L}(q;\Lopt,\muopt),\label{eq:lobj}
\end{align}
where, for any $q\in\mathcal{Q}^{\dag}$, 
\begin{align*}
\mathcal{L}(q;\Lopt,\muopt) & \equiv\int\limits^{\theta^{\dag}}_{\tlb}\Big[V^{\star}(q(\theta))-z^{\star}(\theta)q(\theta)\Big]f^{\star}(\theta){\rm d}\theta+\int\limits^{\theta^{\dag}}_{\tlb}\Bigg[\int\limits^{\tub}_{\theta}(\underline{D}(y)-q(y)){\rm d}y\Bigg]{\rm d}\Lopt(\theta)\\
 & +\muopt\Big[\underline{V}(q(\tlb))-\tlb q(\tlb)-\int\limits^{\tub}_{\tlb}q(y){\rm d}y-G^{*}\Big]
\end{align*}
is the Lagrangian function associated with the above sub-program (\ref{opt:fpdag}).
Furthermore, complementary slackness must hold: 
\begin{equation}
\int\limits^{\theta^{\dag}}_{\tlb}S(\theta;\qopt){\rm d}\Lopt(\theta)=0,~~\textrm{and}~~R(\qopt)\muopt=0.\label{eq:cs1}
\end{equation}
Using integration by parts and rearranging terms, we have that 
\begin{align}
\mathcal{L}(q;\Lopt,\muopt) & =\int\limits^{\theta^{\dag}}_{\tlb}\Big[V^{\star}(q(\theta))-\Big(z^{\star}(\theta)+\frac{\Lopt(\theta)+\muopt}{f^{\star}(\theta)}\Big)q(\theta)\Big]f^{\star}(\theta){\rm d}\theta\label{eq:l1}\\
 & +\int\limits^{\theta^{\dag}}_{\tlb}\Lopt(\theta)\underline{D}(\theta){\rm d}\theta+\Lopt(\theta^{\dag})\int\limits^{\tub}_{\theta^{\dag}}[\underline{D}(\theta)-q(\theta)]{\rm d}\theta\label{eq:l2}\\
 & +\muopt\Big[\underline{V}(q(\tlb))-\tlb q(\tlb)-\int\limits^{\tub}_{\theta^{\dag}}q(y){\rm d}y-G^{*}\Big].\label{eq:l3}
\end{align}
Note that the terms in (\ref{eq:l2}) are independent of the value
that $q$ takes on $[\tlb,\theta^{\dag}]$ whereas the terms in (\ref{eq:l3})
depend on the value that $q$ takes on $[\tlb,\theta^{\dag}]$ only
through $q(\tlb)$. We take advantage of these properties below.

\medskip{}

\noindent\textbf{Step 2. Miscellaneous properties of the optimum. }

\begin{obs}\label{obs:leftof_tmin} Suppose there exists $\tilde{\theta}\in(\tlb,\theta^{\dag})$
such that $\Lopt(\theta)+\muopt=0$ for all $\theta\in(\tlb,\tilde{\theta})$.
Then, $\qopt(\theta)=\qbm(\theta)$ for all $\theta\in(\tlb,\tilde{\theta})$.
\end{obs}

\noindent\textbf{Proof of Observation \ref{obs:leftof_tmin}.} Recall
that $\qopt(\theta)\leq\qbm(\theta)$ for all $\theta\in(\tlb,\tilde{\theta})$
by virtue of Lemma \ref{lem:weakqoptbound} in the main text along
with the fact that $\tilde{\theta}\leq\theta^{\dag}\leq\min\{\theta^{\star},\theta^{\star\star}\}$.
Now assume, towards a contradiction, that there exists a $\theta\in(\tlb,\tilde{\theta})$
such that $\qopt(\theta)<\qbm(\theta)$. The continuity of $\qbm$
along with the fact that $\qopt$ is weakly decreasing, then imply
that there exists an interval $I\subset(\tlb,\tilde{\theta})$ of
positive Lebesgue measure such that $\qopt(y)<\qbm(y)$ for all $y\in I$.
Then consider the quantity schedule 
\[
\tilde{q}(\theta)=\begin{cases}
\qbm(\theta) & \mbox{if}~\theta\in[\underline{\theta},\tilde{\theta})\\
\qopt(\theta) & \mbox{if}~\theta\in[\tilde{\theta},\overline{\theta}].
\end{cases}
\]
Observe that $\tilde{q}$ is weakly decreasing because $\qbm$ is
decreasing, $\qopt$ is weakly decreasing, and $\qopt(\tilde{\theta})\le\qbm(\tilde{\theta})$
by Lemma \ref{lem:weakqoptbound} in the main text. Moreover, $\mathcal{L}(\tilde{q};\Lambda^{{\rm OPT}},\muopt)>\mathcal{L}(\qopt;\Lambda^{{\rm OPT}},\muopt)$,
where the inequality follows because $\qbm(\theta)$ is the unique
maximizer of $V^{\star}({\rm q})-z^{\star}(\theta){\rm q}$ and $\qopt(\theta)<\qbm(\theta)$
on the interval $I$ of positive Lebesgue measure. Because $\tilde{q}\in Q^{\dagger}$,
this strict inequality contradicts Condition (\ref{eq:lobj}) and
the robust optimality of $\qopt$. We therefore conclude that $\qopt(\theta)=\qbm(\theta)$
for all $\theta\in(\tlb,\tilde{\theta})$. $\hfill\blacksquare{}$

Next, let 
\[
\theta_{\min}\equiv\inf\{\theta\in(\tlb,\theta^{\dag}):\Lopt(\theta)+\muopt>0\}.
\]
Step 4 will establish existence of $\theta^{\dag}$ for which $\{\theta\in(\tlb,\theta^{\dag}):\Lopt(\theta)+\muopt>0\}\neq\emptyset$
and hence $\theta_{\min}$ is well-defined. Until then, we will assume
that $\theta_{\min}$ exists and derive a few key properties of $\theta_{\min}$
instrumental to the result in the proposition.

\begin{lemma} \label{lem:tmin2} $\theta_{\min}\le\theta^{m}$. \end{lemma}

\noindent\textbf{Proof of Lemma \ref{lem:tmin2}.} If $\theta_{\min}=\tlb$,
the claim is trivially true. Therefore, suppose that $\theta_{\min}>\tlb$.
We first establish two important properties of $\qopt(\theta_{\min})$:
(a) $\underline{W}(\theta_{\min},\qopt)=G^{*}$, and (b) $\qopt(\theta_{\min})=\underline{D}(\theta_{\min})$.
To see this, note that, because $S(\theta;\qopt)\ge0$ for all $\theta\in(\tlb,\theta^{\dag})$,
the first complementary slackness condition in (\ref{eq:cs1}) implies
that $S(\theta;\qopt)=0$ for $\Lopt$-almost every $\theta\in(\tlb,\theta^{\dag})$.
Hence, if $\Lopt$ has a jump at $\theta_{\min}$, that is, if 
\begin{align*}
\Delta\Lopt(\theta_{\min})\equiv\Lopt(\theta_{\min})-\lim_{\theta\uparrow\theta_{\min}}\Lopt(\theta)>0
\end{align*}
then $S(\theta_{\min},\qopt)=0$. If, instead, there is no jump, i.e.,
if $\Delta\Lopt(\theta_{\min})=0$, then by the definition of $\theta_{\min}$
and the fact that $\theta_{\min}>\tlb$, it must be that $\muopt=0$.
Because $\Lopt$ is weakly increasing, for every $\epsilon\in(0,\theta^{\dag}-\theta_{\min})$,
the interval $(\theta_{\min},\theta_{\min}+\epsilon]$ has strictly
positive $\Lambda^{{\rm OPT}}$-measure. Because $S(\theta;q^{{\rm OPT}})=0$
for $\Lambda^{{\rm OPT}}$-almost every $\theta\in(\tlb,\theta^{\dag})$, each such interval
contains some $\theta'>\theta_{\min}$ such that $S(\theta';\qopt)=0$.
Then, let $(\theta_{n})$ be a monotone sequence such that $\theta_{n}\downarrow\theta_{\min}$,
and $S(\theta_{n};\qopt)=0$ for all $n$. By continuity of $S$,
we then have that $S(\theta_{\min};\qopt)=0$. By Lemma \ref{lemma:majorizations}
in the main text, we have that 
\begin{align*}
\underline{W}(\theta_{\min},\qopt)-G^{*}=S(\theta_{\min};\qopt)-{\rm \underline{{\rm DWL}}}(\theta_{\min},\qopt(\theta_{\min}))\le0.
\end{align*}
Because robust optimality implies that $\underline{W}(\theta,\qopt)-G^{*}\ge0$,
we conclude that (a) $\underline{W}(\theta_{\min},\qopt)=G^{*}$ and
${\rm \underline{{\rm DWL}}}(\theta_{\min},\qopt(\theta_{\min}))=0$.
The last equality (deadweight loss equal to zero) implies that (b)
$\qopt(\theta_{\min})=\underline{D}(\theta_{\min})$.

Now, assume, towards a contradiction to what is claimed in the proposition,
that $\theta_{\min}>\theta^{m}$. Because $\theta_{\min}<\theta^{\dag}$,
Observation \ref{obs:leftof_tmin} implies that $\qopt(\theta)=\qbm(\theta)$
for all $\theta\in(\tlb,\theta_{\min})$. This implies that $\qopt(\theta^{m})=\qbm(\theta^{m})$
when $\theta^{m}>\tlb$. When, instead, $\theta^{m}=\tlb$, $\qopt(\theta^{m})\ge\qbm(\theta^{m})$
because $\qopt$ is weakly decreasing and $\qopt(\theta)=\qbm(\theta)$
for all $\theta\in(\theta^{m},\theta_{\min})$. Thus, we have that
$\qopt(\theta^{m})\ge\qbm(\theta^{m})$. Also, as established earlier,
$\qopt(\theta_{\min})=\underline{D}(\theta_{\min})$. Now, observe
that 
\begin{align*}
\underline{W}(\theta_{\min},\qopt)-\underline{W}(\theta^{m},\qopt) & =\Big[\underline{V}(\qopt(\theta_{\min}))-\theta_{\min}\qopt(\theta_{\min})\Big]\\
 & -\Big[\underline{V}(\qopt(\theta^{m}))-\theta^{m}\qopt(\theta^{m})\Big]+\int\limits^{\theta_{\min}}_{\theta^{m}}\qopt(y){\rm d}y\\
 & \ge_{(i)}\Big[\underline{V}(\qbm(\theta_{\min}))-\theta_{\min}\qbm(\theta_{\min})\Big]\\
 & -\Big[\underline{V}(\qbm(\theta^{m}))-\theta^{m}\qbm(\theta^{m})\Big]+\int\limits^{\theta_{\min}}_{\theta^{m}}\qbm(y){\rm d}y\\
 & =\underline{W}(\theta_{\min},q^{\star})-\underline{W}(\theta^{m},q^{\star})>0,
\end{align*}
where the strict inequality follows from the fact that $\theta_{\min}>\theta^{m}$
and the definition of $\theta^{m}$. To see why inequality $(i)$
holds, first observe that Observation \ref{obs:leftof_tmin} implies
that $\qopt(\theta)=\qbm(\theta)$ for all $\theta\in(\tlb,\theta_{\min})$.
Moreover, $\underline{V}({\rm q})-\theta{\rm q}$ is strictly concave
in ${\rm q}$ with a unique maximum at $\underline{D}(\theta)$, and
(a) $\qbm(\theta_{\min})\ge\qopt(\theta_{\min})=\underline{D}(\theta_{\min})$,
and (b) $\qopt(\theta^{m})\ge\qbm(\theta^{m})\ge\underline{D}(\theta^{m})$,
where the second inequality follows from Lemma \ref{lem:qstardlowint}
in the main text. Because $\qopt$ satisfies the robustness constraints,
$\underline{W}(\theta^{m},\qopt)\ge G^{*}$. Hence, we have that $\underline{W}(\theta_{\min},\qopt)>G^{*}$,
a contradiction to what we established earlier.$\hfill\blacksquare{}$


\noindent\textbf{Step 3. Ironing.} Let 
\[
z^{M}(\theta)\equiv z^{\star}(\theta)+\frac{\Lopt(\theta)+\muopt}{f^{\star}(\theta)}
\]
be the \textsl{modified virtual cost}. Note that $z^{M}$ need not
be weakly increasing. Below, we thus consider an auxiliary optimization
problem, defined on the interval $(\tlb,\theta^{\dag})$, which allows
us to use ironing techniques similar to \citet{BM82} to identify
further properties of the robustly optimal schedule $\qopt$. Let
$\mathcal{Q}^{\dag-}\subset\mathcal{Q}^{\dag}$ be the subset of $\mathcal{Q}^{\dag}$
given by 
\begin{align*}
\mathcal{Q}^{\dag-} & \equiv\{q\in\mathcal{Q}^{\dag}:~\ensuremath{q(\tlb)=\qopt(\tlb)}\}.
\end{align*}
Again, $\mathcal{Q}^{\dag-}$ is a convex set. Taking advantage of
the fact that the terms in (\ref{eq:l2}) and (\ref{eq:l3}) are independent
of the value that $q$ takes on $(\tlb,\theta^{\dag})$, we have that
a necessary condition for $\qopt\in\arg\max_{q\in\mathcal{Q}^{\dag}}\mathcal{L}(q;\Lopt,\muopt)$
is that 
\begin{align}
\qopt\in\arg\max_{q\in\mathcal{Q}^{\dag-}}\int\limits^{\theta^{\dag}}_{\tlb}\Big[V^{\star}(q(\theta))-z^{M}(\theta)q(\theta)\Big]f^{\star}(\theta){\rm d}\theta.\tag{\textbf{UNCONST}}\label{opt:unconst}
\end{align}

The terms in the square brackets in the integrand in (\ref{opt:unconst})
are separable in $q(\theta)$ and $\theta$ in the sense of \citet{T11}.
The solution can thus be obtained by ironing $z^{M}$. For any $\phi\in[0,F^{\star}(\theta^{\dag})]$,
let 
\begin{align*}
h(\phi)\equiv z^{M}((F^{\star})^{-1}(\phi))~~~~~\textrm{and}~~~~H(\phi)\equiv\int\limits^{\phi}_{0}h(y){\rm d}y.
\end{align*}
Let $\overline{H}\equiv{\rm conv}(H)$ be the convex hull of $H$,
i.e., the highest convex function on $[0,F^{\star}(\theta^{\dag})]$
such that $\overline{H}\le H$. Moreover, $\overline{H}(0)=H(0)$
and $\overline{H}(F^{\star}(\theta^{\dag}))=H(F^{\star}(\theta^{\dag}))$.
Because $\overline{H}$ is convex, it is continuously differentiable
almost everywhere. At all $\phi$ where $\overline{H}$ is differentiable,
let $\overline{h}(\phi)\equiv{\rm d}\overline{H}(\phi)/{\rm d}\phi$
and then extend $\overline{h}$ to all $[0,F^{\star}(\theta^{\dag})]$
by right continuity (i.e., taking the right derivative of $\overline{H}$).
Then let 
\begin{align*}
\overline{z}^{M}(\theta) & \equiv\overline{h}(F^{\star}(\theta)),~\qquad~\forall~\theta\in(\tlb,\theta^{\dag})
\end{align*}
be the\textsl{ ironed modified virtual cost}. By Theorem 3.7 in \citet{T11},
$\qopt$ solves (\ref{opt:unconst}) if and only if\footnote{\label{footnote}If $\qopt(\tlb)=\qopt(\theta^{\dag})$, then for
any $\tlb<\theta<\theta'<\theta^{\dag}$, we have $\qopt(\theta)=\qopt(\theta')\le\qbm(\theta')<\qbm(\theta)$,
where the weak inequality follows from Lemma \ref{lem:weakqoptbound}
in the main text whereas the strict inequality follows from the fact
that $\qbm$ is decreasing. Hence, $\qopt(\theta)<\qbm(\theta)$ for
all $\theta\in(\tlb,\theta^{\dag})$, and, by Lemma \ref{lem:leftofthetam}
in the main text, $\theta^{m}=\tlb$. We conclude that, in this case,
$\qopt(\theta)<\qbm(\theta)$ for all $\theta\in(\theta^{m},\theta^{\dag})$.
The rest of the proof then follows from the arguments in Step 4. In
the rest of Step 3 below, we thus assume that $\qopt(\theta^{\dag})<\qopt(\tlb)$.} 
\begin{enumerate}
\item[(a)] for almost all $\theta\in(\tlb,\theta^{\dag})$, 
\begin{align*}
\qopt(\theta)\in\arg\max_{{\rm q}\in[\qopt(\theta^{\dag}),\qopt(\tlb)]}\Big[V^{\star}({\rm q})-\overline{z}^{M}(\theta){\rm q}\Big];
\end{align*}
\item[(b)] for all open intervals $I\subseteq(\tlb,\theta^{\dag})$ such that
$\overline{H}(F^{\star}(\theta))<H(F^{\star}(\theta))$ for all $\theta\in I$,
$\bar{z}^{M}$ is constant over $I$, and therefore, $\qopt$ is also
constant over $I$ (pooling property). 
\end{enumerate}
The following technical lemma, whose proof is in Step 5, establishes
a relationship between $\bar{z}^{M}$ and $\qopt$ which we use below:

\begin{lemma} \label{lem:barz-continuous-middle}If $\bar{z}^{M}$
is discontinuous at $\theta\in(\tlb,\theta^{\dag})$, then $\qopt$
is continuous at $\theta$. \end{lemma}

For any $\theta\in(\tlb,\theta^{\dag})$, the function $V^{\star}({\rm q})-\overline{z}^{M}(\theta){\rm q}$
is strictly concave in ${\rm q}$, with a unique maximum at $D^{\star}(\overline{z}^{M}(\theta))$.
Because $\overline{z}^{M}$ is weakly increasing, the function $D^{\star}(\overline{z}^{M}(\cdot))$
is weakly decreasing. We divide the interval $(\tlb,\theta^{\dag})$
into three parts. Let 
\[
\theta_{1}\equiv\sup\{\theta\in(\tlb,\theta^{\dag}):D^{\star}(\overline{z}^{M}(\theta))>\qopt(\tlb)\}
\]
if $\{\theta\in(\tlb,\theta^{\dag}):D^{\star}(\overline{z}^{M}(\theta))>\qopt(\tlb)\}\neq\emptyset$
and otherwise let $\theta_{1}\equiv\tlb$. Similarly, let 
\[
\theta_{2}\equiv\inf\{\theta\in(\tlb,\theta^{\dag}):D^{\star}(\overline{z}^{M}(\theta))<\qopt(\theta^{\dag})\}
\]
if $\{\theta\in(\tlb,\theta^{\dag}):D^{\star}(\overline{z}^{M}(\theta))<\qopt(\theta^{\dag})\}\neq\emptyset$
and otherwise let $\theta_{2}\equiv\theta^{\dag}$.

\begin{claim} \label{cl:theta12} $\theta_{1}\le\theta_{2}$. If
$\theta_{1}=\theta_{2}=\hat{\theta}$, then $\hat{\theta}\in\{\tlb,\theta^{\dag}\}$.
\end{claim} 

\noindent \textbf{Proof of Claim \ref{cl:theta12}.} Assume, towards a contradiction
to what is claimed, that $\theta_{1}>\theta_{2}$. Then choose $\theta\in(\theta_{2},\theta_{1})$.
By the definition of $\theta_{1}$ and $\theta_{2}$, $\qopt(\tlb)<D^{\star}(\bar{z}^{M}(\theta))<\qopt(\theta^{\dag})$,
a contradiction to the fact that $\qopt$ is weakly decreasing. Thus
$\theta_{1}\le\theta_{2}$, as claimed.

\noindent Next, consider the second part of the claim. That is, suppose
that $\theta_{1}=\theta_{2}=\hat{\theta}$. We want to establish that
$\hat{\theta}\in\{\tlb,\theta^{\dag}\}$. The proof is again by contradiction.
Suppose that $\hat{\theta}\in(\tlb,\theta^{\dag})$. Then, for any
$\theta\in[\tlb,\hat{\theta})$ and $\theta'\in(\hat{\theta},\theta^{\dag}]$,
$\qopt(\theta)=\qopt(\tlb)>\qopt(\theta^{\dag})=\qopt(\theta')$,
where the inequality follows from Footnote \ref{footnote}, i.e.,
from the fact that we are considering the case in which 
\[
\Delta\equiv\qopt(\tlb)-\qopt(\theta^{\dag})>0.
\]
Hence, $\qopt$ is discontinuous at $\hat{\theta}$. By Lemma \ref{lem:barz-continuous-middle},
$\bar{z}^{M}$ must thus be continuous at $\hat{\theta}$. Now consider
a decreasing sequence of types $(\theta_{n})$ converging to $\hat{\theta}=\theta_{2}$.
Using the definition of $\theta_{2}$, we have that $D^{\star}(\bar{z}^{M}(\theta_{n}))\le\qopt(\theta^{\dag})=\qopt(\tlb)-\Delta$.
The continuity of $D^{\star}\circ\bar{z}^{M}$ at $\hat{\theta}$
implies that $D^{\star}(\bar{z}^{M}(\hat{\theta}))=\lim_{n\rightarrow\infty}D^{\star}(\bar{z}^{M}(\theta_{n}))\le\qopt(\tlb)-\Delta<\qopt(\tlb)$.
The same continuity in turn implies existence of $\epsilon>0$ such
that $D^{\star}(\bar{z}^{M}(\hat{\theta}-\epsilon))<\qopt(\tlb)$,
which contradicts the definition of $\theta_{1}$ because $\hat{\theta}=\theta_{1}$.
We conclude that, if $\theta_{1}=\theta_{2}=\hat{\theta}$, then $\hat{\theta}\in\{\tlb,\theta^{\dag}\}$,
as claimed. $\hfill\blacksquare{}$

The optimal schedule $\qopt$ must thus satisfy the following properties:
\begin{align}
\qopt(\theta) & =\qopt(\tlb)~\qquad~\forall~\theta\in(\tlb,\theta_{1}),\label{eq:iron3}\\
\qopt(\theta) & =D^{\star}(\overline{z}^{M}(\theta))~\qquad~\textrm{for almost all}~\theta\in(\theta_{1},\theta_{2}),\label{eq:iron2}\\
\qopt(\theta) & =\qopt(\theta^{\dag})~\qquad~\forall~\theta\in(\theta_{2},\theta^{\dag}).\label{eq:iron1}
\end{align}

\begin{obs} \label{obs:inttheta12} (1) Suppose $\theta_{1}>\tlb$.
Then, $\qopt(\theta)<\qbm(\theta)$ for all $\theta\in(\tlb,\theta_{1})$
and $\theta^{m}=\tlb$; (2) Suppose $\theta_{2}<\theta^{\dag}$. Then,
$\qopt(\theta)<\qbm(\theta)$ for all $\theta\in(\theta_{2},\theta^{\dag})$
and $\theta^{m}\le\theta_{2}$. \end{obs}

\noindent\textbf{Proof of Observation \ref{obs:inttheta12}.} Because
$\theta^{\dag}\le\theta^{\star}$, by Lemma \ref{lem:weakqoptbound}
in the main text, $\qopt(\theta)\le\qbm(\theta)$ for all $\theta\in(\tlb,\theta^{\dag})$.

\noindent\emph{Part (1).} By (\ref{eq:iron3}), for any $\tlb<\theta<\theta'<\theta_{1}$,
$\qopt(\tlb)=\qopt(\theta)=\qopt(\theta')\le\qbm(\theta')<\qbm(\theta)$,
where the first inequality follows from Lemma \ref{lem:weakqoptbound}
in the main text, whereas the second one follos from the fact that
$\qbm$ is decreasing. Hence, $\qopt(\theta)<\qbm(\theta)$ for all
$\theta\in(\tlb,\theta_{1})$. Lemma \ref{lem:leftofthetam} in the
main text then implies that $\theta^{m}=\tlb$.

\noindent\emph{Part (2).} The proof is identical to that of Part
(1): Pick $\theta_{2}<\theta<\theta'<\theta^{\dag}$. By (\ref{eq:iron1}),
$\qopt(\theta^{\dag})=\qopt(\theta)=\qopt(\theta')\le\qbm(\theta')<\qbm(\theta)$.
Hence, $\qopt(\theta)<\qbm(\theta)$ for all $\theta\in(\theta_{2},\theta^{\dag})$.
Lemma \ref{lem:leftofthetam} in the main text then implies that $\theta^{m}\le\theta_{2}$.
$\hfill\blacksquare{}$

Now recall that we want to show that $\qopt(\theta)<\qbm(\theta)$
for all $\theta\in(\theta^{m},\theta^{\dag})$. By virtue of Claim
\ref{cl:theta12} and Observation \ref{obs:inttheta12}, it only remains
to show that, when $\theta^{m}<\theta_{2}$ and $\theta_{1}<\theta_{2}$,
this inequality holds for all $\theta\in I^{\dag}$, where\footnote{To see this, note that, when $\theta_{1}=\theta_{2}=\theta^{\dag}$,
Part (1) of Observation \ref{obs:inttheta12} implies that $\theta^{m}=\tlb$
and $\qopt(\theta)<\qbm(\theta)$ for all $\theta\in(\theta^{m},\theta^{\dag})$.
Similarly, when $\theta_{1}=\theta_{2}=\tlb$, Part (2) of Observation
\ref{obs:inttheta12} implies that $\theta^{m}=\tlb$, and $\qopt(\theta)<\qbm(\theta)$
for all $\theta\in(\theta^{m},\theta^{\dag})$. Claim \ref{cl:theta12}
in turn rules out the possibility that $\theta_{1}=\theta_{2}=\hat{\theta}\in(\tlb,\theta^{\dag})$.
Hence, it only remains to show that, when $\theta^{m}<\theta_{2}$
and $\theta_{1}<\theta_{2}$, the same inequality holds for all $\theta\in I^{\dag}$.} 
\[
I^{\dag}\equiv\left((\theta^{m},\theta_{2})\cap[\theta_{1},\theta_{2})\right)\cup\left(\{\theta_{2}\}\cap(\tlb,\theta^{\dag})\right).
\]
Note that $I^{\dag}$ is the union of the set $(\theta^{m},\theta_{2})\cap[\theta_{1},\theta_{2})$
with the singleton $\{\theta_{2}\}$, when $\theta_{2}<\theta^{\dag}$.
See Figure \ref{fig:step3fig} for an illustration. 
\begin{figure}
\centering \includegraphics[width=0.5\linewidth]{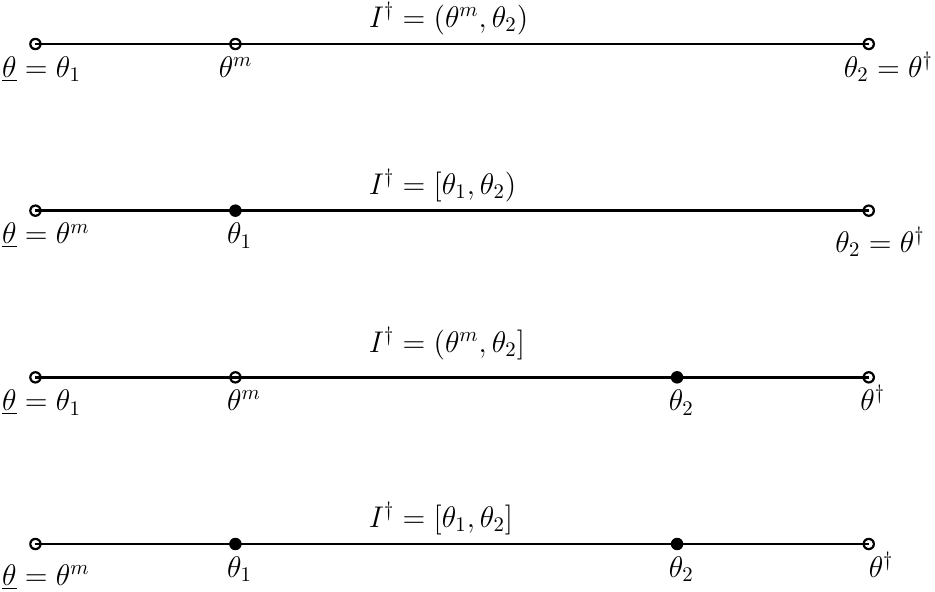}
\caption{Interval $I^{\protect\dag}$ where it remains to show that $\qopt(\theta)<\qbm(\theta)$.}
\label{fig:step3fig} 
\end{figure}

We make use of this additional technical lemma whose proof is relegated
to Step 5 below.

\begin{lemma} \label{lem:iron-ae} For all $\theta\in I^{\dag}$,
$\bar{z}^{M}(\theta)>z^{\star}(\theta)$. \end{lemma}

Now recall that we are considering the case $\theta_{1}<\theta_{2}$.
Condition (\ref{eq:iron2}) then implies that the set $E\equiv\{\theta\in(\theta_{1},\theta_{2}):\qopt(\theta)=D^{\star}(\bar{z}^{M}(\theta))\}$
has strictly positive Lebesgue measure. To show that, when $I^{\dag}\neq\emptyset$,
$\qopt(\theta)<\qbm(\theta)$ for all $\theta\in I^{\dag}$, we consider
two cases.

\noindent\textsc{Case A}: $\bar{z}^{M}$ is discontinuous at $\theta\in I^{\dag}$.

\noindent By Lemma \ref{lem:barz-continuous-middle}, $\qopt$ is
then continuous at $\theta$. First suppose that $\theta\in(\theta^{m},\theta_{2})\cap[\theta_{1},\theta_{2})$.
Let $(\theta_{n})$ be a decreasing sequence such that $\theta_{n}\in E$
for all $n$ and $\lim_{n\rightarrow\infty}\theta_{n}=\theta$. By
the right-continuity of $\bar{z}^{M}$ and the continuity of $D^{\star}$,
we have that $\lim_{n\rightarrow\infty}D^{\star}(\bar{z}^{M}(\theta_{n}))=D^{\star}(\bar{z}^{M}(\theta))$.
The continuity of $\qopt$ at $\theta$ in turn implies that $\lim_{n\rightarrow\infty}\qopt(\theta_{n})=\qopt(\theta)$.
Because, for every $n$, $\theta_{n}\in E$, Condition (\ref{eq:iron2})
implies that $\qopt(\theta_{n})=D^{\star}(\bar{z}^{M}(\theta_{n}))$.
As a result, $D^{\star}(\bar{z}^{M}(\theta))=\qopt(\theta)$. By Lemma
\ref{lem:iron-ae}, $\bar{z}^{M}(\theta)>z^{\star}(\theta)$, and
hence, $\qopt(\theta)=D^{\star}(\bar{z}^{M}(\theta))<D^{\star}(z^{\star}(\theta))=\qbm(\theta)$.
Next suppose that $\theta=\theta_{2}<\theta^{\dag}$. By Condition
(\ref{eq:iron1}), $\qopt(\theta)=\qopt(\theta^{\dag})$ for all $\theta\in(\theta_{2},\theta^{\dag})$.
Because $\qopt$ is continuous at $\theta_{2}$, $\qopt(\theta_{2})=\qopt(\theta^{\dag})$.
This means that there exists $\epsilon>0$ such that $\qopt(\theta_{2})=\qopt(\theta_{2}+\epsilon)<\qbm(\theta_{2}+\epsilon)<\qbm(\theta_{2})$,
where the first inequality follows from Observation \ref{obs:inttheta12}
and the second inequality follows from the fact that $\qbm$ is decreasing.
We conclude that, when $\bar{z}^{M}$ is discontinuous at $\theta\in I^{\dag}$,
no matter whether $\theta\in(\theta^{m},\theta_{2})\cap[\theta_{1},\theta_{2})$
or $\theta=\theta_{2}<\theta^{\dag}$, $\qopt(\theta)<\qbm(\theta)$.

\noindent\textsc{Case B: }$\bar{z}^{M}$ is continuous at $\theta\in I^{\dag}$. 

\noindent First suppose that $\theta\in I^{\dag}$ with $\theta\neq\theta_{1}$.
Then let $(\theta_{n})$ be an increasing sequence such that $\theta_{n}\in E$
for all $n$ and $\lim_{n\rightarrow\infty}\theta_{n}=\theta$. Thus,
for every $n$, $\qopt(\theta_{n})=D^{\star}(\bar{z}^{M}(\theta_{n}))$.
Next, suppose that $\theta\in I^{\dag}$ with $\theta=\theta_{1}$.
Note that this means that $\theta_{1}>\tlb$. Then let $(\theta_{n})$
be any sequence such that $\theta_{n}\in(\tlb,\theta_{1})$ for all
$n$ and $\lim_{n\rightarrow\infty}\theta_{n}=\theta_{1}.$ Condition
(\ref{eq:iron3}) then implies that, along such a sequence $\qopt(\theta_{n})=\qopt(\tlb)\le D^{\star}(\bar{z}^{M}(\theta_{n}))$
for every $n$. No matter which of the two sequences one considers,
the continuity of $\bar{z}^{M}$ along with the continuity of $D^{\star}$,
imply that $\lim_{n\rightarrow\infty}D^{\star}(\bar{z}^{M}(\theta_{n}))=D^{\star}(\bar{z}^{M}(\theta))$.
Furthermore, because $\qopt$ is weakly decreasing, $\qopt(\theta)\le\lim_{n\rightarrow\infty}\qopt(\theta_{n})\le\lim_{n\rightarrow\infty}D^{\star}(\bar{z}^{M}(\theta_{n}))=D^{\star}(\bar{z}^{M}(\theta))<D^{\star}(z^{\star}(\theta))=\qbm(\theta)$,
where the strict inequality follows from Lemma \ref{lem:iron-ae}
(which establishes that $\bar{z}^{M}(\theta)>z^{\star}(\theta)$).

\medskip{}

\noindent\textbf{Step 4. Choosing $\theta^{\dag}$ to satisfy the
Interiority Condition, the existence of $\theta_{\min}$, and the
property that $\qopt(\theta)<\qbm(\theta)$ for all $\theta\in(\theta^{m},\min\{\theta^{\star},\theta^{\star\star}\})$.}

The steps above have established that $\qopt(\theta)<\qbm(\theta)$
for all $\theta\in(\theta^{m},\theta^{\dag})$. To establish the result
in the proposition it thus suffices to construct a sequence of cut-offs
$\theta^{\dag}$, denoted by $(\theta^{\dag}_{n})$, converging to
$\min\{\theta^{\star},\theta^{\star\star}\}$, satisfying the Interiority
Condition of Definition \ref{def:interior}, and such that $\{\theta\in(\tlb,\theta^{\dag}):\Lopt(\theta)+\muopt>0\}\neq\emptyset$,
thus implying that the threshold $\theta_{\min}$ in the previous
step is well defined.

\noindent\textbf{Case 1.} $\theta^{\star}<\tub$. In this case, the
sequence $(\theta^{\dag}_{n})$ can be taken to be degenerate with
$\theta^{\dag}_{n}=\theta^{\dag}=\theta^{\star}$ for all $n$. Because
$\underline{D}(\theta^{\star})>{\rm q}_{\ell}=\qopt(\theta^{\star})$,
to establish that such a $\theta^{\dag}$ satisfies the Interiority
Condition of Definition \ref{def:interior}, consider the schedule
$q^{\dag}$ given by $q^{\dag}(\theta)\equiv\underline{D}(\theta)$
for all $\theta\in[\tlb,\theta^{\dag})$ and $q^{\dag}(\theta)\equiv{\rm q}_{\ell}=\qopt(\theta)$
for all $\theta\in[\theta^{\dag},\overline{\theta}]$. Clearly, $S(\theta;q^{\dag})=\int^{\tub}_{\theta^{\dag}}(\underline{D}(y)-{\rm q}_{\ell}){\rm d}y>0$
for all $\theta\in[\tlb,\theta^{\dag}]$. Furthermore, because $q^{\dag}(\tlb)=\underline{D}(\tlb)$,
${\rm \underline{{\rm DWL}}}(\tlb,q^{\dag}(\tlb))=0$. Lemma \ref{lemma:majorizations}
in the main text then implies that $R(q^{\dag})=S(\tlb;q^{\dag})>0$.
Hence, the sub-program (\ref{opt:fpdag}) with $\theta^{\dag}=\theta^{\star}$
satisfies the Interiority Condition of Definition \ref{def:interior}.

Setting $\theta^{\dag}\equiv\theta^{\star}$ also ensures that $\{\theta\in(\tlb,\theta^{\dag}):\Lopt(\theta)+\muopt>0\}\neq\emptyset$
and hence that $\theta_{\min}$ is well-defined. To see why, suppose,
towards a contradiction, that $\Lopt(\theta)+\muopt=0$ for all $\theta\in(\tlb,\theta^{\star})$.
The optimality condition in (\ref{eq:lobj}) then implies that $\qopt(\theta)=\qbm(\theta)$
for all $\theta\in(\tlb,\theta^{\star})$. Along with the fact that
$\qopt(\theta)={\rm q}_{\ell}$ for all $\theta\in[\theta^{\star},\tub]$,
this means that $\qopt(\theta)=q^{\star}(\theta)$ for all $\theta\in(\tlb,\tub]$.
This contradicts the fact that Baron-Myerson-with-quantity-floor is
not robustly optimal. We thus conclude that $\{\theta\in(\tlb,\theta^{\dag}):\Lopt(\theta)+\muopt>0\}\neq\emptyset$
and hence that $\theta_{\min}$ is well-defined. Because $\min\{\theta^{\star},\theta^{\star\star}\}=\theta^{\star}$,
the results in Step 3 then imply that $\qopt(\theta)<\qbm(\theta)$
for all $\theta\in(\theta^{m},\min\{\theta^{\star},\theta^{\star\star}\})$,
as claimed.

\noindent\textbf{Case 2.} $\theta^{\star}=\tub$. In this case, $\min\{\theta^{\star},\theta^{\star\star}\}=\theta^{\star\star}.$
Recall that we are considering the case in which $\theta^{\star\star}>\theta^{m}$.
The definition of $\theta^{\star\star}$ along with the fact that
$\qopt$ satisfies the constraint $S(\theta;\qopt)\geq0$ for all
$\theta$, implies that for every $\delta\in(0,\theta^{\star\star}-\theta^{m})$,
there exists $\theta^{\dag}\in(\theta^{\star\star}-\delta,\theta^{\star\star})$
such that $\qopt(\theta^{\dag})<\underline{D}(\theta^{\dag})$ and
$S(\theta^{\dag};\qopt)>0$. Now, for any such $\theta^{\dag}$, let
$q^{\dag}\in\mathcal{Q}^{\dag}$ be the schedule defined by 
\[
q^{\dag}(\theta)=\begin{cases}
\underline{D}(\theta) & \mbox{for all}~\theta<\theta^{\dag}\\
\qopt(\theta) & \mbox{otherwise}.
\end{cases}
\]
Clearly, for any $\theta<\theta^{\dag}$, $S(\theta;q^{\dag})=S(\theta^{\dag};\qopt)>0$.
Moreover, because $q^{\dag}(\tlb)=\underline{D}(\tlb)$, ${\rm \underline{{\rm DWL}}}(\tlb,q^{\dag}(\tlb))=0$.
Lemma \ref{lemma:majorizations} in the main text then implies that
$R(q^{\dag})=S(\tlb;q^{\dag})>0$. Hence, the sub-program (\ref{opt:fpdag})
with such a chosen $\theta^{\dag}$ satisfies the Interiority Condition
of Definition \ref{def:interior}.

Next, observe that Lemma \ref{lem:generictheta2} in the main text
implies that $\qopt(\theta)<\qbm(\theta)$ on a set $I\subset(\theta^{m},\theta^{\star\star})$
of positive Lebesgue measure. When combined with Observation \ref{obs:leftof_tmin}
above and the fact that $\Lopt$ is weakly increasing, this property
implies that there exists $\tilde{\delta}\in(0,\theta^{\star\star}-\theta^{m})$
such that, for any $\theta^{\dag}\in$ $\big(\theta^{\star\star}-\tilde{\delta},\theta^{\star\star}\big)$,
the set $\{\theta\in(\tlb,\theta^{\dag}):\Lopt(\theta)+\muopt>0\}\neq\emptyset$.

The above results imply existence of a monotone sequence $(\theta^{\dag}_{n})$,
with $\theta^{\dag}_{n}\uparrow\theta^{\star\star}$ as $n\rightarrow\infty$,
such that, for any $n$, (a) the sub-program (\ref{opt:fpdag}) satisfies
the Interiority Condition of Definition \ref{def:interior}, and (b)
$\{\theta\in(\tlb,\theta^{\dag}_{n}):\Lopt(\theta)+\muopt>0\}\neq\emptyset$.
The result that $\qopt(\theta)<\qbm(\theta)$ for all $\theta\in(\theta^{m},\min\{\theta^{\star},\theta^{\star\star}\})$
then follows from these properties along with what we established
in Step 3.

\medskip{}

\noindent\textbf{Step 5: Proofs of technical lemmas in Step 3.}

\noindent We conclude the proof of (\ref{eq:strong-prop2}) by proving
the two technical lemmas in Step 3.

\noindent\textbf{Proof of Lemma \ref{lem:barz-continuous-middle}}.
We use the following notation in the proof. For any function $g:(\tlb,\theta^{\dag})\rightarrow\mathbb{R}$,
and any $\theta\in(\tlb,\theta^{\dag})$, $g_{+}(\theta)\equiv\lim_{y\downarrow\theta}g(y)$
and $g_{-}(\theta)\equiv\lim_{y\uparrow\theta}g(y)$. The proof is
in three steps.

\noindent\textbf{\textsc{Step A.}} We show that if $\bar{z}^{M}$
is discontinuous at $\theta\in(\tlb,\theta^{\dag})$, then $z^{M}$
is also discontinuous at $\theta$. Let $\varphi\equiv F^{\star}(\theta)$.
Because $F^{\star}$ is continuous and increasing, it suffices to
show that if $\bar{h}\equiv\bar{z}^{M}\circ(F^{\star})^{-1}$ is discontinuous
at $\varphi$, then $h\equiv z^{M}\circ(F^{\star})^{-1}$ is also
discontinuous at $\varphi$. To see this, suppose, towards a contradiction,
that $h$ is continuous at $\varphi$. Then $H$ is differentiable
at $\varphi$. On the other hand, the assumption that $\bar{h}$ is
discontinuous at $\varphi$ along with the fact that $\bar{h}$ is
right-continuous and weakly increasing implies that $\bar{h}_{-}(\varphi)<\bar{h}_{+}(\varphi)=\bar{h}(\varphi)$.
Note that this means that $\overline{H}$ is not differentiable at
$\varphi$.

\noindent Next observe that $\overline{H}(\varphi)=H(\varphi)$. To
see this, recall that, by definition, $\overline{H}(\varphi)\leq H(\varphi)$.
Now suppose that $\overline{H}(\varphi)<H(\varphi)$. The continuity
of $H$ and $\overline{H}$ then implies that there is an open interval
$I$ containing $\varphi$ such that, for all $\tilde{\varphi}\in I$,
$\overline{H}(\tilde{\varphi})<H(\tilde{\varphi})$. The pooling property
then implies that $\overline{H}$ is affine on $I$, contradicting
its non-differentiability at $\varphi$. Hence, $\overline{H}(\varphi)=H(\varphi)$.

\noindent Now observe that, because $\overline{H}\le H$ and $\overline{H}(\varphi)=H(\varphi)$,
for every $x>\varphi$, 
\[
\frac{\overline{H}(x)-\overline{H}(\varphi)}{x-\varphi}\le\frac{H(x)-H(\varphi)}{x-\varphi}.
\]
Taking the limit as $x\downarrow\varphi$, we have that the left-hand
side converges to $\bar{h}_{+}(\varphi)=\bar{h}(\varphi)$. The right-hand
side, instead, converges to $h(\varphi)$, as $H$ is differentiable
at $\varphi$. Thus, $\bar{h}(\varphi)\le h(\varphi)$. Similarly,
for every $x<\varphi$, 
\[
\frac{\overline{H}(\varphi)-\overline{H}(x)}{\varphi-x}\ge\frac{H(\varphi)-H(x)}{\varphi-x}.
\]
Taking the limit as $x\uparrow\varphi$, the left-hand side converges
to $\bar{h}_{-}(\varphi)$ whereas the right-hand side converges to
$h(\varphi)$. Hence $\bar{h}_{-}(\varphi)\ge h(\varphi)$. Therefore,
$\bar{h}_{-}(\varphi)\ge h(\varphi)\ge\bar{h}(\varphi)$. On the other
hand, the convexity of $\overline{H}$ implies that $\bar{h}_{-}(\varphi)\le\bar{h}(\varphi)$.
Hence $\bar{h}_{-}(\varphi)=\bar{h}(\varphi)=\bar{h}_{+}(\varphi)$.
But this contradicts the fact that, as argued above, $\bar{h}_{-}(\varphi)<\bar{h}_{+}(\varphi)=\bar{h}(\varphi)$.
We thus conclude that $h$ must be discontinuous at $\varphi$, which
implies discontinuity of $z^{M}$ at $\theta$, as claimed.

\noindent\textbf{\textsc{Step B.}} We now argue that the discontinuity
of $z^{M}$ at $\theta$ implies that $S(\theta;\qopt)=0$. To see
this, note that, because $z^{\star}$ and $f^{\star}$ are continuous,
the discontinuity must be due to an upward jump of $\Lambda^{{\rm OPT}}$
at $\theta$. The complementary slackness conditions in (\ref{eq:cs1})
then imply that $S(\theta;\qopt)=0$.

\noindent\textbf{\textsc{Step C.}} We now use the fact that $S(\cdot;\qopt)$
is continuous at $\theta$ to establish that $\qopt$ is continuous
at $\theta$. We first argue that $q^{{\rm OPT}}_{+}(\theta)\ge\underline{D}(\theta)$.
To see why, observe that the right-derivative of $S(\cdot;\qopt)$
at $\theta$ is equal to 
\[
\lim_{x\downarrow\theta}\frac{S(x;\qopt)-S(\theta;\qopt)}{x-\theta}=-\lim_{x\downarrow\theta}\frac{\int^{x}_{\theta}\left(\underline{D}(y)-\qopt(y)\right){\rm d}y}{x-\theta}=q^{{\rm OPT}}_{+}(\theta)-\underline{D}(\theta).
\]
That $\qopt$ satisfies the robustness constraints implies that $S(y;\qopt)\ge0$
for any $y$ in a neighborhood of $\theta$. Because $S(\theta;\qopt)=0$,
it must be that $S(\cdot;\qopt)$ has a minimum at $\theta$. In turn,
this means that the right derivative of $S(\cdot;\qopt)$ is nonnegative
at $\theta$, which implies that $q^{{\rm OPT}}_{+}(\theta)\ge\underline{D}(\theta)$.
Similarly, by considering the left derivative of $S(\cdot;\qopt)$
at $\theta$, we obtain that $\underline{D}(\theta)\ge q^{{\rm OPT}}_{-}(\theta)$.
Combining the last two inequalities, we have that $q^{{\rm OPT}}_{+}(\theta)\ge\underline{D}(\theta)\ge q^{{\rm OPT}}_{-}(\theta)$.
Because $\qopt$ is weakly decreasing, $q^{{\rm OPT}}_{-}(\theta)\ge q^{{\rm OPT}}_{+}(\theta)$.
This implies that $q^{{\rm OPT}}_{+}(\theta)=q^{{\rm OPT}}_{-}(\theta)=\underline{D}(\theta)$.
This last result, together with the fact that $\qopt$ is weakly decreasing,
implies that $\qopt$ is continuous at $\theta$, as claimed. $\hfill\blacksquare{}$

\medskip{}

\noindent\textbf{Proof of Lemma \ref{lem:iron-ae}}. First, we show
that $\theta_{\min}\le\theta_{2}$. Next, we show that, when $\theta_{\min}<\theta_{2}$,
$\bar{z}^{M}(\theta)>z^{\star}(\theta)$ for all $\theta\in(\theta_{\min},\theta_{2})\cap[\theta_{1},\theta_{2})$.
We then show that $\theta_{\min}=\theta^{m}$. Finally, we show that,
when $\theta_{2}<\theta^{\dag}$, $\bar{z}^{M}(\theta_{2})>z^{\star}(\theta_{2})$.
We proceed in five steps.

\noindent\textbf{\textsc{Step A.}} In this step, we establish that
$\theta_{\min}\le\theta_{2}$. The result follows from Lemma \ref{lem:tmin2}
(establishing that $\theta_{\min}\le\theta^{m}$) along with Observation
\ref{obs:inttheta12} (establishing that $\theta^{m}\le\theta_{2}$
if $\theta_{2}<\theta^{\dag}$) and the maintained assumption that
$\theta^{m}\leq\theta^{\dag}$.

\noindent\textbf{\textsc{Step B.}} Next, we establish that the weaker
inequality $\bar{z}^{M}(\theta)\geq z^{\star}(\theta)$ holds for
almost all $\theta\in(\theta_{1},\theta_{2})$. To do so, observe
that, because $\theta_{2}\le\theta^{\dag}\le\min\{\theta^{\star},\theta^{\star\star}\}$,
Lemma \ref{lem:weakqoptbound} in the main text implies that $\qopt(\theta)\le\qbm(\theta)$
for any $\theta\in(\tlb,\theta_{2})$. By (\ref{eq:iron2}), $\qopt(\theta)=D^{\star}(\bar{z}^{M}(\theta))$
for almost all $\theta\in(\theta_{1},\theta_{2})$. Combining the
two properties, we have that $D^{\star}(\bar{z}^{M}(\theta))\le D^{\star}(z^{\star}(\theta))$
for almost all $\theta\in(\theta_{1},\theta_{2})$. Because $D^{\star}$
is decreasing, this means that $\bar{z}^{M}(\theta)\ge z^{\star}(\theta)$
for almost all $\theta\in(\theta_{1},\theta_{2})$.

\noindent\textbf{\textsc{Step C.}} Equipped with the results in the
previous two steps, we now establish that, when $\theta_{\min}<\theta_{2}$,
the inequality $\bar{z}^{M}(\theta)\geq z^{\star}(\theta)$ is strict
for all $\theta\in(\theta_{\min},\theta_{2})\cap[\theta_{1},\theta_{2})$.
Let $\varphi_{\min}\equiv F^{\star}(\theta_{\min})$, $\varphi_{1}\equiv F^{\star}(\theta_{1})$
and $\varphi_{2}\equiv F^{\star}(\theta_{2})$. Next take any $\varphi\in(\varphi_{\min},\varphi_{2})\cap[\varphi_{1},\varphi_{2})$,
and let $\theta\equiv F^{\star-1}(\varphi)$. The proof considers
two cases.

\noindent\textbf{\textsc{Case C.1:}} There exists $\delta>0$ such
that $\overline{H}(\varphi')<H(\varphi')$ for all $\varphi'\in\tilde{I}\equiv(\varphi,\varphi+\delta)$. 

\noindent By the pooling property, $\bar{h}$ is constant on $\tilde{I}$.
Because $\bar{h}$ is right-continuous, this constant value is equal
to $\bar{h}(\varphi)$. Using $\bar{h}(\varphi)=\bar{z}^{M}((F^{\star})^{-1}(\varphi))=\bar{z}^{M}(\theta)$,
we thus have that, for all $\varphi'\in\tilde{I}$, $\bar{z}^{M}((F^{\star})^{-1}(\varphi'))=\bar{z}^{M}(\theta)$.
Now, choose $\epsilon>0$ such that $F^{\star}(\theta+\epsilon)\in\tilde{I}$
and $\bar{z}^{M}(\theta+\epsilon)\ge z^{\star}(\theta+\epsilon)$.
Note that, because $\bar{z}^{M}(\theta')\ge z^{\star}(\theta')$ for
almost all $\theta'\in(\theta_{1},\theta_{2})$ (by Step B), such
an $\epsilon$ exists. Then, $\bar{z}^{M}(\theta)=\bar{z}^{M}(\theta+\epsilon)\ge z^{\star}(\theta+\epsilon)>z^{\star}(\theta)$,
where the last inequality follows from the fact that $z^{\star}$
is increasing.

\noindent\textbf{\textsc{Case C.2.}} For every $\delta>0$, there
exists $\varphi'\in(\varphi,\varphi+\delta)\cap(\varphi,\varphi_{2})$
such that $\overline{H}(\varphi')=H(\varphi')$. 

\noindent Note that, in this case, $\overline{H}(\varphi)=H(\varphi)$.
If this was not true, the continuity of $\overline{H}$ and $H$ (which
follows from their definitions), would imply existence of a right
neighborhood $I'$ of $\varphi$ such that $\overline{H}(\varphi')<H(\varphi')$
for all $\varphi'\in I'$, contradicting the assumption defining Case
C.2. Hence, there exists a non-increasing sequence $(\varphi_{n})$
such that $\overline{H}(\varphi_{n})=H(\varphi_{n})$ for all $n$
and with $\varphi_{n}\rightarrow\varphi$ as $n\rightarrow\infty$.
By taking the right-derivatives of $\overline{H}$ and $H$ along
this sequence and using the fact that $\bar{h}$ and $h$ are right-continuous,
we have that 
\begin{align*}
\bar{h}(\varphi)=\lim_{n\rightarrow\infty}\frac{\overline{H}(\varphi_{n})-\overline{H}(\varphi)}{\varphi_{n}-\varphi}=\lim_{n\rightarrow\infty}\frac{H(\varphi_{n})-H(\varphi)}{\varphi_{n}-\varphi}=h(\varphi).
\end{align*}
Hence, $\bar{z}^{M}(\theta)=\bar{h}(\varphi)=h(\varphi)=z^{M}(\theta)>z^{\star}(\theta)$,
where the last inequality follows because $\theta>\theta_{\min}$,
and therefore, $\Lambda^{{\rm OPT}}(\theta)+\muopt>0$.

Combining Case C.1 with Case C.2, we thus conclude that, for any $\theta\in(\theta_{\min},\theta_{2})\cap[\theta_{1},\theta_{2})$,
$\bar{z}^{M}(\theta)>z^{\star}(\theta)$, as claimed.

\noindent\textbf{\textsc{Step D.}} Next, we establish that $\theta^{m}=\theta_{\min}$.
Assume, towards a contradiction, that $\theta^{m}\ne\theta_{\min}$.
Lemma \ref{lem:tmin2} (establishing that $\theta_{\min}\le\theta^{m}$)
then implies that $\tlb\leq\theta_{\min}<\theta^{m}$. By Lemma \ref{lem:leftofthetam}
in the main text, 
\begin{align}
\qopt(\theta)=\qbm(\theta)~\qquad~\forall~\theta\in(\tlb,\theta^{m}).\label{eq:temp1-2}
\end{align}
Hence, Observation \ref{obs:inttheta12} (establishing that $\qopt(\theta)<\qbm(\theta)$
for all $\theta\in(\tlb,\theta_{1})$ if $\theta_{1}>\tlb$) implies
that $\theta_{1}=\tlb$. The arguments in Step A above establish that
$\theta^{m}\le\theta_{2}$. Hence, $(\theta_{\min},\theta^{m})\subset(\theta_{1},\theta_{2})$.
Step C above then implies $\bar{z}^{M}(\theta)>z^{\star}(\theta)$
for all $\theta\in(\theta_{\min},\theta^{m})$. Condition (\ref{eq:iron2})
in turn implies that there exists $\theta'\in(\theta_{\min},\theta^{m})$
such that $\qopt(\theta')=D^{\star}(\bar{z}^{M}(\theta'))$. Combining
the two properties, we conclude that $\qopt(\theta')=D^{\star}(\bar{z}^{M}(\theta'))<D^{\star}(z^{\star}(\theta'))=\qbm(\theta')$,
which contradicts (\ref{eq:temp1-2}). Hence, $\theta^{m}=\theta_{\min}$.

\noindent\textbf{\textsc{Step E.}} We conclude by showing that, when
$\theta_{2}<\theta^{\dag}$, $\bar{z}^{M}(\theta_{2})>z^{\star}(\theta_{2})$.
First, we argue that $\qopt$ is continuous at $\theta_{2}$. If $\bar{z}^{M}$
is discontinuous at $\theta_{2}$, this follows from Lemma \ref{lem:barz-continuous-middle}.
If, instead, $\bar{z}^{M}$ is continuous at $\theta_{2}$, the continuity
of $\qopt$ at $\theta_{2}$ follows from the following arguments.
Because $D^{\star}$ is continuous, $D^{\star}(\bar{z}^{M}(\cdot))$
is also continuous at $\theta_{2}$. Then consider an increasing sequence
$(\theta_{n})$ such that $\theta_{n}\in E$ for all $n$ and $\lim_{n\rightarrow\infty}\theta_{n}=\theta_{2}$.
Along such a sequence $\qopt(\theta_{n})=D^{\star}(\bar{z}^{M}(\theta_{n}))$
and hence 
\begin{align*}
\lim_{n\rightarrow\infty}\qopt(\theta_{n})=\lim_{n\rightarrow\infty}D^{\star}(\bar{z}^{M}(\theta_{n}))=D^{\star}(\bar{z}^{M}(\theta_{2})),
\end{align*}
where the last equality follows from continuity of $D^{\star}(\bar{z}^{M}(\cdot))$.
Furthermore, because $D^{\star}(\bar{z}^{M}(\cdot))$ is continuous,
the definition of $\theta_{2}$ along with (\ref{eq:iron1}) implies
that $D^{\star}(\bar{z}^{M}(\theta_{2}))=\qopt(\theta^{\dag})$. The
last two properties jointly imply that $\lim_{n\rightarrow\infty}\qopt(\theta_{n})=\qopt(\theta^{\dag})$.
Because $\qopt$ is weakly decreasing and $\qopt(\theta)=\qopt(\theta^{\dag})$
for all $\theta\in(\theta_{2},\theta^{\dag})$, we have that $\qopt(\theta_{2})=\qopt(\theta^{\dag})$
and hence $\qopt$ is continuous at $\theta_{2}$.

\noindent Next, observe that there exists $\epsilon>0$ such that
\[
D^{\star}(\bar{z}^{M}(\theta_{2}))\leq\qopt(\theta^{\dag})=\qopt(\theta_{2})=\qopt(\theta_{2}+\epsilon)<\qbm(\theta_{2}+\epsilon)<\qbm(\theta_{2})=D^{\star}(z^{\star}(\theta_{2})).
\]
The first inequality follows from the definition of $\theta_{2}$,
along with the property that $\bar{z}^{M}$ is right-continuous. The
first and second equalities follow from the continuity of $\qopt$
at $\theta_{2}$, along with Condition (\ref{eq:iron1}). The second
inequality follows from Observation \ref{obs:inttheta12}. The third
inequality follows from the fact that $\qbm$ is decreasing. The last
equality follows from the definition of the $\qbm$ schedule. Because
$D^{\star}$ is decreasing, we conclude that $\bar{z}^{M}(\theta_{2})>z^{\star}(\theta_{2})$,
as claimed $\hfill\blacksquare{}$

This completes the proof of Proposition \ref{prop:part2_strong}. $\hfill Q.E.D.$

\subsection{More permissive robustness constraints}
\label{sec:gamma}

We consider more permissive robustness constraints whereby the designer
can choose over mechanisms whose guarantee is not too small relative
to the maximal one. Formally, the shortlist of approved mechanisms
is $\mathcal{M}^{{\rm SL}}_{\gamma}\equiv\{M\in\mathcal{M}:G(M)\ge\gamma G(M')~\forall~M'\in\mathcal{M}\}$,
where $\gamma\in[0,1]$. The analysis in the main text corresponds
to the case $\gamma=1$. Here, we consider the case $\gamma\in[0,1)$.
The proofs of all results in this section are omitted as they build
on arguments similar to those for the corresponding results in the
main text and are available upon request. Lemma \ref{lemma-guarantee}
in the main text remains unchanged. Lemma \ref{lemma-SL-characterization}
instead is generalized as follows:

\begin{lemma}[Shortlist characterization] \label{lemma-SL-characterization-gamma}Take
any IC and IR mechanism $M\equiv(q,u)\in\mathcal{M}$. Then $M\in\mathcal{M}^{{\rm SL}}_{\gamma}$
if and only if 
\begin{align}
\underline{V}(q(\theta))-\theta q(\theta)-u(\theta)\ge\gamma G^{*}\qquad\forall~\theta\in\Theta.\label{SL-constraint-gamma}
\end{align}
\end{lemma}

There are two important distinctions from the case $\gamma=1$ in
the main text. First, when $\gamma<1$ the rent $u(\overline{\theta})$
for the least efficient type is not necessarily zero. Second, $q(\overline{\theta})$
is not restricted to be equal to ${\rm q}_{\ell}$. An immediate implication
of the last lemma is that the shortlist for $\gamma'>\gamma$ is
a subset of the shortlist for $\gamma$. Below, we focus on mechanisms
in $\mathcal{M}^{{\rm SL}}_{\gamma}$ for which $u(\overline{\theta})=0$.
The reason for focusing on this subset of $\mathcal{M}^{{\rm SL}}_{\gamma}$
is that robustly optimal mechanisms always feature no rent for the
highest-cost type. The next lemma generalizes Lemma \ref{lemma:majorizations}
in the main text.

\noindent\begin{lemma} \label{lemma:majorizations-gamma} Take any
weakly decreasing function $q:\Theta\rightarrow\mathbb{R}_{+}$. (1)
For any $\theta\in\Theta$, constraint (\ref{SL-constraint-gamma})
is equivalent to 
\begin{align}
\int^{\overline{\theta}}_{\theta}q(y)dy\leq\int^{\overline{\theta}}_{\theta}\underline{D}(y)dy-\int^{\underline{P}(q(\theta))}_{\theta}\left(\underline{D}(y)-q(\theta)\right)dy+(1-\gamma)G^{*}.\label{eq:general-majorization}
\end{align}
(2) The following statements are equivalent: (a) the inequality in
(\ref{eq:general-majorization}) holds for all $\theta\in\Theta$;
(b) the inequality in (\ref{eq:general-majorization}) holds for $\theta\in\{\underline{\theta},\overline{\theta}\}$
and, for all $\theta\in(\underline{\theta},\overline{\theta})$, 
\begin{align*}
\int^{\overline{\theta}}_{\theta}q(y)dy\leq\int^{\overline{\theta}}_{\theta}\underline{D}(y)dy+(1-\gamma)G^{*}.
\end{align*}
\end{lemma}

We now modify the definition of Baron-Myerson-with-quantity-floor
mechanism to account for the more permissive shortlist. Because the
function $\underline{V}({\rm q})-\overline{\theta}{\rm q}$ is concave
in ${\rm q}$, there are two solutions to the equation $\underline{V}({\rm q})-\overline{\theta}{\rm q}=\gamma G^{*}$.
Denote these two solutions by ${\rm \underline{q}}^{\gamma}_{\ell}$
and ${\rm \bar{q}}^{\gamma}_{\ell}$, with ${\rm \bar{q}}^{\gamma}_{\ell}>{\rm \underline{q}}^{\gamma}_{\ell}$.
The generalization of Baron-Myerson-with-quantity-floor mechanism
uses the lowest of the two floors, ${\rm \underline{q}}^{\gamma}_{\ell}$.
This is because, as shown in Proposition \ref{prop:robust-quantity-mechanism-general-gamma}
below, the relevant constraint imposed by robustness is the one requiring
that $q(\theta)\geq{\rm \underline{q}}^{\gamma}_{\ell}$ for all $\theta\in\Theta$.

\begin{defn} The \textbf{Baron-Myerson-with-quantity-floor} mechanism
is the mechanism $M^{\star}_{\gamma}\equiv(q^{\star}_{\gamma},u^{\star}_{\gamma})$
where $q^{\star}_{\gamma}$ is the quantity schedule defined, for
all $\theta$, by 
\begin{equation*}
q^{\star}_{\gamma}(\theta)\equiv\max\{q^{{\rm {BM}}}(\theta),{\rm \underline{q}}^{\gamma}_{\ell}\}
\end{equation*}
and where $u^{\star}_{\gamma}$ is given by $u^{\star}_{\gamma}(\theta)=\intop^{\overline{\theta}}_{\theta}q^{\star}_{\gamma}(y)dy$
for all $\theta$.\end{defn}

\begin{proposition}[Optimality of Baron-Myerson-with-quantity-floor]
\label{prop:BM-floor-gamma} The mechanism $M^{\star}_{\gamma}$ is
robustly optimal if and only if the following conditions jointly hold:
\begin{align*}
\int^{\overline{\theta}}_{\underline{\theta}}q^{\star}_{\gamma}(y)dy\leq\int^{\overline{\theta}}_{\underline{\theta}}\underline{D}(y)dy-\int^{\underline{P}(q^{\star}_{\gamma}(\underline{\theta}))}_{\underline{\theta}}\Big[\underline{D}(y)-q^{\star}_{\gamma}(\underline{\theta})\Big]dy+(1-\gamma)G^{*},
\end{align*}
\begin{align*}
\int^{\overline{\theta}}_{\theta}q^{\star}_{\gamma}(y)dy & \le\int^{\overline{\theta}}_{\theta}\underline{D}(y)dy+(1-\gamma)G^{*}\qquad\forall\theta>\underline{\theta}.
\end{align*}
\end{proposition}

Let $\theta^{\star}_{\gamma}$ be the threshold defined as follows.
If $\qbm(\overline{\theta})\le{\rm \underline{q}}^{\gamma}_{\ell}$,
by continuity of $\qbm$ along with the fact that $\qbm(\underline{\theta})>{\rm \underline{q}}^{\gamma}_{\ell}$,
let $\theta^{\star}_{\gamma}$ be the unique solution to $\qbm(\theta^{\star}_{\gamma})={\rm \underline{q}}^{\gamma}_{\ell}$.
If, instead, $\qbm(\overline{\theta})>{\rm \underline{q}}^{\gamma}_{\ell}$
(i.e., if $\qbm$ never crosses ${\rm \underline{q}}^{\gamma}_{\ell}$),
then let $\theta^{\star}_{\gamma}\equiv\overline{\theta}$. Next let
\begin{align*}
\theta^{m}_{\gamma} & \equiv\max\{\theta:\theta\in\arg\min_{y\in\Theta}\underline{W}(y,q^{\star}_{\gamma})\},
\end{align*}
where, as in the main text, given any function $q$, for all $\theta\in\Theta$,
$\underline{W}(\theta,q)\equiv\underline{V}(q(\theta))-\theta q(\theta)-\int^{\overline{\theta}}_{\theta}q(y)dy$.

The following proposition establishes a result similar to the one
in Proposition \ref{prop:robust-quantity-mechanism-general} in the
main text:\footnote{The difference between Proposition \ref{prop:robust-quantity-mechanism-general}
in the main text (with the stronger version of part (2) established
in Section \ref{subsec:Stronger-Prop-2} in this supplement) and Proposition
\ref{prop:robust-quantity-mechanism-general-gamma} stems from
the fact that, when $\gamma<1$, constraint qualification of the Lagrangian
methods used to establish that $\qopt(\theta)<\qbm(\theta)$ for \emph{all}
$\theta\in(\theta^{m}_{\gamma},\theta^{\star}_{\gamma})$ holds irrespective
of whether $\theta^{\star}_{\gamma}<\overline{\theta}$ or $\theta^{\star}_{\gamma}=\overline{\theta}$.}

\begin{proposition} \label{prop:robust-quantity-mechanism-general-gamma}
If Baron-Myerson-with-quantity-floor is not robustly optimal, then
$\theta^{m}_{\gamma}<\theta^{\star}_{\gamma}$, and every robustly
optimal mechanism $M^{{\rm {OPT}}}=(q^{{\rm {OPT}}},u^{{\rm {OPT}}})$
is such that $\qopt(\theta)=q^{{\rm BM}}(\theta)$ for all $\theta\in(\underline{\theta},\theta^{m}_{\gamma}]$,
$\qopt(\theta)<\qbm(\theta)$ for all $\theta\in(\theta^{m}_{\gamma},\theta^{\star}_{\gamma})$,
and $\qopt(\theta)={\rm \underline{q}}^{\gamma}_{\ell}$ for all $\theta\in(\theta^{\star}_{\gamma},\overline{\theta}]$.
\end{proposition}

Next, consider price regulation. One can then show that every robustly
optimal price regulation is such that, for all $\theta>\underline{\theta},$
the price is as in the Baron-Myerson-with-price-cap mechanism. That
is, $p^{\mathrm{OPT}}_{\gamma}(\theta)=\min(z^{\star}(\theta),\underline{P}(\underline{{\rm q}}^{\gamma}_{\ell}))$,
with $\underline{P}(\underline{{\rm q}}^{\gamma}_{\ell})>\overline{\theta}$.
Finally, one can show that a result analogous to Proposition \ref{prop:comp_reg}
in the main text (which compares price and quantity regulation)
holds. However, the conditions for quantity regulation to dominate
price regulation are more stringent because of the larger price-cap. 

\begin{proposition}[quantity vs price regulation]\label{prop:qty-vs-price-general-gamma}
(1) If the Baron-Myerson-with-quantity-floor mechanism $M^{\star}_{\gamma}$
is robustly optimal, quantity regulation dominates price regulation
(strictly if $D^{\star}(\underline{P}(\underline{{\rm q}}^{\gamma}_{\ell}))>q^{\star}_{\gamma}(\bar{\theta})$).
(2) If the Baron-Myerson-with-quantity-floor mechanism $M^{\star}_{\gamma}$
is not robustly optimal and $D^{\star}(\underline{P}(\underline{{\rm q}}^{\gamma}_{\ell}))=q^{\star}_{\gamma}(\bar{\theta})$,
price regulation strictly dominates quantity regulation.

\end{proposition}

\subsection{More general forms of technological uncertainty}

\label{sec:Arbitrary-Cost-Uncertainty}

In this section, we consider the case in which the set of feasible
technologies $\mathcal{F}$ need not coincide with the entire set
$\textsc{CDF}(\Theta)$ of cdfs supported on $\Theta$. Namely, we
assume there exists a cdf $\underline{F}\in\textsc{CDF}(\Theta)$
such that $\mathcal{F}$ is the set of all cdfs $F\in\textsc{CDF}(\Theta)$
with $F(\theta)\ge\underline{F}(\theta)$ for all $\theta\in\Theta$.
In the main text, we assumed $\underline{F}(\theta)$ is the Dirac
distribution that puts unit mass at $\overline{\theta}$. Here, instead,
we allow $\underline{F}$ to have support $[\theta_{s},\overline{\theta}]$,
where $\theta_{s}\in(\underline{\theta},\overline{\theta})$.\footnote{The results below do not hinge on the fact that the largest element
of the support of $\underline{F}$ coincides with that of $F^{\star}$,
i.e., with $\bar{\theta}$. They extend to the case where the support
of $\underline{F}$ is $[\theta_{s},\theta^{H})$, with $\theta^{H}\geq\bar{\theta}$
possibly equal to $\theta^{H}=+\infty$.}

\begin{defn}The cdf $\underline{F}$ is regular with respect to $\theta_{s}$
if it is absolutely continuous over $\mathbb{R}$ with density $f(\theta)>0$
if and only if $\theta\in[\theta_{s},\overline{\theta}]$ and with
$\underline{z}(\theta)\equiv\theta+\underline{F}(\theta)/\underline{f}(\theta)$
continuous and increasing over $[\theta_{s},\overline{\theta}]$.
\end{defn}

Let $\underline{M}^{\textsc{BM}}=(\underline{q}^{\textsc{BM}},\underline{u}^{\textsc{BM}})\in\mathcal{M}$
denote an arbitrary IC and IR mechanism that is optimal under the
model $(\underline{V},\underline{F})$, in the usual Bayesian sense.
Note that such a mechanism is not unique, but in any such mechanism
$\underline{q}^{\textsc{BM}}$ is weakly decreasing, $\underline{u}^{\textsc{BM}}(\bar{\theta})=0$,
and $\underline{u}^{\textsc{BM}}(\theta)=\intop^{\overline{\theta}}_{\theta}\underline{q}^{\textsc{BM}}(y)dy$
for all $\theta$. Then let 
\begin{align*}
G^{*}_{s}\equiv\int^{\overline{\theta}}_{\theta_{s}}\underline{W}(\theta,\underline{q}^{\textsc{BM}})\underline{F}(\text{d}\theta),
\end{align*}
where, as in the main text, for any weakly decreasing functions $q$
and all $\theta\in\Theta$, $\underline{W}(\theta,q)\equiv\underline{V}(q(\theta))-\theta q(\theta)-\int\limits^{\overline{\theta}}_{\theta}q(y)dy$.
Thus, $G^{*}_{s}$ denotes the buyer's expected welfare under the
mechanism $\underline{M}^{\textsc{BM}}$ when the gross value function
is $\underline{V}$ and the technology is $\underline{F}$. Formally,
$W(\underline{M}^{\textsc{BM}};\underline{V},\underline{F})=G^{*}_{s}$.
When $\theta_{s}=\bar{\theta}$, as in the main text, $G^{*}_{s}=G^{*}$.
Finally, for any IC and IR mechanism $M=(q,u)\in\mathcal{M}$, let
$\underline{w}_{q}\equiv\inf_{\theta\le\theta_{s}}\underline{W}(\theta,q)$.
The next proposition generalizes Lemma \ref{lemma-SL-characterization}
in the main text by providing a complete characterization of the short
list.

\begin{proposition}\label{prop:SL_local_uncertainty} Suppose $\underline{F}$
is regular with respect to $\theta_{s}$. (1) For any $M\in\mathcal{M}^{{\rm SL}}$,
$G(M)=G^{*}_{s}$. (2) A mechanism $M\equiv(q,u)\in\mathcal{M}^{{\rm SL}}$
if and only if the following conditions jointly hold: (a) $q$ is
weakly decreasing; (b) for all $\theta$, $u(\theta)=\intop^{\overline{\theta}}_{\theta}q(y)dy$;
(c) $q(\theta)=\underline{q}^{\textsc{BM}}(\theta)$ for all $\theta\in(\theta_{s},\overline{\theta})$;
(d) $\underline{V}(q(\theta))-\theta q(\theta)-\int^{\overline{\theta}}_{\theta}q(y)dy\geq G^{*}_{s}$
for all $\theta\le\theta_{s}$; and (e) $\underline{w}_{q}\underline{F}(\theta)+\int^{\overline{\theta}}_{\theta}\underline{W}(y,\underline{q}^{\textsc{BM}})\underline{F}(\textrm{d}y)\ge G^{*}_{s}$
for all $\theta\in[\theta_{s},\overline{\theta}]$.

\end{proposition}

The formal proof of the proposition is below. Before getting to it,
we sketch a few key arguments. Part (1) follows from the fact that
Nature can always pick $(\underline{V},\underline{F})$, which implies
that the guarantee of any IC and IR mechanism is bounded above by
the maximal welfare attainable under the lowest gross value function
$\underline{V}$ and the worst technology $\underline{F}$. This upper
bound on the guarantee can be achieved by offering a mechanism $M\equiv(q,u)$
in which $q(\theta)=\underline{q}^{\textsc{BM}}(\theta_{s})$ for
all $\theta\leq\theta_{s}$, $q($$\theta)=\underline{q}^{\textsc{BM}}(\theta)$
for all $\theta>\theta_{s}$, and $u(\theta)=\int^{\overline{\theta}}_{\theta}q(y)dy$
for all $\theta$. Against such a mechanism an adversarial Nature
cannot do better than selecting $(V,F)=(\underline{V},\underline{F})$,
yielding the buyer a payoff of $G^{*}_{s}$. Conditions (a)-(d) in
Part (2) are generalizations of the robustness constraints in Lemma
\ref{lemma-SL-characterization} in the main text. Condition (c) follows
from the fact that Nature can always choose the model $(\underline{V},\underline{F})$
in which case worst-case optimality dictates that the output procured
from types $\theta\in(\theta_{s},\overline{\theta})$ be the same
as under the mechanism $\underline{M}^{\textsc{BM}}\equiv(\underline{q}^{\textsc{BM}},\underline{u}^{\textsc{BM}})$.
Note that Condition (d) applies only to types $\theta\le\theta_{s}$
as Nature can now select a Dirac distribution at $\theta$ only when
$\theta\le\theta_{s}$. Condition (e) stems from the fact that Nature's
best response can be a non-Dirac distribution. To understand this
constraint, consider Figure \ref{fig:new_tech}. Let $\theta_{1}<\theta_{s}$
be a cost level at which the function $\underline{W}(\cdot,q)$ reaches
the minimum over $[\underline{\theta},\theta_{s}]$, i.e., $\underline{W}(\theta_{1},q)=\underline{w}_{q}$,
and let $\theta_{2}>\theta_{s}$ be a cost level such that $\underline{W}(\theta_{2},q)=\underline{W}(\theta_{2},\underline{q}^{\textsc{BM}})=\underline{w}_{q}$.
Suppose Nature picks a distribution $F_{1}$ with an atom at $\theta_{1}$
equal to $\underline{F}(\theta_{2})$ and which agrees with $\underline{F}$
on all $\theta\geq\theta_{2}$. Because $\underline{W}(\cdot,\underline{q}^{\textsc{BM}})$
is decreasing, even if $\underline{w}_{q}\geq G^{*}_{s}$, in the
absence of the new constraint in Condition (e), it may well be the
case that 
\[
\underline{w}_{q}F_{1}(\theta_{1})+\int^{\overline{\theta}}_{\theta_{2}}\underline{W}(\theta,\underline{q}^{\textsc{BM}})F_{1}(\textrm{d}\theta)=\underline{w}_{q}\underline{F}(\theta_{2})+\int^{\overline{\theta}}_{\theta_{2}}\underline{W}(\theta,\underline{q}^{\textsc{BM}})\underline{F}(\textrm{d}\theta)<G^{*}_{s},
\]
in which case the mechanism's guarantee is below $G^{*}_{s}$. To
account for such possibilities, worst-case optimality requires that
the constraint in Condition (e) holds.

\begin{figure}
\centering %
\begin{minipage}[b]{0.5\linewidth}%
\includegraphics[scale=0.54]{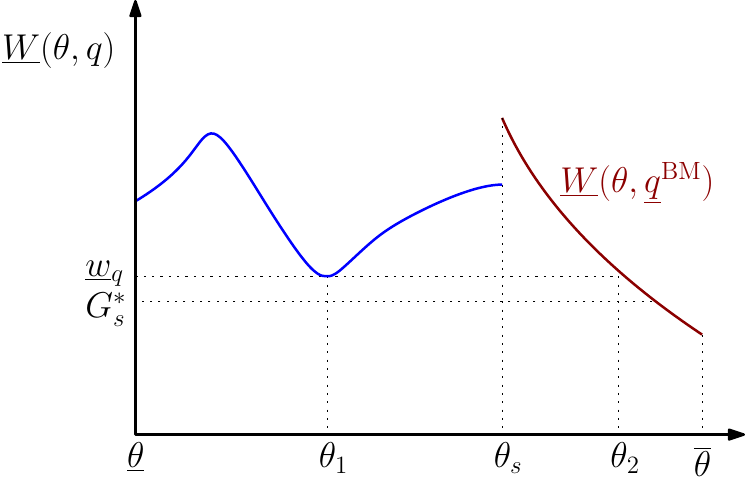} \caption*{\centering A: Function $\underline{W}(\cdot,q)$.}
\end{minipage}%
\begin{minipage}[b]{0.5\linewidth}%
\includegraphics[scale=0.54]{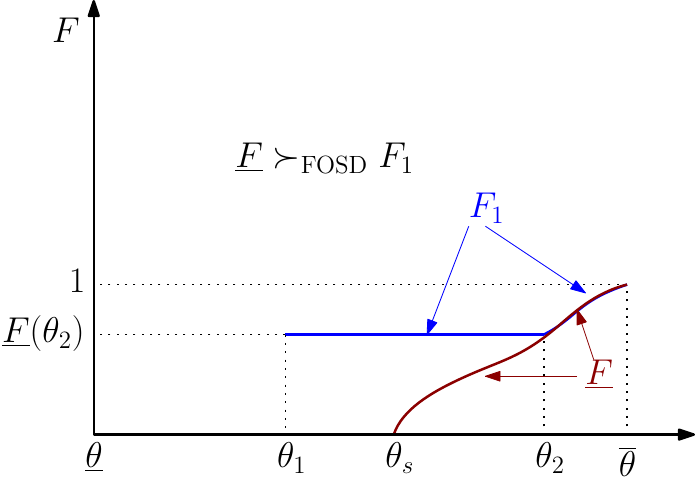} \caption*{\centering B: Technology $F_{1}$ yielding welfare below $G^{*}_{s}$.}
\end{minipage}\caption{Graphical illustration of robustness constraint in (e) of Proposition
\ref{prop:SL_local_uncertainty}.}
\label{fig:new_tech} 
\end{figure}

\noindent\textbf{Proof of Proposition \ref{prop:SL_local_uncertainty}}.
\textbf{Part 1}. For any IC and IR mechanism $M=(q,u)\in\mathcal{M}$,
\begin{align}
G(M)=\inf_{(V,F)\in\mathcal{V}\times\mathcal{F}}W(M;V,F)\leq W(M;\underline{V},\underline{F})\leq_{(i)}W(\underline{M}^{\textsc{BM}};\underline{V},\underline{F})=G^{*}_{s}.
\label{eq-general-guarantee}
\end{align}
Inequality $(i)$ holds because $\underline{M}^{\textsc{BM}}$ maximizes
$W(\cdot;\underline{V},\underline{F})$ over $\mathcal{M}$. We now
show that there exists an IC and IR mechanism $\underline{M}\in\mathcal{M}$
such that $G(\underline{M})=G^{*}_{s}$. Let $\underline{M}\equiv(\underline{q},\underline{u})$
be the mechanism in which 
\[
\underline{q}(\theta)=\begin{cases}
\underline{q}^{\textsc{BM}}(\theta_{s}) & \mbox{if}~\theta<\theta_{s}\\
\underline{q}^{\textsc{BM}}(\theta) & \mbox{if}~\theta\ge\theta_{s},
\end{cases}
\]
and $\underline{u}(\theta)=\int^{\overline{\theta}}_{\theta}\underline{q}(y)dy$
for all $\theta$. Clearly, $\underline{M}\in\mathcal{M}$. Furthermore,
because $\underline{F}$ has support $[\theta_{s},\overline{\theta}]$,
$W(\underline{M};\underline{V},\underline{F})=W(\underline{M}^{\textsc{BM}};\underline{V},\underline{F})$.
Now, recall that $\underline{M}^{\textsc{BM}}=(\underline{q}^{\textsc{BM}},\underline{u}^{\textsc{BM}})$
is the optimal mechanism for the model $(\underline{V},\underline{F})$.
When $\underline{F}$ is regular with respect to $\theta_{s}$, $\underline{q}^{\textsc{BM}}$
is such that, for all $\theta\ge\theta_{s}$, $\underline{q}^{\textsc{BM}}(\theta)=\underline{D}(\underline{z}(\theta))$,
where, for all $\theta\ge\theta_{s}$, $\underline{z}(\theta)\equiv\theta+\underline{F}(\theta)/\underline{f}(\theta)$.
Thus, $\underline{q}(\theta)\le\underline{D}(\theta)$ for all $\theta$,
with the inequality strict for $\theta\neq\theta_{s}$. Part A of Lemma \ref{Lem-Mono} in the main text then implies that
$\underline{W}(\cdot,\underline{q})$ is weakly decreasing over $\Theta$.
Furthermore, because, for all $F\in\mathcal{F}$, $\underline{F}\succ_{FOSD}F$,
we have that $W(\underline{M};\underline{V},F)\geq W(\underline{M};\underline{V},\underline{F})$.
Because, for any $V\in\mathcal{V}$ and any $F\in\mathcal{F}$ , $W(\underline{M};V,F)\geq W(\underline{M};\underline{V},F)$,
we thus have that $W(\underline{M};V,F)\geq W(\underline{M};\underline{V},\underline{F})$.
We conclude that $G(\underline{M})=W(\underline{M}^{\textsc{BM}};\underline{V},\underline{F})=G^{*}_{s}$.

\noindent\textbf{Part 2: Necessity}. If $M=(q,u)\in\mathcal{M}^{\textsc{SL}}$,
then $M$ is IC and IR, and, therefore, $q$ is weakly decreasing
and $u(\theta)=u(\overline{\theta})+\int^{\overline{\theta}}_{\theta}q(y)dy$
for all $\theta$. Furthermore, by the result in Part 1, it must be
that $G(M)=G^{*}_{s}$. Hence, $u(\overline{\theta})=0$.

Recall that, for any $\theta\geq\theta_{s}$, $\underline{q}^{\textsc{BM}}(\theta)=\arg\max_{\mathrm{q}\in[0,\mathrm{\bar{q}}]}\{\underline{V}(\mathrm{q})-\underline{z}(\theta)\mathrm{q}\}$.
If $q(\theta)\neq\underline{q}^{\textsc{BM}}(\theta)$ for a positive
Lebesgue measure subset of $[\theta_{s},\bar{\theta}]$, then inequality
$(i)$ in (\ref{eq-general-guarantee}) is strict and hence $W(M;\underline{V},\underline{F})<W(\underline{M}^{\textsc{BM}};\underline{V},\underline{F})=G^{\star}_{s}$.
This means that $G(M)<G^{*}_{s}$ and hence $M\notin\mathcal{M}^{\textsc{SL}}$,
a contradiction. Because $\underline{q}^{\textsc{BM}}$ is weakly
decreasing and continuous over $[\theta_{s},\bar{\theta}]$, we conclude
that $q(\theta)=\underline{q}^{\textsc{BM}}(\theta)$ for all $\theta\in(\theta_{s},\bar{\theta})$.

Next, observe that, for any $\theta<\theta_{s}$, $\mathcal{F}$ contains
a distribution $F$ corresponding to a Dirac measure at $\theta<\theta_{s}$
(indeed, $\underline{F}\succ_{{\rm FOSD}}F$). Welfare under the lowest
gross value function $\underline{V}$ and such an $F$ is $\underline{V}(q(\theta))-\theta q(\theta)-\int^{\overline{\theta}}_{\theta}q(y)dy=\underline{W}(\theta,q)$.
Hence, it must be that $\underline{W}(\theta,q)\ge G^{*}_{s}$.

Finally, observe that, by Condition (c), the inequality in the constraint
in Condition (e) holds with an equality for $\theta=\theta_{s}$.
Suppose there exists $\tilde{\theta}\in(\theta_{s},\overline{\theta}]$
such that 
\begin{align}
{\underline{w}}_{q}\underline{F}(\tilde{\theta})+\int^{\overline{\theta}}_{\tilde{\theta}}\underline{W}(y,\underline{q}^{\textsc{BM}})\underline{F}(\textrm{d}y)<G^{*}_{s}.\label{eq:temp1-1}
\end{align}
By definition of $\underline{w}_{q}$, there exists $\theta'\le\theta_{s}$
such that $\underline{W}(\theta',q)$ is arbitrarily close to ${\underline{w}}_{q}$.
Let $\widetilde{F}$ be the cdf given by 
\[
\widetilde{F}(\theta)=\begin{cases}
0 & \textrm{if}~\theta<\theta'\\
\underline{F}(\tilde{\theta}) & \textrm{if}~\theta\in[\theta',\tilde{\theta})\\
\underline{F}(\theta) & \textrm{if}~\theta\ge\tilde{\theta}.
\end{cases}
\]
Clearly, $\widetilde{F}(\theta)\ge\underline{F}(\theta)$ for all
$\theta$, and hence, $\widetilde{F}\in\mathcal{F}$. Welfare under
the mechanism $M=(q,u)$ when Nature selects the model $(\underline{V},\widetilde{F})$
is equal to 
\begin{align*}
W(M;\underline{V},\widetilde{F})=\underline{W}(\theta',q)\underline{F}(\tilde{\theta})+\int^{\overline{\theta}}_{\tilde{\theta}}\underline{W}(y,\underline{q}^{\textsc{BM}})\underline{F}(\textrm{d}y)<G^{*}_{s},
\end{align*}
where the inequality follows from Condition (\ref{eq:temp1-1}) and
the fact that $\underline{W}(\theta',q)$ is arbitrarily close to
${\underline{w}}_{q}$. This, however, is a contradiction to $M\in\mathcal{M}^{\textsc{SL}}$.

We conclude that properties (a)-(e) are jointly necessary for any
$M\in\mathcal{M}^{\textsc{SL}}$.

\noindent\textbf{Part 2: Sufficiency}. Take any mechanism $M$ satisfying
properties (a)-(e). By virtue of (a) and (b), $M$ is IC and IR. By
virtue of the result in Part 1 of the proposition, it thus suffices
to show that $W(M;V,F)\ge G^{*}_{s}$ for any model $(V,F)\in\mathcal{V}\times\mathcal{F}$.
First, suppose $F$ is a Dirac distribution on some $\theta\le\theta_{s}$.
Then, Condition (d) implies that 
\begin{align*}
W(M;V,F)\ge W(M;\underline{V},F)=\underline{V}(q(\theta))-\theta q(\theta)-\int^{\overline{\theta}}_{\theta}q(y)dy\ge G^{*}_{s}.
\end{align*}
Next consider any model $(V,F)\in\mathcal{V}\times\mathcal{F}$ in
which $F$ puts positive probability on $\theta<\theta_{s}$. Then,\footnote{The integration in any expression of the form $\underline{w}_{q}F(\theta)+\int^{\overline{\theta}}_{\theta}\underline{W}(\theta,\underline{q}^{{\rm BM}})F(\mathrm{d}\theta)$,
where $\theta\ge\theta_{s}$, should be interpreted to be over $(\theta,\overline{\theta}]$
to avoid the double counting of any atom at $\theta$.} 
\begin{align}
W(M;V,F) & \ge W(M;\underline{V},F)\ge\underline{w}_{q}F(\theta_{s})+\int^{\overline{\theta}}_{\theta_{s}}\underline{W}(\theta,q)F(\mathrm{d}\theta)=\underline{w}_{q}F(\theta_{s})+\int^{\overline{\theta}}_{\theta_{s}}\underline{W}(\theta,\underline{q}^{{\rm BM}})F(\mathrm{d}\theta),\label{eq:nn0}
\end{align}
where the second inequality follows from the definition of $\underline{w}_{q}$,
whereas the equality follows from the fact that $\ensuremath{q(\theta)=\underline{q}^{\textsc{BM}}(\theta)}$
for all $\theta\in(\theta_{s},\bar{\theta})$. Now partition $[\theta_{s},\overline{\theta}]$
into $\Theta_{1}\equiv\{\theta\in[\theta_{s},\overline{\theta}]:\underline{W}(\theta,\underline{q}^{{\rm BM}})\leq\underline{w}_{q}\}$
and $\Theta_{2}\equiv[\theta_{s},\overline{\theta}]\setminus\Theta_{1}$.
Note that $\underline{W}(\cdot,\underline{q}^{{\rm BM}})$ is decreasing
over $[\theta_{s},\overline{\theta}]$ and hence $\Theta_{1}$ is
a (possibly empty) interval. If $\Theta_{1}$ is empty, let $\hat{\theta}\equiv\theta_{s}$.
Else, let $\hat{\theta}$ be the left endpoint of $\Theta_{1}$. Using
(\ref{eq:nn0}), we have that 
\begin{align}
W(M;V,F) & \ge\underline{w}_{q}F(\theta_{s})+\int^{\hat{\theta}}_{\theta_{s}}\underline{W}(\theta,\underline{q}^{{\rm BM}})F(\mathrm{d}\theta)+\int^{\overline{\theta}}_{\hat{\theta}}\underline{W}(\theta,\underline{q}^{{\rm BM}})F(\mathrm{d}\theta)\nonumber \\
 & \ge\underline{w}_{q}F(\theta_{s})+\underline{w}_{q}\left(F(\hat{\theta})-F(\theta_{s})\right)+\int^{\overline{\theta}}_{\hat{\theta}}\underline{W}(\theta,\underline{q}^{\textsc{BM}})F(\mathrm{d}\theta)\label{eq:nn1}\\
 & =\underline{w}_{q}F(\hat{\theta})+\int^{\overline{\theta}}_{\hat{\theta}}\underline{W}(\theta,\underline{q}^{\textsc{BM}})F(\mathrm{d}\theta),\nonumber 
\end{align}
where the second inequality follows from the fact that $\ensuremath{\underline{W}(\cdot,\underline{q}^{\textsc{BM}})\ge\underline{w}_{q}}$
on $\ensuremath{[\theta_{s},\hat{\theta}]}$. Now, let $h:\Theta\rightarrow\mathbb{R}$
be the weakly decreasing function defined by 
\begin{align*}
h(\theta)=\begin{cases}
\underline{w}_{q} & \textrm{if}~\theta\le\hat{\theta}\\
\underline{W}(\theta,\underline{q}^{\textsc{BM}}) & \textrm{otherwise}.
\end{cases}
\end{align*}
From (\ref{eq:nn1}), we then have that 
\begin{align*}
W(M;V,F) & \ge\int^{\hat{\theta}}_{\underline{\theta}}h(\theta)F(\mathrm{d}\theta)+\int^{\overline{\theta}}_{\hat{\theta}}h(\theta)F(\mathrm{d}\theta)=\int^{\overline{\theta}}_{\underline{\theta}}h(\theta)F(\mathrm{d}\theta)\\
 & \ge\int^{\overline{\theta}}_{\underline{\theta}}h(\theta)\underline{F}(\mathrm{d}\theta)=\underline{w}_{q}\underline{F}(\hat{\theta})+\int^{\overline{\theta}}_{\hat{\theta}}\underline{W}(\theta,\underline{q}^{\textsc{BM}})\underline{F}(\mathrm{d}\theta)\ge G^{*}_{s},
\end{align*}
where the second inequality follows from the fact that $\underline{F}\succ_{{\rm FOSD}}F$,
the equality follows by the definition of $h$, and the last inequality
follows from Condition (e).

\noindent Hence, $G(M)\geq G^{*}_{s}$. Part 1 establishes that, for
any $M\in\mathcal{M}^{SL}$, $G(M)=G^{*}_{s}$. We conclude that any
mechanism satisfying Conditions (a)-(e) is in the shortlist. This
completes the proof of the proposition.\hfill{}Q.E.D.

The next proposition provides a characterization of properties of
robustly-optimal mechanisms along the lines of Propositions \ref{prop:BM-floor}
and \ref{prop:robust-quantity-mechanism-general} in the main text.
To do so, we first generalize the definitions of $\text{q}_{\ell}$,
$q^{\star}$, $\theta^{\star}$ and $\theta^{m}$ as follows. Let
$\text{q}^{s}_{\ell}\equiv\underline{D}(\theta_{s})$ denote the efficient
output when the inverse demand is $\underline{D}$ and the cost is
$\theta_{s}$. Then let $q^{\star}_{s}$ be the quantity schedule
defined by 
\begin{align}
q^{\star}_{s}(\theta)\equiv\begin{cases}
\max\{\qbm(\theta),\text{q}^{s}_{\ell}\} & \textrm{\ensuremath{\theta<\theta_{s}}}\\
\underline{q}^{\textsc{BM}}(\theta) & \textrm{\ensuremath{\theta\geq\theta_{s},}}
\end{cases}\label{eq:BM-floor-floor-1}
\end{align}
where $\qbm$ continues to denote the optimal quantity schedule of
Baron and Myerson (1982) when the model is $(V^{\star},F^{\star})$,
with $F^{\star}$ regular, whereas $\underline{q}^{\textsc{BM}}$
is the optimal quantity schedule of Baron and Myerson (1982) when
the model is $(\underline{V},\underline{F})$. The following mechanism
is a generalization of Baron-Myerson-with-quantity-floor in the main
text.

\begin{defn}\label{def-BM-bridge} The \textbf{Baron-Myerson-with-quantity-bridge}
is the mechanism $M^{\star}_{s}=(q^{\star}_{s},u^{\star}_{s})$ where
$q^{\star}_{s}$ is given by (\ref{eq:BM-floor-floor-1}) and $u^{\star}_{s}$
is given by $u^{\star}_{s}(\theta)=\int^{\overline{\theta}}_{\theta}q^{\star}_{s}(y)dy$
for all $\theta$. \end{defn}

Finally, let $\theta^{\star}_{s}$ be the threshold cost defined as
follows. If $\qbm(\theta_{s})\le\text{q}^{s}_{\ell}$, by continuity
of $\qbm$ along with the fact that $\qbm(\underline{\theta})>\text{q}^{s}_{\ell}$
(assured by the regularity of $F^{\star}$), $\theta^{\star}_{s}$
is the unique solution to $\qbm(\theta^{\star}_{s})=\text{q}^{s}_{\ell}$.
If, instead, $\qbm(\theta_{s})>\text{q}^{s}_{\ell}$ (i.e., if $\qbm$
never crosses $\text{q}^{s}_{\ell}$ over the interval $[\underline{\theta},\theta_{s}]$),
then $\theta^{\star}_{s}\equiv\theta_{s}$. In either case, $\theta^{\star}_{s}\leq\theta_{s}$.
Similarly, let $\theta^{m}_{s}\equiv\max\{\theta:\theta\in\arg\min_{y\in[\underline{\theta},\theta_{s}]}\underline{W}(y,q^{\star}_{s})\}.$
Thus, $\underline{w}_{q^{\star}_{s}}=\underline{W}(\theta^{m}_{s},q^{\star}_{s})$.
Finally, let 
\begin{align}
G^{**}_{s}\equiv\sup_{\theta\in(\theta_{s},\overline{\theta}]}\frac{1}{\underline{F}(\theta)}\int^{\theta}_{\theta_{s}}\underline{W}(y,\underline{q}^{\textsc{BM}})\underline{F}(\textrm{d}y).\label{eq:gstar}
\end{align}

\begin{proposition} \label{prop:robust-quantity-mechanism-general-cost}
Suppose $\underline{F}$ is regular with respect to $\theta_{s}$. 
\begin{enumerate}
\item The Baron-Myerson-with-quantity-bridge mechanism is robustly optimal
if and only if $\underline{W}(\theta^{m}_{s},q^{\star}_{s})\ge\max\{G^{*}_{s},G^{**}_{s}\}$.
\item If $\underline{W}(\theta^{m}_{s},q^{\star}_{s})<\max\{G^{*}_{s},G^{**}_{s}\}$,
then $\theta^{m}_{s}<\theta^{\star}_{s}$ and every robustly optimal
mechanism $M^{{\rm {OPT}}}=(q^{{\rm {OPT}}},u^{{\rm {OPT}}})$ satisfies
the following properties: (a) $\qopt(\theta)=\underline{q}^{{\rm BM}}(\theta)$
for all $\theta\in(\theta_{s},\overline{\theta})$, (b) $\qopt(\theta)={\rm q}^{s}_{\ell}$
for all $\theta\in(\theta^{\star}_{s},\theta_{s})$, (c) $\qopt(\theta)\le\qbm(\theta)$
for all $\theta\in(\underline{\theta},\theta^{\star}_{s})$, with
the inequality strict over a subset of types $I\subseteq (\theta^{m}_{s},\theta^{\star}_{s})$
of positive Lebesgue measure, and (d) $q^{{\rm OPT}}(\theta)=q^{{\rm BM}}(\theta)$
for all $\theta\in(\underline{\theta},\theta^{m}_{s})$. 
\end{enumerate}
\end{proposition} Before proving the result, we provide a few observations.

\begin{figure}
\centering %
\begin{minipage}[b]{0.5\linewidth}%
\includegraphics[scale=0.47]{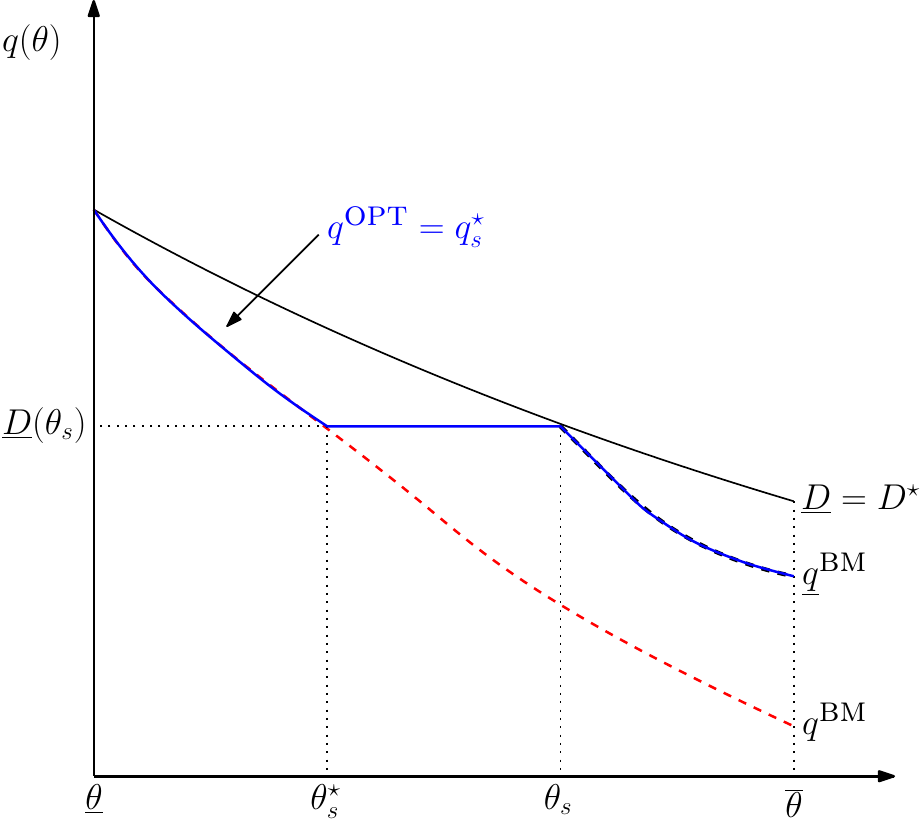} \caption*{\centering A: Baron-Myerson-with-quantity-bridge. \protect\protect\protect\protect\protect\protect\protect\protect\protect\protect\protect\protect\phantom{xxxx}}
\end{minipage}%
\begin{minipage}[b]{0.5\linewidth}%
\includegraphics[scale=0.47]{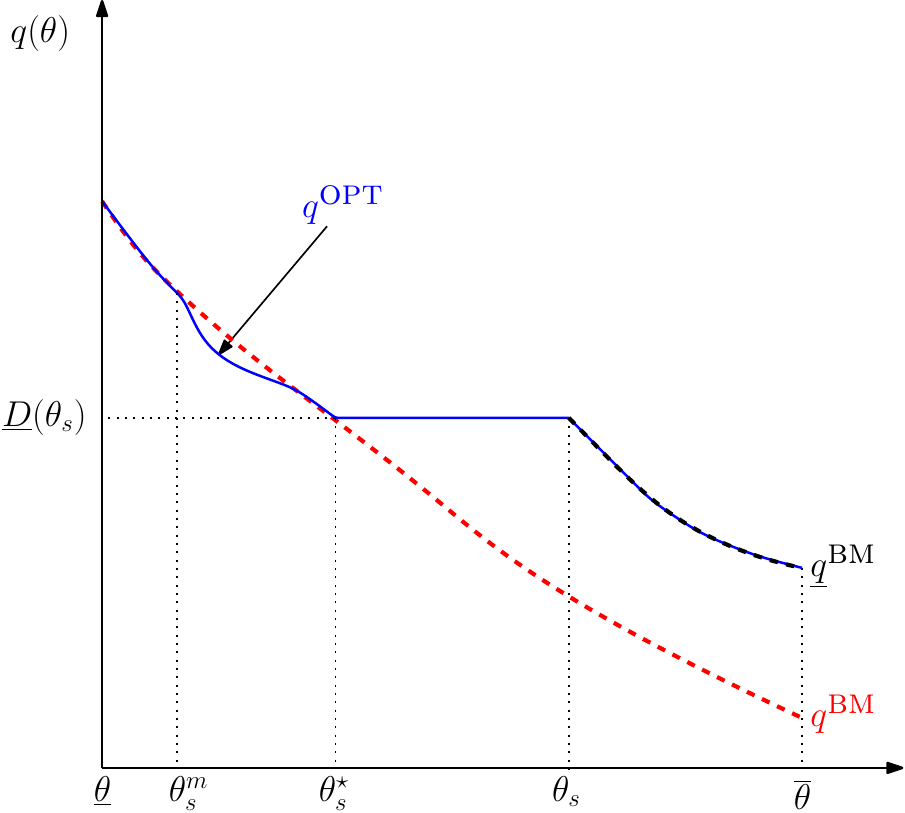} \caption*{\centering B: Robustly optimal mechanism when $V^{\star}\protect\neq\underline{V}$.}
\end{minipage}\caption{Graphical illustration of Proposition \ref{prop:robust-quantity-mechanism-general-cost}.}
\label{fig:general_tech} 
\end{figure}

\begin{remark}\label{rem:r2} {\rm If $V^{\star}=\underline{V}$ (as when
the only uncertainty is over the cost technology), then $\theta^{m}_{s}=\theta_{s}$,
and the condition in Part 1 of Proposition \ref{prop:robust-quantity-mechanism-general-cost}
is satisfied. In this case, the Baron-Myerson-with-quantity-bridge
mechanism, illustrated in Panel A of Figure \ref{fig:general_tech},
is a robustly optimal mechanism. Panel B illustrates the features
of robustly optimal mechanism when, instead, $V^{\star}\neq\underline{V}$.
A comparison of these figures with Figure \ref{fig:BM-floor} and
Figure \ref{fig:alpha0robust} in the main text illustrates that the
key features of the robustly optimal mechanisms continue to extend
to general technological uncertainty.} $\hfill\blacksquare{}$ \end{remark}

\begin{remark}\label{rem:r3} {\rm Proposition \ref{prop:robust-quantity-mechanism-general-cost}
further highlights that the key forces identified in the main text
continue to determine the shape of robustly optimal mechanisms to
the left of $\theta_{s}$. As in the main text, the quantity plateau
in robustly optimal mechanisms is determined by $\underline{D}(\theta_{s})$,
the efficient quantity at $\theta_{s}$ under the lowest possible
demand (recall that, in the main text, $\theta_{s}=\bar{\theta}$,
and $\underline{D}(\theta_{s})={\rm q}_{\ell}$).} $\hfill\blacksquare{}$
\end{remark}

\noindent\textbf{Proof of Proposition \ref{prop:robust-quantity-mechanism-general-cost}.}
The proof is in several parts, each corresponding to a part of the
proposition. Before proceeding with the proof, observe that for any
mechanism $M=(q,u)\in\mathcal{M}^{\textrm{SL}}$, Condition (d) in
Part 2 of Proposition \ref{prop:SL_local_uncertainty} implies that
$\underline{w}_{q}\ge G^{*}_{s}$. Condition (e) in the same proposition
in turn implies that, for all $\theta\ge\theta_{s}$, 
\begin{align*}
\underline{w}_{q}\underline{F}(\theta)+\int^{\overline{\theta}}_{\theta}\underline{W}(y,\underline{q}^{\textsc{BM}})\underline{F}(\textrm{d}y) & \ge G^{*}_{s}\equiv\int^{\overline{\theta}}_{\theta_{s}}\underline{W}(y,\underline{q}^{\textsc{BM}})\underline{F}(\textrm{d}y),
\end{align*}
which is equivalent to 
\begin{align*}
\underline{w}_{q} & \ge\sup_{\theta\in(\theta_{s},\overline{\theta}]}\frac{1}{\underline{F}(\theta)}\int^{\theta}_{\theta_{s}}\underline{W}(y,\underline{q}^{\textsc{BM}})\underline{F}(\textrm{d}y)\equiv G^{**}_{s}.
\end{align*}
Thus, Condition (d) and (e) in Proposition \ref{prop:SL_local_uncertainty}
can be equivalently written as $\underline{w}_{q}\ge\max\{G^{*}_{s},G^{**}_{s}\}$,
or alternatively, $\underline{W}(\theta,q)\ge\max\{G^{*}_{s},G^{**}_{s}\}$
for all $\theta\le\theta_{s}$.

\noindent\textbf{Part 1}. If $M^{\star}_{s}$ is robustly optimal,
then $\text{\ensuremath{M^{\star}_{s}}}\in\mathcal{M}^{{\rm SL}}$.
Observe that $\underline{w}_{q^{\star}_{s}}=\underline{W}(\theta^{m}_{s},q^{\star}_{s})$.
Thus, as argued above, Conditions (d) and (e) in Part 2 of Proposition
\ref{prop:SL_local_uncertainty} imply that $\underline{W}(\theta^{m}_{s},q^{\star}_{s})\ge\max\{G^{*}_{s},G^{**}_{s}\}$.

Next, suppose that $M^{\star}_{s}$ is such that $\underline{W}(\theta^{m}_{s},q^{\star}_{s})\ge\max\{G^{*}_{s},G^{**}_{s}\}$.
We now show that $M^{\star}_{s}$ is robustly optimal. By definition,
$M^{\star}_{s}$ satisfies Conditions (a)-(c) in Part 2 of Proposition
\ref{prop:SL_local_uncertainty}. As argued above, that $\text{\ensuremath{\underline{W}}(\ensuremath{\theta^{m}_{s}},\ensuremath{q^{\star}_{s}})\ensuremath{\ge\max}\{\ensuremath{G^{*}_{s}},\ensuremath{G^{**}_{s}}\}}$
implies that Conditions (d) and (e) in Part 2 of Proposition \ref{prop:SL_local_uncertainty}
are satisfied. This means that $M^{\star}_{s}\in\mathcal{M}^{{\rm SL}}$.
To see that $M^{\star}_{s}$ maximizes the buyer's payoff under the
conjectured model $(V^{\star},F^{\star})$ over $\mathcal{M}^{{\rm SL}}$,
first observe that $q^{\star}_{s}$ is weakly decreasing because $F^{\star}$
is regular. Second note that every mechanism $M=(q,u)$ in the short
list has a weakly decreasing quantity schedule $q$ that agrees with
$\underline{q}^{\textsc{BM}}$ over $(\theta_{s},\bar{\theta})$ (by
virtue of Condition (c) in Proposition \ref{prop:SL_local_uncertainty}).
This means that, in any such mechanism, $q(\theta)\geq{\rm q}^{s}_{\ell}$
for all $\theta\in[\underline{\theta},\theta_{s})$. Because, for
every $\theta\in[\underline{\theta},\theta_{s})$, 
\[
q^{\star}_{s}(\theta)=\arg\max_{{\rm q}\in[{\rm q}^{s}_{\ell},{\rm \bar{q}}]}\left\{ V^{\star}({\rm q})-z^{\star}(\theta){\rm q}\right\} ,
\]
we conclude that, for any $M=(q,u)\in\mathcal{M}^{{\rm SL}}$, 
\[
\int\limits^{\overline{\theta}}_{\underline{\theta}}\Big[V^{\star}(q(\theta))-z^{\star}(\theta)q(\theta)\Big]F^{\star}({\rm d}\theta)\leq\int\limits^{\overline{\theta}}_{\underline{\theta}}\Big[V^{\star}(q^{\star}_{s}(\theta))-z^{\star}(\theta)q^{\star}_{s}(\theta)\Big]F^{\star}({\rm d}\theta),
\]
implying that indeed $M^{\star}_{s}$ maximizes the buyer's payoff
(under $(V^{\star},F^{\star})$) over $\mathcal{M}^{{\rm SL}}$.

\noindent\textbf{Part 2. }We start with the following lemma:

\begin{lemma}\label{lem:thetam} Suppose $\underline{w}_{q^{\star}_{s}}<\max\{G^{*}_{s},G^{**}_{s}\}$.
Then the following are true: (1) $\theta^{m}_{s}<\theta^{\star}_{s}$,
(2) if $\theta^{m}_{s}>\underline{\theta}$, then $q^{\star}_{s}(\theta^{m}_{s})=\underline{D}(\theta^{m}_{s})$.
\end{lemma}

\noindent\textbf{Proof of Lemma \ref{lem:thetam}.} \textbf{Part
1.} We consider two cases. 

\noindent Case1: $q^{\textsc{BM}}(\theta_{s})>\text{q}^{s}_{\ell}=\underline{D}(\theta_{s})$. 

\noindent Then, $\theta^{\star}_{s}=\theta_{s}$. Because $\underline{D}$
and $q^{\textsc{BM}}$ are decreasing and continuous, there exists
a non-empty left neighborhood of $\theta_{s}$ where $q^{\textsc{BM}}(\theta)>\underline{D}(\theta)$.
Part B of Lemma \ref{Lem-Mono} in the main text then implies that
$\underline{W}(\cdot,q)$ is increasing on this interval implying
that $\theta^{m}_{s}<\theta_{s}=\theta^{\star}_{s}$.

Case 2: $q^{\textsc{BM}}(\theta_{s})\le\text{q}^{s}_{\ell}=\underline{D}(\theta_{s})$. 

Then, $\theta^{\star}_{s}\le\theta_{s}$, and for every $\theta\in[\theta^{\star}_{s},\theta_{s}]$,
$q^{\star}_{s}(\theta)=\text{q}^{s}_{\ell}$ and $\underline{W}(\theta,q^{\star}_{s})=\underline{W}(\theta_{s},q^{\star}_{s})\ge\max\{G^{*}_{s},G^{**}_{s}\}>\underline{w}_{q^{\star}_{s}}$.
The first inequality follows from the definitions of $G^{*}_{s}$
and $G^{**}_{s}$ along with Part A of Lemma \ref{Lem-Mono} in the
main text (which implies that $\underline{W}(\cdot,q)$ is weakly
decreasing over $[\theta_{s},\bar{\theta}]$). The second inequality
follows from the assumption of Lemma \ref{lem:thetam}. This implies
that $\underline{W}(\theta^{\star}_{s},q^{\star}_{s})>\underline{w}_{q^{\star}_{s}}$.
Hence, $\theta^{m}_{s}<\theta^{\star}_{s}$.

\noindent\textbf{Part 2.} Because $\theta^{m}_{s}<\theta^{\star}_{s}$,
we have that $q^{\star}_{s}(\theta^{m}_{s})=\qbm(\theta^{m}_{s})$.
This property, together with arguments similar to those in the proof of Lemma \ref{lem:qstardlowint} in the
main text, then implies the result. \hfill{}$\blacksquare$

Equipped with the last lemma, we now establish Parts 2(a) and 2(b)
of the proposition.

\noindent\textbf{Parts 2(a) and 2(b)}. Part 2(a) follows from Part
2(c) of Proposition \ref{prop:SL_local_uncertainty} because any robustly
optimal mechanism belongs to $\mathcal{M}^{{\rm SL}}$. Thus consider
Part 2(b). We consider two cases. First, suppose that $q^{\textsc{BM}}(\theta_{s})\ge\text{q}^{s}_{\ell}$.
By the definition of $\theta^{\star}_{s}$, we then have that $\theta^{\star}_{s}=\theta_{s}$,
and hence the interval $(\theta^{\star}_{s},\theta_{s})$ is empty
and the result applies vacuously. Next, suppose that $q^{\textsc{BM}}(\theta_{s})<\text{q}^{s}_{\ell}$.
Then, $\theta^{\star}_{s}<\theta_{s}$. Now, assume, toward a contradiction,
that there exists $\theta'\in(\theta^{\star}_{s},\theta_{s})$ such
that $\qopt(\theta')>\text{q}^{s}_{\ell}$. Monotonicity of $\qopt$
then implies that $\qopt(\theta)>\text{q}^{s}_{\ell}$ for all $\theta\in[\theta^{\star}_{s},\theta'].$
This means that there exists a set of types of positive Lebesgue measure
such that $\qopt(\theta)>\text{q}^{s}_{\ell}$. Then, consider the
IC and IR mechanism $\widetilde{M}=(\tilde{q},\tilde{u})$ where the
quantity schedule is given by 
\begin{align*}
\tilde{q}(\theta)=\begin{cases}
q^{{\rm {OPT}}}(\theta) & \textrm{if}~~\theta<\theta^{\star}_{s}\\
\text{q}^{s}_{\ell} & \textrm{if}~~\theta\in[\theta^{\star}_{s},\theta_{s}]\\
\underline{q}^{\textsc{BM}}(\theta)=q^{{\rm {OPT}}}(\theta) & \textrm{if}~~\theta\ge\theta_{s},
\end{cases}
\end{align*}
and where the rents $\tilde{u}$ are given by $\tilde{u}(\theta)=\int\limits^{\overline{\theta}}_{\theta}\tilde{q}(y)dy$
for all $\theta$. The buyer's payoff from $\widetilde{M}$ (under
the conjectured model $(V^{\star},F^{\star})$) is equal to $\int\limits^{\overline{\theta}}_{\underline{\theta}}\Big[V^{\star}(\tilde{q}(\theta))-z^{\star}(\theta)\tilde{q}(\theta)\Big]F^{\star}(\mathrm{d}\theta)$
which is strictly higher than under $M^{{\rm {OPT}}}$. This follows
from the fact that, for any $\theta\in[\theta^{\star}_{s},\theta_{s}]$,
$\text{q}^{s}_{\ell}$ maximizes $V^{\star}(\text{q})-z^{\star}(\theta)\text{q}$
over $\text{q}\ge\text{q}^{s}_{\ell}$, along with the fact that $F^{\star}$
is absolutely continuous. Thus, to produce a contradiction to the
robust optimality of $M^{\textsc{OPT}}$, it suffices to show that
$\widetilde{M}\in\mathcal{M}^{\textsc{SL}}$, which we show next.

By definition, $\widetilde{M}$ satisfies Conditions (a)-(c) in Part
2 of Proposition \ref{prop:SL_local_uncertainty}. As for Conditions
(d) and (e), recall that these conditions are equivalent to $\underline{w}_{\tilde{q}}\ge\max\{G^{*}_{s},G^{**}_{s}\}$,
as established above. That $\underline{w}_{\tilde{q}}\ge G^{*}_{s}$
follows from arguments analogous to those in the proof of Lemma \ref{lem:weakqoptbound}
in the main text, along with the fact that $\tilde{q}(\theta)\le\qopt(\theta)$
for all $\theta$. To establish that $\underline{w}_{\tilde{q}}\ge G^{**}_{s}$,
notice that, for all $\theta\in[\theta^{\star}_{s},\theta_{s}]$,
$\underline{W}(\theta,\tilde{q})=\underline{W}(\theta_{s},\underline{q}^{\textsc{BM}})\ge G^{**}_{s}$.
Thus, it suffices to focus on $\theta<\theta^{\star}_{s}$. Observe
that $\underset{\theta<\theta^{\star}_{s}}{\inf}~\underline{W}(\theta,\tilde{q})\ge\underset{\theta<\theta^{\star}_{s}}{\inf}~\underline{W}(\theta,q^{\textsc{OPT}})\ge G^{**}_{s}$,
where the first inequality follows from the fact that $\tilde{q}(\theta)=q^{\textsc{OPT}}(\theta)$
for $\theta<\theta^{\star}_{s}$, along with the fact that $\tilde{q}(\theta)\le\qopt(\theta)$
for all $\theta$, which implies that $\tilde{u}(\theta)\leq u^{\textsc{OPT}}(\theta)$
for all $\theta$. The second inequality follows from the fact that
$M^{\textsc{OPT}}\in\mathcal{M}^{\textsc{SL}}$.

\noindent\textbf{Part 2(c)}. From Part 2(b), $\qopt(\theta)=\text{q}^{s}_{\ell}$
for all $\theta\in(\theta^{\star}_{s},\theta_{s})$. Now suppose there
is a set $S\subseteq (\underline{\theta},\theta^{\star}_{s})$ of positive
Lebesgue measure such that $q^{{\rm OPT}}(\theta)>q^{\star}_{s}(\theta)=\qbm(\theta)$
for all $\theta\in S$. Consider the IC and IR mechanism $\widetilde{M}=(\tilde{q},\tilde{u})$
where the quantity schedule is given by 
\begin{align*}
\tilde{q}(\theta)=\min\{q^{\star}_{s}(\theta),q^{{\rm {OPT}}}(\theta)\}~\qquad~\forall~\theta\in\Theta,
\end{align*}
and where $\tilde{u}(\theta)=\int^{\overline{\theta}}_{\theta}\tilde{q}(y){\rm d}y$
for all $\theta$. Parts 2(a) and 2(b) in the proposition then imply
that $\tilde{q}(\theta)=q^{{\rm {OPT}}}(\theta)=q^{\star}_{s}(\theta)$
for all $\theta\ge\theta^{\star}_{s}$.

The buyer's payoff under $\widetilde{M}$ is strictly higher than
under $M^{{\rm {OPT}}}$ (this follows from $q^{\star}_{s}$ maximizing
virtual surplus). Clearly, $\widetilde{M}$ satisfies Conditions (a)-(c)
of Part 2 of Proposition \ref{prop:SL_local_uncertainty}. The next
two claims establish that $\widetilde{M}$ also satisfies Conditions
(d) and (e), that is, $\underline{w}_{\tilde{q}}\ge\max\{G^{*}_{s},G^{**}_{s}\}$
or, equivalently, $\underline{W}(\theta,\tilde{q})\ge\max\{G^{*}_{s},G^{**}_{s}\}$
for all $\theta\le\theta_{s}$. First, observe that, for every $\theta\in[\theta^{\star}_{s},\theta_{s}]$,
$\tilde{q}(\theta)=q^{\star}_{s}(\theta)=\text{q}^{s}_{\ell}$ and
$\underline{W}(\theta,\tilde{q})=\underline{W}(\theta_{s},\tilde{q})\geq\max\{G^{*}_{s},G^{**}_{s}\}$.
Below we establish that $\underline{W}(\theta,\tilde{q})\geq\max\{G^{*}_{s},G^{**}_{s}\}$
also for $\theta<\theta^{\star}_{s}$. The result is a consequence
of the following two claims.

\begin{claim} \label{cl:left0-g} Suppose $\theta<\theta^{\star}_{s}$
is such that either $\tilde{q}(\theta)=q^{{\rm OPT}}(\theta)$ or
$\underline{D}(\theta)\le\tilde{q}(\theta)=q^{\star}_{s}(\theta)<q^{{\rm OPT}}(\theta)$.
Then $\underline{W}(\theta,\tilde{q})\ge\max\{G^{*}_{s},G^{**}_{s}\}$.
\end{claim}

\noindent\textbf{Proof of Claim \ref{cl:left0-g}}. Pick $\theta<\theta^{\star}_{s}$.
It suffices to show that $\underline{W}(\theta,\tilde{q})\ge\underline{W}(\theta,\qopt)$;
because $\underline{W}(\theta,\qopt)\ge\max\{G^{*}_{s},G^{**}_{s}\}$,
the claim then follows. Note that $q^{\star}_{s}(\theta)=q^{{\rm BM}}(\theta)$.
For any $\theta$ such that $\tilde{q}(\theta)=q^{{\rm OPT}}(\theta)$,
the fact that $\tilde{q}(y)\leq q^{{\rm {OPT}}}(y)$ for all $y\geq\theta$,
implies that $\underline{W}(\theta,\tilde{q})\ge\underline{W}(\theta,\qopt)$.
Thus, consider a $\theta$ for which $\underline{D}(\theta)\le\tilde{q}(\theta)=q^{\star}_{s}(\theta)=q^{{\rm BM}}(\theta)<q^{{\rm OPT}}(\theta)$.
The quasi-concavity of the function $\underline{V}(\text{q})-\theta\text{q}$
in $\text{q}$ implies that $\underline{V}(q^{\star}_{s}(\theta))-\theta q^{\star}_{s}(\theta)>\underline{V}(q^{{\rm OPT}}(\theta))-\theta q^{{\rm OPT}}(\theta)$.
Because $\tilde{q}(y)\leq q^{{\rm {OPT}}}(y)$ for all $y\geq\theta$,
we have that $\underline{W}(\theta,\tilde{q})\ge\underline{W}(\theta,q^{{\rm OPT}})$.\hfill{}$\blacksquare$

\noindent\begin{claim} \label{cl:left1} If $\theta<\theta^{\star}_{s}$
and $q^{\star}_{s}(\theta)<\min\{\underline{D}(\theta),q^{{\rm {OPT}}}(\theta)\}$,
then $\underline{W}(\theta,\tilde{q})\ge\max\{G^{*}_{s},G^{**}_{s}\}$.\end{claim}

\noindent\textbf{Proof of Claim \ref{cl:left1}.} The proof considers
two cases to establish the existence of $\theta'>\theta$ such that
$W(\cdot,\tilde{q})$ is weakly decreasing on $[\theta,\theta']$
with $W(\theta',\tilde{q})\ge\max\{G^{*}_{s},G^{**}_{s}\}$.

\noindent\textbf{Case 1.} Suppose $\qbm(\theta_{s})\le\text{q}^{s}_{\ell}=\underline{D}(\theta_{s})$.
Because $q^{\star}_{s}$ and $\underline{D}$ are both continuous,
there exists $\theta<\theta'\le\theta_{s}$ such that $q^{\star}_{s}(y)\le\underline{D}(y)$
for all $y\in[\theta,\theta']$, with $q^{\star}_{s}(\theta')=\underline{D}(\theta')$.
Thus, $\tilde{q}(\theta')=\min\{\underline{D}(\theta'),q^{{\rm {OPT}}}(\theta')\}$.
Furthermore, for all $y\in[\theta,\theta']$, $\tilde{q}(y)=\min\{q^{{\rm {OPT}}}(y),q^{\star}_{s}(y)\}\leq\underline{D}(y)$.
Part A of Lemma \ref{Lem-Mono} in the main text implies that $\underline{W}(\cdot,\tilde{q})$
is weakly decreasing over $[\theta,\theta']$ whereas Claim \ref{cl:left0-g}
implies that $\underline{W}(\theta',\tilde{q})\ge\max\{G^{*}_{s},G^{**}_{s}\}$.
Hence, $\underline{W}(\theta,\tilde{q})\ge\max\{G^{*}_{s},G^{**}_{s}\}$.

\noindent\textbf{Case 2.} Now suppose $\qbm(\theta_{s})>\text{q}^{s}_{\ell}=\underline{D}(\theta_{s})$.
Then, because $q^{\star}_{s}(\theta)<\underline{D}(\theta)$, and
$\underline{D}$ and $q^{\star}_{s}$ are continuous (the latter due
to the regularity of $F^{\star}$), there exists $\theta<\hat{\theta}<\overline{\theta}$
such that $q^{\star}_{s}(\hat{\theta})=\underline{D}(\hat{\theta})$
and $q^{\star}_{s}(y)>\underline{D}(y)$ for all $y\in(\hat{\theta},\theta_{s})$.
Again, just like we argued in Case 1, there exists $\theta<\theta'\le\hat{\theta}$
such that $q^{\star}_{s}(y)\le\underline{D}(y)$ for all $y\in[\theta,\theta']$
with $q^{\star}_{s}(\theta')=\underline{D}(\theta')$. The same arguments
as in Case 1 then imply that $\underline{W}(\theta,\tilde{q})\ge\max\{G^{*}_{s},G^{**}_{s}\}$.\hfill{}$\blacksquare$

\noindent The above two claims establish that $\tilde{q}$ is such
that $\underline{w}_{\tilde{q}}\ge\max\{G^{*}_{s},G^{**}_{s}\}$.
From Proposition \ref{prop:SL_local_uncertainty}, we thus conclude
that $\widetilde{M}=(\tilde{q},\tilde{u})\in\mathcal{M}^{\mathrm{SL}}$.

The arguments above imply that $\qopt(\theta)\le\qbm(\theta)$ for
almost all $\theta\in(\underline{\theta},\theta^{\star}_{s}).$ Because
$\qopt$ is weakly decreasing and $\qbm$ is continuous, the inequality
must in fact hold for all $\theta\in(\underline{\theta},\theta^{\star}_{s})$. 

Next, we show that there exists a set of types $I\subseteq (\underline{\theta},\theta^{\star}_{s})$
of positive Lebesgue measure such that $\qopt(\theta)<\qbm(\theta)$
for all $\theta\in I$. To see this, assume, towards a contradiction,
that $\qopt(\theta)=\qbm(\theta)$ for almost all $\theta\in (\underline{\theta},\theta^{\star}_{s})$.
Because $q^{\star}_{s}$ and $\qopt$ are weakly decreasing and $q^{\star}_{s}$
is continuous, this means that $\qopt(\theta)=\qbm(\theta)$ for all
$\theta\in(\tlb,\theta^{\star}_{s})$. This however implies that $M^{\star}_{s}$
is robustly optimal, a contradiction. That $I$ is in fact a subset
of $(\theta^{m}_{s},\theta^{\star}_{s})$ follows from what just established
along with part (d) in the proposition, which we prove next. 

\noindent\textbf{Part 2(d)}. Assume, toward a contradiction, that
there exists a $\theta\in(\underline{\theta},\theta^{m}_{s})$ such
that $\qopt(\theta)\neq\qbm(\theta)$. Because $\qopt(y)\leq\qbm(y)$
for all $y\in(\underline{\theta},\theta^{m}_{s})$ (by virtue of Part
2(c) of Proposition \ref{prop:robust-quantity-mechanism-general-cost}), and because $\qbm$ and $\qopt$ are weakly
decreasing and $\qbm$ is continuous, this means that there exists
a set of types $I\subseteq(\underline{\theta},\theta^{m}_{s})$ of
positive Lebesgue measure such that $\qopt(y)<\qbm(y)$ for all $y\in I$.
 Then, let $\widetilde{M}=(\tilde{q},\tilde{u})$ be the mechanism
where the quantity schedule is given by 
\[
\tilde{q}(\theta)=\begin{cases}
q^{{\rm {BM}}}(\theta) & \textrm{if}~~\theta\in[\underline{\theta},\theta^{m}_{s}]\\
q^{{\rm {OPT}}}(\theta) & \textrm{otherwise}
\end{cases}
\]
and where $\tilde{u}(\theta)=\int^{\overline{\theta}}_{\theta}\tilde{q}(y)dy$
for all $\theta$. Clearly, $\widetilde{M}$ is IC and IR. Below,
we show that $\widetilde{M}$ yields a higher payoff to the buyer
than $M^{{\rm {OPT}}}$ and $\widetilde{M}\in\mathcal{M}^{\mathrm{SL}}$,
contradicting the optimality of $M^{{\rm {OPT}}}$.

Because, for any $\theta$, $q^{{\rm {BM}}}(\theta)$ is the unique
maximizer of $V^{\star}(\text{q})-z^{\star}(\theta)\text{q}$, 
\[
\int\limits^{\overline{\theta}}_{\underline{\theta}}\Big[V^{\star}(\tilde{q}(\theta))-z^{\star}(\theta)\tilde{q}(\theta)\Big]F^{\star}(\mathrm{d}\theta)>\int\limits^{\overline{\theta}}_{\underline{\theta}}\Big[V^{\star}(q^{{\rm {OPT}}}(\theta))-z^{\star}(\theta)q^{{\rm {OPT}}}(\theta)\Big]F^{\star}(\mathrm{d}\theta).
\]

We now show that $\tilde{q}$ is such that $\underline{w}_{\tilde{q}}\ge\max\{G^{*}_{s},G^{**}_{s}\}$.
To do so, it suffices to show that $\underline{W}(\theta,\tilde{q})\ge\max\{G^{*}_{s},G^{**}_{s}\}$
for $\theta\le\theta^{m}_{s}$ (in fact, for $\theta>\theta^{m}_{s}$,
$\underline{W}(\theta,\tilde{q})=\underline{W}(\theta,q^{{\rm {OPT}}})\ge\max\{G^{*}_{s},G^{**}_{s}\}$,
where the inequality follows from the fact that $M^{\textsc{OPT}}\in\mathcal{M}^{\textsc{SL}}$).
Thus consider $\theta\in[\underline{\theta},\theta^{m}_{s}]$. For
any such $\theta$, $\tilde{q}(\theta)=\qbm(\theta)=q^{\star}_{s}(\theta)$.
The latter equality follows from Lemma \ref{lem:thetam}, which establishes
that $\theta^{m}_{s}<\theta^{\star}_{s}$. Moreover, 
\begin{align*}
 & \underline{W}(\theta,\tilde{q})-\underline{W}(\theta,q^{\star}_{s})=~~~\int\limits^{\overline{\theta}}_{\theta^{m}_{s}}q^{\star}_{s}(y)dy-\int\limits^{\overline{\theta}}_{\theta^{m}_{s}}q^{{\rm {OPT}}}(y)dy\\
 & \ge_{(i)}\Big[\underline{V}(q^{{\rm OPT}}(\theta^{m}_{s}))-\theta^{m}_{s}q^{{\rm OPT}}(\theta^{m}_{s})\Big]-\Big[\underline{V}(q^{{\rm BM}}(\theta^{m}_{s}))-\theta^{m}_{s}q^{{\rm BM}}(\theta^{m}_{s})\Big]\\
 & ~~~~~+\int\limits^{\overline{\theta}}_{\theta^{m}_{s}}q^{\star}_{s}(y)dy-\int\limits^{\overline{\theta}}_{\theta^{m}_{s}}q^{{\rm {OPT}}}(y)dy=_{(ii)}\underline{W}(\theta^{m}_{s},q^{{\rm {OPT}}})-\underline{W}(\theta^{m}_{s},q^{\star}_{s})\\
 & \ge_{(iii)}\max\{G^{*}_{s},G^{**}_{s}\}-\underline{W}(\theta^{m}_{s},q^{\star}_{s})\ge_{(iv)}\max\{G^{*}_{s},G^{**}_{s}\}-\underline{W}(\theta,q^{\star}_{s}).
\end{align*}
Inequality $(i)$ follows from the fact that $\underline{D}(\theta^{m}_{s})$
maximizes $\underline{V}(\text{q})-\theta^{m}_{s}\text{q}$ over all
$[0,\bar{\text{q}}]$ and $\qbm(\theta^{m}_{s})=\underline{D}(\theta^{m}_{s})$
(Lemma \ref{lem:thetam}). Equality $(ii)$ follows from the fact
that $q^{\star}_{s}(\theta)=\qbm(\theta)$. Inequality $(iii)$ follows
from the fact that $M^{{\rm {OPT}}}\in\mathcal{M}^{\mathrm{SL}}$
which implies that $q^{{\rm {OPT}}}(y)$ satisfies the robustness
constraints $\underline{W}(\theta^{m}_{s},q^{{\rm {OPT}}})\ge\max\{G^{*}_{s},G^{**}_{s}\}$.
Inequality $(iv)$ follows from the definition of $\theta^{m}_{s}$.
Hence, $\underline{W}(\theta,\tilde{q})\geq\max\{G^{*}_{s},G^{**}_{s}\}$.
We conclude that $\widetilde{M}\in\mathcal{M}^{\mathrm{SL}}$ and
yields a higher payoff to the buyer than $M^{{\rm {OPT}}}$, contradicting
the optimality of $M^{{\rm {OPT}}}$. \hfill{}Q.E.D.}

\newpage{} 
\appendix
{
\setcounter{section}{4}
\section{Supplement II}\label{Sec:additional_supplement}
\global\long\def\thesublemma{\thelemma\Alph{sublemma}}%
\global\long\def\thelemma{E.\arabic{lemma}}%
\global\long\def\thesection{E.\arabic{section}}%
\global\long\def\thesubsection{E.\arabic{subsection}}%
\global\long\def\theproposition{E.\arabic{proposition}}%
\global\long\def\thedefn{E.\arabic{defn}}%
\global\long\def\thetheorem{E.\arabic{theorem}}%
\global\long\def\thecorollary{E.\arabic{corollary}}%
\global\long\def\theremark{E.\arabic{remark}}%
\global\long\def\theclaim{E.\arabic{claim}}%
\global\long\def\thefigure{E.\arabic{figure}}%
\global\long\def\theequation{E.\arabic{equation}}%
\global\long\def\theobs{E.\arabic{obs}}%

The notation in this supplement is the same as in the original article.
All sections, definitions, displayed conditions, and results specific
to this document have the prefix ``E'' to avoid confusion with the
corresponding parts in the main text. Section \ref{sec:existence}
establishes existence of robustly optimal mechanisms. Section \ref{sec:undom}
studies the connection between robustly optimal and undominated mechanisms.
Section \ref{sec:compstat} examines comparative statics of robustly
optimal mechanisms with respect to the lowest distribution $\underline{F}$
in the feasible set $\mathcal{F}$ (in the main text such a distribution
is a Dirac measure at $\bar{\theta}$).

\subsection{Existence of robustly optimal mechanisms}

\label{sec:existence}

\begin{proposition}\label{Prop:exist} A robustly optimal mechanism
exists.\footnote{We thank Issan Patri for providing helpful comments and suggestions
on this part.} \end{proposition}

\noindent\textbf{Proof of Proposition \ref{Prop:exist}}. Recall
that a mechanism $M=(q,u)$ is robustly optimal if and only if $u$
satisfies $u(\theta)=\int^{\bar{\theta}}_{\theta}q(s){\rm d}s$ for
all $\theta\in\Theta$, with $q$ solving the following program (where
we use the fact that $F^{\star}$ is absolutely continuous over $\mathbb{R}$,
with density $f^{\star}(\theta)>0$ if and only if $\theta\in\Theta$):
\[
\begin{array}{c}
\max_{q:[\underline{\theta},\overline{\theta}]\rightarrow[0,\bar{{\rm q}}]}\int\limits^{\overline{\theta}}_{\underline{\theta}}\Big[V^{\star}(q(\theta))-z^{\star}(\theta)q(\theta)\Big]f^{\star}(\theta){\rm d}\theta\\
\textrm{s.t. }q~\textrm{weakly decreasing and \ensuremath{\underline{V}(q(\theta))-\theta q(\theta)-\int\limits^{\overline{\theta}}_{\theta}q(y){\rm d}y\ge G^{*}\quad\forall\theta\in\Theta.} }
\end{array}
\]
Hereafter, we refer to the set of schedules $q$ satisfying the restrictions
in the above problem as the ``feasible set''.

Because each $q$ in the feasible set is uniformly bounded, i.e.,
$0\le q(\theta)\le{\rm \bar{q}}$ for all $\theta\in[\underline{\theta},\overline{\theta}]$,
by Helly's selection theorem, the set of weakly decreasing schedules
$q$ is sequentially compact under the point-wise convergence topology.
Furthermore, because $\underline{V}$ is continuous and $q$ is uniformly
bounded, by the dominated convergence theorem, for any $\theta\in\Theta$,
the function $H(q;\theta)$ defined by $H(q;\theta)\equiv\underline{V}(q(\ensuremath{\theta}))-\theta q(\theta)-\int^{\overline{\theta}}_{\theta}q(y){\rm d}y$
is sequentially continuous in $q$. Hence, the feasible set is sequentially
compact under the point-wise convergence topology. That this set is
non-empty is immediate given that the function $q$ satisfying $q(\theta)=\text{q}_{\ell}$
for all $\theta\in\Theta$ meets all the requirements.

Next, observe that the objective function in the above program is
sequentially continuous in $q$. To see this, let $\phi$ be the function
defined, for all $(\theta,{\rm q})\in\Theta\times[0,{\rm \bar{q}}]$
by 
\begin{align*}
\phi(\theta,{\rm q})=\Big[V^{\star}({\rm q})-z^{\star}(\theta){\rm q}\Big]f^{\star}(\theta).
\end{align*}
Note that, given any function $q$ in the feasible set, the value
of the objective function is 
\begin{align*}
\int\limits^{\overline{\theta}}_{\underline{\theta}}\phi(\theta,q(\theta)){\rm d}\theta.
\end{align*}
Clearly, $\phi(\theta,{\rm q})$ is continuous over $\Theta\times[0,{\rm \bar{q}}]$.
By Weierstrass’ theorem, for all $\theta\in\Theta$, $\underline{\phi}(\theta)\equiv\min_{{\rm q}\in[0,{\rm \bar{q}}]}\phi(\theta,{\rm q})$
and $\bar{\phi}(\theta)\equiv\max_{{\rm q}\in[0,{\rm \bar{q}}]}\phi(\theta,{\rm q})$
exist. Furthermore, by the theorem of the maximum, $g(\theta)\equiv\max\{|\underline{\phi}(\theta)|,|\bar{\phi}(\theta)|\}$
is continuous. Because $\Theta$ is compact, $g$ is bounded and hence
integrable over $\Theta.$ Now take a sequence $(q_{n})$ of feasible
schedules converging to $q$ under the point-wise convergence topology.
Then, for every $\theta\in[\underline{\theta},\overline{\theta}]$,
we have that $\lim_{n\rightarrow\infty}\phi(\theta,q_{n}(\theta))=\phi(\theta,q(\theta))$
by continuity of $\phi$ in the second argument. Furthermore, for
each $q_{n}$ in the sequence, we have that $|\phi(\theta,q_{n}(\theta))|\le g(\theta)$
for all $\theta\in\Theta$. Then, by the dominated convergence theorem,
\begin{align*}
\lim_{n\rightarrow\infty}\int^{\overline{\theta}}_{\underline{\theta}}\phi(\theta,q_{n}(\theta)){\rm d}\theta=\int^{\overline{\theta}}_{\underline{\theta}}\phi(\theta,q(\theta)){\rm d}\theta.
\end{align*}
This establishes the sequential continuity of the objective function
under the point-wise convergence topology. Because the range of the
objective function is a subset of $\mathbb{R}$, from the extreme
value theorem, we conclude that the above optimization program has
a solution, i.e., a robustly optimal mechanism exists. This completes
the proof of the proposition. \hfill{}Q.E.D.

\subsection{Undomination and robustness}

\label{sec:undom}

We formally define what it means for a mechanism to be undominated,
and then establish that robustly optimal mechanisms are undominated.

\begin{defn} For any pair of mechanisms $M=(q,u)$ and $\widehat{M}=(\hat{q},\hat{u})$,
$M$ \textbf{dominates} $\widehat{M}$ if, for every $(V,F)\in\mathcal{A}$,\footnote{As in the main text, we assume that $\mathcal{F}$ is the collection
of all cdfs supported on $\Theta$ and that $(V^{\star},F^{\star})\in\mathcal{A}\equiv\mathcal{V}\times\mathcal{F}.$} 
\begin{align*}
\int^{\overline{\theta}}_{\underline{\theta}}\Big[V(q(\theta))-\theta q(\theta)-u(\theta)\Big]F(\mathrm{d}\theta) & \ge\int^{\overline{\theta}}_{\underline{\theta}}\Big[V(\hat{q}(\theta))-\theta\hat{q}(\theta)-\hat{u}(\theta)\Big]F(\mathrm{d}\theta),
\end{align*}
with the inequality strict for some $(V,F)$. A mechanism $\widehat{M}\in\mathcal{M}$
is \textbf{undominated} if there does not exist a mechanism $M\in\mathcal{M}$
that \textbf{dominates} it. \end{defn}

The following lemma points to an internal consistency property of
the set of robustly optimal mechanisms: each robustly optimal mechanism
is either undominated, or it is dominated by another robustly optimal
mechanism.

\begin{lemma} \label{lem:dom1} Suppose $M^{\mathrm{OPT}}=(\qopt,\uopt)$
is a robustly optimal mechanism and $M=(q,u)\in\mathcal{M}$ dominates
$M^{\mathrm{OPT}}$. Then $M$ is a robustly optimal mechanism. \end{lemma}

\noindent\textbf{Proof of Lemma \ref{lem:dom1}}. Because $M^{\mathrm{OPT}}$
is a robustly optimal mechanism, for all $\theta\in\Theta$, $\underline{V}(\qopt(\theta))-\theta\qopt(\theta)-\uopt(\theta)\geq G^{*}$.
Now pick any $\theta\in\Theta$. Because $M$ dominates $\text{\ensuremath{M^{\mathrm{OPT}}}}$,
under $V=\underline{V}$ and $F=\delta_{\theta}$ (where $\delta_{\theta}$
is the Dirac distribution that puts unit point mass at $\theta$),
we have that $\underline{V}(q(\theta))-\theta q(\theta)-u(\theta)\geq\underline{V}(\qopt(\theta))-\theta\qopt(\theta)-\uopt(\theta)$.
Combining the two inequalities, we have that $\underline{V}(q(\theta))-\theta q(\theta)-u(\theta)\geq G^{*}$.
Because the last inequality holds for all $\theta$, $M\in\mathcal{M}^{{\rm SL}}$.

Next, pick any $\theta\in\Theta$. Because $M$ dominates $M^{\mathrm{OPT}}$,
under $V=V^{\star}$ and $F=\delta_{\theta}$, we have that $V^{\star}(q(\theta))-\theta q(\theta)-u(\theta)\geq V^{\star}(\qopt(\theta))-\theta\qopt(\theta)-\uopt(\theta)$,
which implies that 
\begin{align*}
\int\limits^{\overline{\theta}}_{\underline{\theta}}\Big[V^{\star}(q(\theta))-\theta q(\theta)-u(\theta)\Big]F^{\star}(\mathrm{d}\theta)\geq\int\limits^{\overline{\theta}}_{\underline{\theta}}\Big[V^{\star}(\qopt(\theta))-\theta\qopt(\theta)-\uopt(\theta)\Big]F^{\star}(\mathrm{d}\theta).
\end{align*}
Because $M\in\mathcal{M}^{{\rm SL}}$ and $M^{\mathrm{OPT}}$ is robustly
optimal, the above inequality implies that $M$ is also robustly optimal.
In turn, this implies that the above inequality is an equality. This
completes the proof of the lemma. \hfill{}Q.E.D.

\begin{proposition} \label{Prop:dom} If the Baron-Myerson-with-quantity-floor
mechanism $M^{\star}=(q^{\star},u^{\star})$ is robustly optimal,
it is undominated. \end{proposition}

\noindent\textbf{Proof of Proposition \ref{Prop:dom}}. Suppose $M^{\star}$
is robustly optimal and $M=(q,u)$ dominates it. By Lemma \ref{lem:dom1},
$M$ is also robustly optimal. By Corollary \ref{cor:qstaropt} in
the main text, $q(\theta)=q^{\star}(\theta)$ for all $\theta>\underline{\theta}$.
This implies that $u(\theta)=u^{\star}(\theta)$ for all $\theta$.

Consider the pair $(V^{\star},\delta_{\underline{\theta}})$, where
$\delta_{\underline{\theta}}$ is the Dirac distribution that puts
unit mass at $\underline{\theta}$. Then $V^{\star}(q^{\star}(\underline{\theta}))-\underline{\theta}q^{\star}(\underline{\theta})-u^{\star}(\underline{\theta})>V^{\star}(q(\underline{\theta}))-\underline{\theta}q(\underline{\theta})-u(\underline{\theta})$.
The inequality holds because $u(\underline{\theta})=u^{\star}(\underline{\theta})$
and because $q^{\star}(\underline{\theta})\equiv\qbm(\underline{\theta})$
uniquely maximizes surplus $V^{\star}({\rm q})-\underline{\theta}{\rm q}$.
This inequality, however, contradicts the fact that $M=(q,u)$ dominates
$(q^{\star},u^{\star})$. \hfill{}Q.E.D.

Hence, when the conditions in Proposition \ref{prop:BM-floor} in
the main text are satisfied, the Baron-Myerson-with-quantity-floor
mechanism is not only robustly optimal but also undominated. This
result generalizes to other robustly optimal mechanisms, albeit under
a mild condition, which is satisfied by $M^{\star}=(q^{\star},u^{\star})$.
Let 
\[
Q^{\#}(\underline{\theta})\equiv\cup_{V\in\mathcal{V}}\left\{ \arg\max_{{\rm q}\in[0,\ensuremath{\bar{{\rm q}}}]}\left(V({\rm q})-\underline{\theta}{\rm q}\right)\right\} 
\]
denote the set of output levels that, when $\theta=\underline{\theta}$,
deliver the highest total surplus for some $V\in\mathcal{V}$.

\begin{proposition}\label{prop-undomination-gen} Let $M^{\mathrm{OPT}}=(\qopt,\uopt)$
be a robustly optimal mechanism in which $\qopt$ is left-continuous
and, if $\theta^{m}>\tlb$, $\qopt(\underline{\theta})\in Q^{\#}(\underline{\theta})$.
Then $M^{\mathrm{OPT}}$ is undominated. \end{proposition}

\noindent\textbf{Proof of Proposition \ref{prop-undomination-gen}}.
Towards a contradiction, assume that $M^{\mathrm{OPT}}$ is dominated
by another mechanism $\widehat{M}=(\hat{q},\hat{u})$. By Proposition
\ref{Prop:dom}, we then have that $M^{\mathrm{OPT}}\neq M^{\star}$.
Furthermore, by Lemma \ref{lem:dom1}, $\widehat{M}$ is also robustly
optimal.

\noindent\textbf{Case A}. Suppose first that $M^{\star}$ is robustly
optimal, in which case Lemma \ref{lemm-theta_m-star} in the main
text implies that $\theta^{m}=\tub>\tlb$. Because $\widehat{M}$
is also robustly optimal, Corollary \ref{cor:qstaropt} in the main
text then implies that $\qopt(\theta)=\hat{q}(\theta)$ for all $\theta>\underline{\theta}$.
This in turn implies that, for all $\theta$, $\uopt(\theta)=\hat{u}(\theta)$.
Because $\qopt(\underline{\theta})\in Q^{\#}(\underline{\theta})$,
there exists $V\in\mathcal{V}$ such that ${V}(\qopt(\underline{\theta}))-\underline{\theta}\qopt(\underline{\theta})-\uopt(\underline{\theta})>V(\hat{q}{\rm (\underline{\theta})})-\underline{\theta}{\rm \hat{q}(\underline{\theta})-}\hat{u}(\underline{\theta})$,
where the strict inequality follows from the concavity of $V({\rm q)-\tlb{\rm q}}$
in ${\rm q}$. Because Nature can always pair $V$ with a Dirac distribution
selecting $\underline{\theta}$ with probability one, the above inequality
contradicts the assumption that $\widehat{M}$ dominates $M^{\mathrm{OPT}}$.

\noindent\textbf{Case B}. Next, suppose that $M^{\star}$ is not
robustly optimal. Because $M^{\mathrm{OPT}}$ and $\widehat{M}$ are
both robustly optimal and $M^{\star}$ is not, Proposition \ref{prop:robust-quantity-mechanism-general}
in the main text then implies that $\hat{q}$ and $\qopt$ can disagree
only over the set $\{\underline{\theta}\}\cup[\theta^{m},\theta^{\star})$,
with the thresholds $\theta^{m}$ and $\theta^{\star}$ as defined
in the main text.

\begin{lemma}\label{lem:Dqr} Suppose $M^{\mathrm{OPT}}=(\qopt,\uopt)$
is a robustly optimal mechanism such that $\qopt$ is left-continuous
and, if $\theta^{m}>\underline{\theta}$, $\qopt(\underline{\theta})\in Q^{\#}(\underline{\theta})$.
If $M^{\mathrm{OPT}}$ is dominated by $\widehat{M}=(\hat{q},\hat{u})$,
then the set 
\[
\Theta_{\hat{q},\qopt}\equiv\{\theta\in[\theta^{m},\theta^{\star}):\hat{q}(\theta)<\qopt(\theta)\}
\]
has positive Lebesgue measure. \end{lemma}

\noindent\textbf{Proof of Lemma \ref{lem:Dqr}}. Assume, toward a
contradiction, that $\Theta_{\hat{q},\qopt}$ has zero Lebesgue measure.
Then, either (1) $\hat{q}(\theta)\geq\qopt(\theta)$ for all $\theta\in[\theta^{m},\theta^{\star})$,
or (2) $\hat{q}(\theta)<\qopt(\theta)$ only on countably many $\theta\in[\theta^{m},\theta^{\star})$.
Below, we establish a contradiction to the fact that $M^{\mathrm{OPT}}$
is dominated by $\widehat{M}$ in each of these two cases.

\noindent\textbf{Case 1.} Suppose that $\hat{q}(\theta)\geq\qopt(\theta)$
for all $\theta\in[\theta^{m},\theta^{\star})$. If this last inequality
holds as an equality for all $\theta\in[\theta^{m},\theta^{\star})$,
then, because $M^{{\rm OPT}}\neq\widehat{M}$, it must be that $\theta^{m}>\tlb$.
Furthermore, because $\hat{q}(\theta)=\qopt(\theta)={\rm q}_{\ell}$
for all $\theta\ge\theta^{\star}$, $\hat{q}$ and $\qopt$ disagree
only at $\theta=\underline{\theta}$, implying that $\uopt(\theta)=\hat{u}(\theta)$
for all $\theta$. The assumption that $\qopt(\underline{\theta})\in Q^{\#}(\underline{\theta})$
then implies that $\widehat{M}$ does not dominate $M^{\mathrm{OPT}}$,
a contradiction to the assumption it does (the arguments are identical
to those of Case A above). Thus there must exist $\theta'\in[\theta^{m},\theta^{\star})$
such that $\hat{q}(\theta')>\qopt(\theta')$.

First suppose that $\theta^{m}>\underline{\theta}$. Because $\qopt$
is left-continuous and $\hat{q}$ is weakly decreasing, and because
$\hat{q}(\theta)=\qopt(\theta)$ for all $\theta\in(\underline{\theta},\theta^{m})$,
it must be that $\theta'>\theta^{m}$. That $\hat{q}(\theta')>\qopt(\theta')$,
along with the fact that $\qopt$ is left continuous and $\hat{q}$
is weakly decreasing, implies that there exists a subset of $(\theta^{m},\theta^{\star})$
of positive Lebesgue measure over which $\hat{q}(y)>\qopt(y)$. Together
with the fact that $\hat{q}(\theta)\ge\qopt(\theta)$ for all $\theta\ge\theta^{m}$,
this last property implies that $\hat{u}(\theta^{m})>\uopt(\theta^{m})$.
Next, observe that, because $\qopt$ is left continuous, and because
$\qopt(\theta)=\qbm(\theta)$ for all $\theta\in(\tlb,\theta^{m})$
(by virtue of Proposition \ref{prop:robust-quantity-mechanism-general}
in the main text), we have that $\qopt(\theta^{m})=\qbm(\theta^{m})$.
Furthermore, Lemma \ref{lem:qstardlowint} in the main text implies
that $\qbm(\theta^{m})=\underline{D}(\theta^{m})$. Because $\widehat{M}$
is robustly optimal, Lemma \ref{lem:weakqoptbound} in the main text
then implies that $\hat{q}(\theta^{m})\le\qbm(\theta^{m})$. Jointly,
the last three properties imply that $\hat{q}(\theta^{m})\le\qbm(\theta^{m})=\underline{D}(\theta^{m})=\qopt(\theta^{m})$.
This last inequality, together with the fact that $\hat{u}(\theta^{m})>\uopt(\theta^{m})$,
and the fact that $\underline{D}(\theta^{m})$ is the unique maximizer
of $\underline{V}({\rm q})-\theta^{m}{\rm q}$, implies that 
\[
\underline{V}(\hat{q}(\theta^{m}))-\theta^{m}\hat{q}(\theta^{m})-\hat{u}(\theta^{m})<\underline{V}(\qopt(\theta^{m}))-\theta^{m}\qopt(\theta^{m})-\uopt(\theta^{m}).
\]
This means that $W(\widehat{M};\underline{V},\delta_{\theta^{m}})<W(M^{\mathrm{OPT}};\underline{V},\delta_{\theta^{m}})$,
a contradiction to the assumption that $\widehat{M}$ dominates $M^{\mathrm{OPT}}$.

Next suppose that $\theta^{m}=\underline{\theta}$. Claim \ref{claim:qopt_tlb}
(stated and proved below at the end of the proof) establishes that
$\qopt(\underline{\theta})\geq\underline{D}(\underline{\theta})$.
Hence, $\underline{D}(\underline{\theta})\le\qopt(\underline{\theta})\le\hat{q}(\underline{\theta})$,
where the second inequality follows from the fact, by assumption,
$\hat{q}(\theta)\geq\qopt(\theta)$ for all $\theta\in[\theta^{m},\theta^{\star})$.
Now recall that we have established above that there must exist $\theta'\in[\theta^{m},\theta^{\star})$
such that $\hat{q}(\theta')>\qopt(\theta')$.

If $\theta'=\theta^{m}$, because $\theta^{m}=\underline{\theta}$
and because $\underline{D}(\underline{\theta})\le\qopt(\underline{\theta})\le\hat{q}(\underline{\theta})$,
we then have that $\underline{D}(\underline{\theta})\le\qopt(\underline{\theta})<\hat{q}(\underline{\theta})$.
This last property, together with the fact that $\hat{u}(\underline{\theta})\geq\uopt(\underline{\theta})$
(by virtue of the fact that $\hat{q}(\theta)\geq\qopt(\theta)$ for
all $\theta\geq\theta^{m}=\underline{\theta}$) and the property that
$\underline{V}({\rm q})-\underline{\theta}{\rm q}$ is concave in
${\rm q}$ with a unique maximum at $\underline{D}(\underline{\theta})$
then implies that $W(\widehat{M};\underline{V},\delta_{\theta^{m}})<W(M^{\mathrm{OPT}};\underline{V},\delta_{\theta^{m}})$,
which contradicts the assumption that $\widehat{M}$ dominates $M^{\mathrm{OPT}}$.

If, instead, $\theta'>\theta^{m}$, then, as argued above, $\hat{u}(\theta^{m})>\uopt(\theta^{m})$.
This property, together with the fact that $\underline{D}(\underline{\theta})\le\qopt(\underline{\theta})\leq\hat{q}(\underline{\theta})$
and the concavity of $\underline{V}({\rm q})-\underline{\theta}{\rm q}$
in ${\rm q}$ then implies again that $W(\widehat{M};\underline{V},\delta_{\theta^{m}})<W(M^{\mathrm{OPT}};\underline{V},\delta_{\theta^{m}})$,
a contradiction to the assumption that $\widehat{M}$ dominates $M^{\mathrm{OPT}}$.

\noindent\textbf{Case 2.} Now suppose that there exists $\theta'\in[\theta^{m},\theta^{\star})$
such that $\hat{q}(\theta')<\qopt(\theta')$ but the set 
\[
\{\theta\in[\theta^{m},\theta^{\star}):\hat{q}(\theta)<\qopt(\theta)\}
\]
has zero Lebesgue measure. 
Because $\hat{q}(\theta)\geq\qopt(\theta)$ for almost all $\theta\geq\theta^{m}$,
$\hat{u}(\theta')\geq\uopt(\theta')$.

\noindent Suppose $\theta'>\underline{\theta}$. Lemma \ref{lem:weakqoptbound}
in the main text then implies that $\hat{q}(\theta')<\qopt(\theta')\leq\qbm(\theta')\leq D^{\star}(\theta')$.
Along with the concavity of the function $V^{\star}({\rm q})-\theta'{\rm q}$
in ${\rm q}$ (reaching a maximum at $D^{\star}(\theta'$)), these
last series of inequalities implies that 
\[
V^{\star}(\hat{q}(\theta'))-\theta'\hat{q}(\theta')-\hat{u}(\theta')<V^{\star}(\qopt(\theta'))-\theta'\qopt(\theta')-\uopt(\theta').
\]
Hence, $W(\widehat{M};V^{\star},\delta_{\theta'})<W(M^{\mathrm{OPT}};V^{\star},\delta_{\theta'})$,
a contradiction to the assumption that $\widehat{M}$ dominates $M^{\mathrm{OPT}}$.

Finally, suppose that $\theta'=\underline{\theta}$. Because $\theta'\in[\theta^{m},\theta^{\star})$,
this implies that $\theta^{m}=\tlb$. Thus, we have that $\hat{q}(\tlb)<\qopt(\tlb)\le D^{\star}(\tlb)$,
where the last inequality follows from Claim \ref{claim:qopt_tlb}.
The same arguments as above then imply that $W(\widehat{M};V^{\star},\delta_{\theta'})<W(M^{\mathrm{OPT}};V^{\star},\delta_{\theta'})$,
a contradiction to the assumption that $\widehat{M}$ dominates $M^{\mathrm{OPT}}$.

\noindent Combining the two cases, we thus conclude that $\Theta_{\hat{q},\qopt}$
has positive Lebesgue measure.\hfill{}$\blacksquare$

\noindent We now use Lemma \ref{lem:Dqr} to establish a contradiction
to the assumption that $\widehat{M}$ dominates $M^{\mathrm{OPT}}$.
Let 
\[
\theta_{h}\equiv\sup\{\theta\in{\Theta}_{\hat{q},\qopt}\}.
\]
Observe that, because ${\Theta}_{\hat{q},\qopt}$ has positive Lebesgue
measure, $\theta_{h}>\theta^{m}$. Furthermore, by definition, $\theta_{h}\le\theta^{\star}$.
Proposition \ref{prop:robust-quantity-mechanism-general} in the main
text then implies that 
\begin{align}
\qopt(\theta_{h})\le\qbm(\theta_{h})<D^{\star}(\theta_{h}).\label{eq:e1}
\end{align}
The rest of the proof below shows that there exists $\theta$ such
that $W(M^{{\rm OPT}};V^{\star},\delta_{\theta})>W(\widehat{M};V^{\star},\delta_{\theta})$.

First, consider the case in which $\hat{q}(\theta_{h})<\qopt(\theta_{h})$.
Condition (\ref{eq:e1}) then implies that $\hat{q}(\theta_{h})<\qopt(\theta_{h})<D^{\star}(\theta_{h})$.
Because $V^{\star}(\mathrm{q})-\theta_{h}\mathrm{q}$ is concave in
$\mathrm{q}$, with a maximum at $D^{\star}(\theta_{h})$, $V^{\star}(\hat{q}(\theta_{h}))-\theta_{h}(\hat{q}(\theta_{h}))<V^{\star}(\qopt(\theta_{h}))-\theta_{h}\qopt(\theta_{h})$.
Furthermore, because, for all $y>\theta_{h}$, $\hat{q}(y)\geq\qopt(y)$,
$\int^{\bar{\theta}}_{\theta_{h}}\hat{q}(y)dy\geq\int^{\bar{\theta}}_{\theta_{h}}\qopt(y)dy$.
Jointly, these properties imply that $W(\widehat{M};V^{\star},\delta_{\theta_{h}})<W(M^{{\rm {OPT}}};V^{\star},\delta_{\theta_{h}})$,
a contradiction to the assumption that $\widehat{M}$ dominates $M^{{\rm {OPT}}}$.

Next, suppose that $\hat{q}(\theta_{h})=\qopt(\theta_{h})$. The left-continuity
of $\qopt$ along with the fact that $\qopt(\theta_{h})<D^{\star}(\theta_{h})$
(by Condition (\ref{eq:e1})) then imply that there exists $\theta_{\ell}\in(\theta^{m},\theta_{h})$
such that, for all $\theta\in\mathcal{N}\equiv[\theta_{\ell},\theta_{h})$,
$\qopt(\theta)<D^{\star}(\theta_{h})$. Then let $\Delta\equiv\sup_{y\in\mathcal{N}}\big[\qopt(y)-\hat{q}(y)\big]$.
Because $\mathcal{N}\cap\Theta_{\hat{q},\qopt}\neq\emptyset$, $\Delta>0$.
By the definition of $\Delta$, there exists $\theta\in\mathcal{N}$
and $\epsilon>0$, with 
\begin{equation}
\epsilon<\Delta\frac{P^{\star}(\qopt(\theta_{\ell}))-\theta_{h}}{P^{\star}(\qopt(\theta_{\ell}))-\theta^{m}}<\Delta,\label{eq:new-inequality-APP}
\end{equation}
such that $\qopt(\theta)-\hat{q}(\theta)>\Delta-\epsilon$. Note that
the second inequality in (\ref{eq:new-inequality-APP}) follows from
the fact that $\qopt(\theta_{\ell})<D^{\star}(\theta_{h})$ (by virtue
of the fact that $\theta_{\ell}\in\mathcal{N}$) which in turn implies
that $P^{\star}(\qopt(\theta_{\ell}))>\theta_{h}>\theta^{m}$. Moreover,
because $\Delta-\epsilon>0$, $\theta\in\Theta_{\hat{q},\qopt}$,
which implies that $\qopt(\theta)-\hat{q}(\theta)>0$. Next, observe
that 
\begin{align*}
 & W(M^{\mathrm{OPT}};V^{\star},\delta_{\theta})-W(\widehat{M};V^{\star},\delta_{\theta})=\int^{\qopt(\theta)}_{\hat{q}(\theta)}\big(P^{\star}(z)-\theta\big){\rm d}z-\int^{\overline{\theta}}_{\theta}\big(\qopt(y)-\hat{q}(y)\big){\rm d}y\\
 & =\int^{\qopt(\theta)}_{\hat{q}(\theta)}\big(P^{\star}(z)-\theta\big){\rm d}z-\int^{\theta_{h}}_{\theta}\big(\qopt(y)-\hat{q}(y)\big){\rm d}y-\int^{\overline{\theta}}_{\theta_{h}}\big(\qopt(y)-\hat{q}(y)\big){\rm d}y\\
 & \ge\int^{\qopt(\theta)}_{\hat{q}(\theta)}\big(P^{\star}(z)-\theta\big){\rm d}z-\int^{\theta_{h}}_{\theta}\big(\qopt(y)-\hat{q}(y)\big){\rm d}y\\
 & \ge\big(P^{\star}(\qopt(\theta_{\ell}))-\theta\big)(\qopt(\theta)-\hat{q}(\theta))-\Delta(\theta_{h}-\theta)\\
 & \ge\big(P^{\star}(\qopt(\theta_{\ell}))-\theta\big)(\Delta-\epsilon)-\Delta(\theta_{h}-\theta)=\Delta\big(P^{\star}(\qopt(\theta_{\ell}))-\theta_{h}\big)-\epsilon\big(P^{\star}(\qopt(\theta_{\ell}))-\theta\big)\\
 & \ge\Delta\big(P^{\star}(\qopt(\theta_{\ell}))-\theta_{h}\big)-\epsilon\big(P^{\star}(\qopt(\theta_{\ell}))-\theta^{m}\big)>0.
\end{align*}
The first two equalities follow from the definitions of $V^{\star}$,
$W(M^{\mathrm{OPT}};V^{\star},\delta_{\theta})$, and $W(\widehat{M};V^{\star},\delta_{\theta})$.
The first inequality follows from the fact that, by definition of
$\theta_{h}$, $\ensuremath{\qopt(y)\le\hat{q}(y)}$ for all $y>\theta_{h}$.
The second inequality follows from the monotonicity of $\qopt$, the
fact that $\theta>\theta_{\ell}$, and the definition of $\Delta$.
The third inequality follows from the choice of $\theta$. The third
equality is immediate. The fourth inequality follows from the fact
that $\theta>\theta^{m}$. The last inequality follows from the fact
that $\epsilon<\Delta\frac{P^{\star}(\qopt(\theta_{\ell}))-\theta_{h}}{P^{\star}(\qopt(\theta_{\ell}))-\theta^{m}}$.
We thus conclude that $W(M^{\mathrm{OPT}};V^{\star},\delta_{\theta})>W(\widehat{M};V^{\star},\delta_{\theta})$.
This, however, contradicts the assumption that $\widehat{M}$ dominates
$M^{{\rm {OPT}}}$.

Because a contradiction to the assumption that $\widehat{M}$ dominates
$M^{{\rm {OPT}}}$ is reached under all possible scenarios, we conclude
that $M^{{\rm {OPT}}}$ is undominated. \hfill{}Q.E.D.

We now state and prove a claim that we made in the proof of Lemma
\ref{lem:Dqr} above.

\begin{claim}\label{claim:qopt_tlb} Let $M^{{\rm OPT}}\equiv(\qopt,\uopt)$
be a robustly optimal mechanism. (A) Then $\qopt(\tlb)\geq\underline{D}(\tlb)$.
(B) Furthermore, if $\theta^{m}=\tlb$, then $\qopt(\tlb)\le D^{\star}(\tlb)$.
\end{claim}

\noindent\textbf{Proof of Claim \ref{claim:qopt_tlb}} Part (A).
Assume, towards a contradiction, that $\qopt(\tlb)<\underline{D}(\tlb)$.
Then $\qopt(\tlb)<\underline{D}(\tlb)\leq D^{\star}(\tlb)=\qbm(\tlb)$.
Continuity of $\qbm$ and $\underline{D}$, along with the monotonicity
of $\qopt$, then imply that there exists $0<\epsilon<\theta^{\star}-\tlb$
such that, for all $\theta\in[\tlb,\tlb+\epsilon]$, $\qopt(\theta)<\min\{\qbm(\theta),\underline{D}(\theta)\}$.
Then consider the mechanism $\widetilde{M}\equiv(\tilde{q},\tilde{u})$,
where the quantity schedule is given by 
\begin{align*}
\tilde{q}(\theta)=\begin{cases}
\min\{\qbm(\theta),\underline{D}(\theta)\} & \textrm{if}~\theta\in[\tlb,\tlb+\epsilon]\\
\qopt(\theta) & \textrm{otherwise,}
\end{cases}
\end{align*}
and $\tilde{u}(\theta)=\int^{\tub}_{\theta}\tilde{q}(y){\rm d}y$
for all $\theta\in\Theta$. We claim that $\widetilde{M}\in{\mathcal{M}}^{{\rm SL}}$.
To see this, first observe that $\tilde{q}$ is weakly decreasing
because $\qopt$, $\underline{D}$, and $\qbm$ are weakly decreasing,
and $\tilde{q}(\tlb+\epsilon)>\qopt(\tlb+\epsilon)$. Hence, $\widetilde{M}$
is IC and IR. Next, observe that, because $\tlb+\epsilon<\theta^{\star}$,
$\tilde{q}(\tub)=\qopt(\tub)={\rm q}_{\ell}$. Because $\qbm(\tlb)=D^{\star}(\tlb)$,
$\min\{\qbm(\tlb),\underline{D}(\theta)\}=\underline{D}(\theta)$
and hence $\tilde{q}(\tlb)=\underline{D}(\tlb)$. This means that
$\underline{{\rm DWL}}(\tlb,\tilde{q}(\tlb))=0$. Lemma \ref{lemma:majorizations}
in the main text then implies that the robustness constraint for every
$\theta\in[\tlb,\tub)$ reduces to $\int^{\tub}_{\theta}\left(\underline{D}(y)-\tilde{q}(y)\right){\rm d}y\ge0$.
For $\theta>\tlb+\epsilon$, the inequality holds because $\tilde{q}=\qopt$
and $\qopt$ is robustly optimal. For $\theta\in[\tlb,\tlb+\epsilon]$,
we have that 
\[
\int\limits^{\tub}_{\theta}\left(\underline{D}(y)-\tilde{q}(y)\right){\rm d}y=\int\limits^{\tlb+\epsilon}_{\theta}\left(\underline{D}(y)-\tilde{q}(y)\right){\rm d}y+\int\limits^{\tub}_{\tlb+\epsilon}\left(\underline{D}(y)-\qopt(y)\right){\rm d}y\ge0,
\]
where the inequality follows from the fact that the first integral
is non-negative because of the definition of $\tilde{q}$ and the
second integral is non-negative because of the fact that $\qopt$
is a robustly optimal schedule. Together the above properties thus
imply that $\widetilde{M}\in{\mathcal{M}}^{{\rm SL}}$.

Now observe that, by the definition of $\tilde{q}$, for all $\theta\in[\tlb,\tlb+\epsilon]$,
$\qopt(\theta)<\tilde{q}(\theta)\leq\qbm(\theta)$. Because, for all
$\theta$, $V^{\star}({\rm q})-z^{\star}(\theta){\rm q}$ is concave
with a maximum at $\qbm(\theta)$, and because $\tilde{u}(\tub)=0$,
we conclude that $W(\widetilde{M};V^{\star},F^{\star})>W(M^{{\rm OPT}};V^{\star},F^{\star})$,
a contradiction to the robust optimality of $M^{{\rm OPT}}$. Hence,
$\qopt(\tlb)<\underline{D}(\tlb)$, as claimed.

Part (B). Next we establish that, if $\theta^{m}=\tlb$, then $\qopt(\tlb)\le D^{\star}(\tlb)$.
Because $\theta^{m}=\tlb$, Lemma \ref{lemm-theta_m-star} in the
main text implies that $M^{\star}$ is not robustly optimal. Consider
the following two cases: (i) $\theta^{m}=\theta^{\star\star}$, and
(ii) $\theta^{m}<\theta^{\star\star}$.

First, suppose that $\theta^{m}=\theta^{\star\star}$. Then, $\qopt(\theta)=\underline{D}(\theta)$
for all $\theta>\tlb$. Thus, 
\[
S(\tlb;\qopt)\equiv\int\limits^{\tub}_{\theta}\left(\underline{D}(y)-\qopt(y)\right){\rm d}y=0,
\]
and, therefore, 
\[
R(\qopt)=S(\tlb;\qopt)-\underline{{\rm DWL}}(\tlb,\qopt(\tlb))=-\underline{{\rm DWL}}(\tlb,\qopt(\tlb)).
\]
Because $\underline{{\rm DWL}}(\tlb,\qopt(\tlb))\ge0$, and because
$R(\qopt)\geq0$, we conclude that $\underline{{\rm DWL}}(\tlb,\qopt(\tlb))=0$,
which implies that $\qopt(\tlb)=\underline{D}(\tlb)\le D^{\star}(\tlb)$,
as claimed.

Next suppose that $\theta^{m}<\theta^{\star\star}$, and assume, towards
a contradiction, that $\qopt(\tlb)>D^{\star}(\tlb)$. Consider the
mechanism $\widetilde{M}\equiv(\tilde{q},\tilde{u})$, where 
\[
\tilde{q}(\theta)=\begin{cases}
D^{\star}(\theta) & \mbox{if}~\theta=\tlb\\
\qopt(\theta) & \mbox{otherwise}
\end{cases}~~~~~~~~\tilde{u}(\theta)=\int\limits^{\tub}_{\theta}\tilde{q}(y){\rm d}y~~~\forall~\theta\in\Theta.
\]
Note that $\tilde{q}$ is weakly decreasing because $\qopt$ is weakly
decreasing and, for every $\theta\in(\tlb,\theta^{\star})$, $\qopt(\theta)\le\qbm(\theta)<\qbm(\tlb)=D^{\star}(\tlb)$,
where the first inequality follows from Lemma \ref{lem:weakqoptbound},
whereas the second inequality follows from the fact that $\qbm$ is
decreasing. Next, observe that 
\begin{equation}
\underline{V}(D^{\star}(\tlb))-\tlb D^{\star}(\tlb)-\uopt(\tlb)>\underline{V}(\qopt(\tlb))-\tlb\qopt(\tlb)-\uopt(\tlb)\ge G^{*}.\label{eq:qopt_tlb}
\end{equation}
The first inequality follows from the fact that $\qopt(\tlb)>D^{\star}(\tlb)\ge\underline{D}(\tlb)$
along with the the fact that $\underline{V}({\rm q})-\tlb{\rm q}$
is concave in ${\rm q}$ attaining a maximum at $\underline{D}(\tlb)$.
The second inequality holds because $M^{{\rm OPT}}\in\mathcal{M}^{{\rm SL}}$.
Because $\tilde{u}(\tlb)=\int\limits^{\tub}_{\tlb}\tilde{q}(y){\rm d}y=\uopt(\tlb)$
and $\tilde{q}(\tlb)=D^{\star}(\tlb)$, we thus have that 
\[
R(\tilde{q})=\underline{V}(\tilde{q}(\tlb))-\tlb\tilde{q}(\tlb)-\int\limits^{\tub}_{\tlb}\tilde{q}(y){\rm d}y-G^{*}>0.
\]
Furthermore, because for all $\theta>\tlb$, $\tilde{q}(\theta)=\qopt(\theta)$,
and because $\mathcal{M}^{{\rm OPT}}$ is robustly optimal, we have
that $S(\theta;\tilde{q})=S(\theta;\qopt)\ge0$ for all $\theta\in(\tlb,\tub)$
and $\tilde{q}(\tub)={\rm q}_{\ell}$. We thus conclude that $\widetilde{M}\in\mathcal{M}^{{\rm SL}}.$
Finally, by definition of $\widetilde{M}$, we have that $W(\widetilde{M};V^{\star},F^{\star})=W(M^{{\rm OPT}};V^{\star},F^{\star})$.
Thus, $\widetilde{M}$ is also robustly optimal.

Because $\widetilde{M}$ is robustly optimal, the analysis in Section
\ref{subsec:Stronger-Prop-2} of \nameref{Sec:OS} implies that there
exist $\widetilde{\mu}\ge0$ and a weakly increasing, right continuous
function $\widetilde{\Lambda}$, with $\widetilde{\Lambda}(\tlb)=0$,
such that $\tilde{q}$ satisfies the analog of Conditions (\ref{eq:lobj})
and (\ref{eq:cs1}) with $\muopt$ replaced by $\widetilde{\mu}$
and $\Lopt$ replaced by $\widetilde{\Lambda}$. The complementary
slackness conditions along with the fact that $R(\tilde{q})>0$ implies
that $\widetilde{\mu}=0$. Moreover, because $R(\tilde{q})=S(\tlb;\tilde{q})-\underline{{\rm DWL}}(\tlb,\tilde{q}(\tlb))>0$
and because $\underline{{\rm DWL}}(\tlb,\tilde{q}(\tlb))\ge0$, we
have that $S(\tlb;\tilde{q})>0$. Continuity of $S(\cdot;\tilde{q})$
then implies that there exists $\delta>0$ such that $S(\theta;\tilde{q})>0$
for all $\theta\in[\tlb,\tlb+\delta]$. Thus, by the complementary
slackness conditions, we have that $\widetilde{\Lambda}(\theta)=0$
for all $\theta\in[\tlb,\tlb+\delta]$. We thus have that $\widetilde{\Lambda}(\theta)+\widetilde{\mu}=0$
for all $\theta\in[\tlb,\tlb+\delta]$. Because $\tlb=\theta^{m}$,
Observation \ref{obs:leftof_tmin} in \nameref{Sec:OS} then implies
that $\tilde{q}(\theta)=\qbm(\theta)$ for all $\theta\in(\theta^{m},\theta^{m}+\delta]$.
Because $\tilde{q}(\theta)=\qopt(\theta)$ for all $\theta\in(\theta^{m},\theta^{m}+\delta]$,
this implies that $\qopt(\theta)=\qbm(\theta)$ for all $\theta\in(\theta^{m},\theta^{m}+\delta]$.
However, this contradicts the fact that $\qopt(\theta)<\qbm(\theta)$
for all $\theta\in(\theta^{m},\min\{\theta^{\star},\theta^{\star\star}\})$
as established in Section \ref{subsec:Stronger-Prop-2} of \nameref{Sec:OS}.\footnote{Recall that, for any robustly optimal quantity schedule $q$, $\theta^{\star\star}_{q}\equiv\inf\{\theta\in[\theta^{m},\tub]:q(y)=\underline{D}(y)~\forall~y\in[\theta,\tub]\}$.
For $\tilde{q}$ and $\qopt$ the cut-off is the same, that is, $\theta^{\star\star}_{\tilde{q}}=\theta^{\star\star}_{\qopt}=\theta^{\star\star}$.} We thus conclude that $\qopt(\tlb)\le D^{\star}(\tlb)$, as claimed.\hfill{}Q.E.D.

\noindent We conclude by commenting on the assumption that, when $\theta^{m}>\underline{\theta}$,
$\qopt(\underline{\theta})\in Q^{\#}(\underline{\theta})$. Proposition
\ref{prop:robust-quantity-mechanism-general} in the main text implies
that 
\[
\lim_{\theta\downarrow\underline{\theta}}\qopt(\theta)=\qbm(\underline{\theta})=D^{\star}(\underline{\theta})=\arg\max_{{\rm q}\in[0,\ensuremath{\bar{{\rm q}}}]}\{V^{\star}({\rm q})-\underline{\theta}{\rm q}\}.
\]
Because $\qopt$ is weakly decreasing, $\qopt(\underline{\theta})\geq D^{\star}(\underline{\theta})$.
Because $D^{\star}(\underline{\theta})\geq\underline{D}(\underline{\theta})$
and 
\[
\underline{D}(\underline{\theta})=\arg\max_{{\rm q}\in[0,\ensuremath{\bar{{\rm q}}}]}\left\{ \underline{V}({\rm q})-\underline{\theta}{\rm q}-\int^{\bar{\theta}}_{\underline{\theta}}\qopt(y)\mathrm{d}y\right\} 
\]
by choosing $\qopt(\underline{\theta})=D^{\star}(\underline{\theta})$,
the designer can then relax the robustness constraint at $\underline{\theta}$
with no effect on the robustness constraints for any $\theta>\underline{\theta}$
and on the buyer's payoff under the conjectured model. This means
that, when $\theta^{m}>\underline{\theta}$, starting from any robustly
optimal mechanism $M^{{\rm OPT}}$ in which $\qopt$ is left-continuous,
one can construct another robustly optimal mechanism $\mathring{M}^{{\rm OPT}}$
in which $\mathring{q}^{{\rm OPT}}$ is left-continuous and satisfies
$\mathring{q}^{{\rm OPT}}(\underline{\theta})=D^{\star}(\underline{\theta})$.
Because $D^{\star}(\underline{\theta})\in Q^{\#}(\underline{\theta})$,
Proposition \ref{prop-undomination-gen} then implies that the mechanism
$\mathring{M}^{{\rm OPT}}$ is undominated.

\subsection{Non-monotonicity of the procured quantity in $\underline{F}$}

\label{sec:compstat}

As shown in \nameref{Sec:OS}, the quantity procured under robustly
optimal quantity mechanisms depends on the set $\mathcal{F}$ of feasible
technologies only through $F^{\star}$ and $\underline{F}$. In this
section, we show that, as $\underline{F}$ becomes larger (in a sense
made precise below), the quantity procured under robustly optimal
mechanisms may change in a non-monotonic way. To simplify the exposition,
we assume $V^{\star}=\underline{V}$, which implies that Baron-Myerson-with-quantity-bridge
mechanisms of Definition \ref{def-BM-bridge} are robustly optimal.
Below, we analyze how the quantity procured under such mechanisms
changes as $\underline{F}$ changes.

Fix $F^{\star}$ and consider a sequence $(\underline{F}_{n})$ of
cdfs (corresponding to the lowest elements of the set $\mathcal{F}$
of feasible technologies) with the following properties: 
\begin{enumerate}
\item[(a)] for every $n$ there exists ${\underline{\theta}}_{n}\in\Theta$
and $\delta_{n}\ge0$ such that, 
\begin{itemize}
\item[(1)] $\underline{F}_{n}$ is absolutely continuous over $(-\infty,\overline{\theta})$,
with density $\underline{f}_{n}(\theta)>0$ for all $\theta\in[{\underline{\theta}}_{n},\overline{\theta})$;
if there is no atom at $\overline{\theta}$, we let $\underline{f}_{n}(\overline{\theta})>0$
denote the density of $\underline{F}_{n}$ at $\theta=\overline{\theta}$, 
\item[(2)] $\underline{F}_{n}(\theta)=0$ for all $\theta<{\underline{\theta}}_{n}$,
$\underline{F}_{n}(\theta)=1$ for all $\theta\geq\overline{\theta}$, 
\item[(3)] $\lim_{\theta\uparrow\overline{\theta}}\underline{F}_{n}(\theta)=1-\delta_{n}$, 
\end{itemize}
\item[(b)] for every $n$, $\underline{\theta}_{n+1}\ge\underline{\theta}_{n}$,
and for every $\theta\in(\underline{\theta},\overline{\theta})$,
there exists $n$ such that $\theta<\underline{\theta}_{n}<\overline{\theta}$,
\item[(c)] for every $n$, $\delta_{n+1}\geq\delta_{n}$, 
\item[(d)] there exists $\overline{n},\overline{\overline{n}}\in\mathbb{N}\cup\{+\infty\}$
with $\overline{\overline{n}}>\overline{n}$ such that $\underline{\theta}_{n}=\underline{\theta}$
if, and only if, $n\leq\overline{n}$, and $\delta_{n}>0$ if, and
only if, $n\geq\overline{\overline{n}}$. 
\item[(e)] for every $n$, the function $\underline{z}_{n}:[{\underline{\theta}}_{n},\overline{\theta}]\rightarrow\mathbb{R}$
defined by 
\begin{align*}
\underline{z}_{n}(\theta)\equiv\begin{cases}
\theta+\underline{F}_{n}(\theta)/\underline{f}_{n}(\theta) & ~~\textrm{if}~\theta\in[{\underline{\theta}}_{n},\overline{\theta})\\
\overline{\theta}+1/\underline{f}_{n}(\overline{\theta}) & ~~\textrm{if}~\theta=\overline{\theta}~\textrm{and}~\delta_{n}=0\\
\overline{\theta}+(1-\delta_{n})/\delta_{n} & ~~\textrm{if}~\theta=\overline{\theta}~\textrm{and}~\delta_{n}>0
\end{cases}
\end{align*}
is increasing over $[{\underline{\theta}}_{n},\overline{\theta}]$
and continuous over $[{\underline{\theta}}_{n},\overline{\theta})$. 
\item[(f)] for all $\theta\in[\underline{\theta}_{n+1},\overline{\theta}]$,
\begin{equation}
\frac{\underline{F}_{n+1}(\theta)}{\underline{f}_{n+1}(\theta)}\le\frac{\underline{F}_{n}(\theta)}{\underline{f}_{n}(\theta)},\label{eq:rhr}
\end{equation}
and for every $n$ and every $\theta\in[\underline{\theta}_{n},\overline{\theta}]$,
\begin{equation}
\underline{z}_{n}(\theta)<z^{\star}(\theta),
\label{eq:rhr-star}
\end{equation}
which is the case when $\frac{\underline{F}_{n}(\theta)}{\underline{f}_{n}(\theta)}<\frac{F^{\star}(\theta)}{f^{\star}(\theta)}$. 
\end{enumerate}
Figure \ref{fig:sequence} provides an illustration of the sequence
$(\underline{F}_{n})$. Note that property (f) above means that the
technologies are ranked in the reverse-hazard-rate order. The sequence
can thus be interpreted as capturing an increase in the severity of
the buyer's (downside) uncertainty over the technology that determines
the seller's cost.

\begin{figure}
\centering \includegraphics[width=0.7\linewidth]{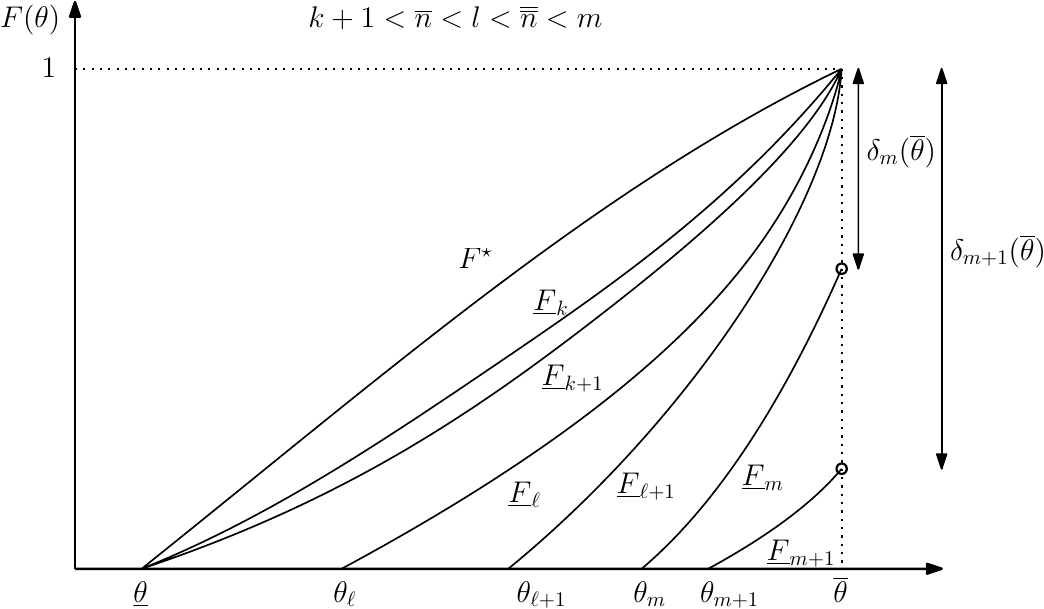}
\caption{Pictorial depiction of the sequence $(\underline{F}_{n})$.}
\label{fig:sequence} 
\end{figure}

Let $\qopt_{n}$ be a robustly optimal quantity schedule when the
lowest technology in $\mathcal{F}$ is $\underline{F}_{n}$. The following
proposition establishes that the quantity procured under any robustly
optimal quantity mechanism is not monotone in the lowest distribution
$\underline{F}_{n}$. This property holds despite the fact that, as
is well known, the Baron-Myerson quantity schedule $\underline{q}^{\textsc{BM}}_{n}$
defined, for all $\theta\in[\underline{\theta}_{n},\overline{\theta})$,
by 
\[
\underline{q}^{\textsc{BM}}_{n}(\theta)\equiv\arg\max_{\text{q}\in[0,\bar{\mathrm{q}}]}\left\{ V^{\star}(\text{q})-\underline{z}_{n}(\theta)\text{q}\right\} 
\]
is increasing in the inverse-hazard rare order: for any $n,n'\in\mathbb{N}$,
with $n'>n$ and any $\theta\geq\underline{\theta}_{n'}$, $\underline{q}^{\textsc{BM}}_{n'}(\theta)\geq\underline{q}^{\textsc{BM}}_{n}(\theta)$.
That is, when the buyer's model over the technology governing the
seller's cost coincides with the distribution $\underline{F}_{n}$,
an increase in the distribution (in the inverse-hazard-rate order)
leads to an increase in the output procured.

\begin{proposition}[Non-mon. of output in cost uncertainty] \label{prop:changes-cost-uncertainty}
Suppose that $V^{\star}=\underline{V}$. Let $(\underline{F}_{n})$
be any sequence of cdfs satisfying properties (a)-(f) above and let
$(M^{\textsc{opt}}_{n})$ be any sequence of mechanisms such that,
for each $n$, $M^{\textsc{opt}}_{n}\equiv(q^{\textsc{opt}}_{n},u^{\textsc{opt}}_{n})$
is a robustly optimal mechanism when the lowest distribution in $\mathcal{F}$
is $\underline{F}_{n}$. Then, for every $\theta\in(\underline{\theta},\overline{\theta})$, 
\begin{enumerate}
\item there exists $n(\theta)\in\mathbb{N}$ such that $\qopt_{n}(\theta)$
is weakly increasing in $n$ (respectively, weakly decreasing in $n$)
over $n\le n(\theta)-1$ (respectively, over $n>n(\theta)$). 
\item there exists $j,k\in\mathbb{N}$ with $j<k$ such that $\qopt_{j}(\theta)>\qopt_{k}(\theta)$. 
\end{enumerate}
\end{proposition}

\begin{figure}
\centering \includegraphics[width=0.7\linewidth]{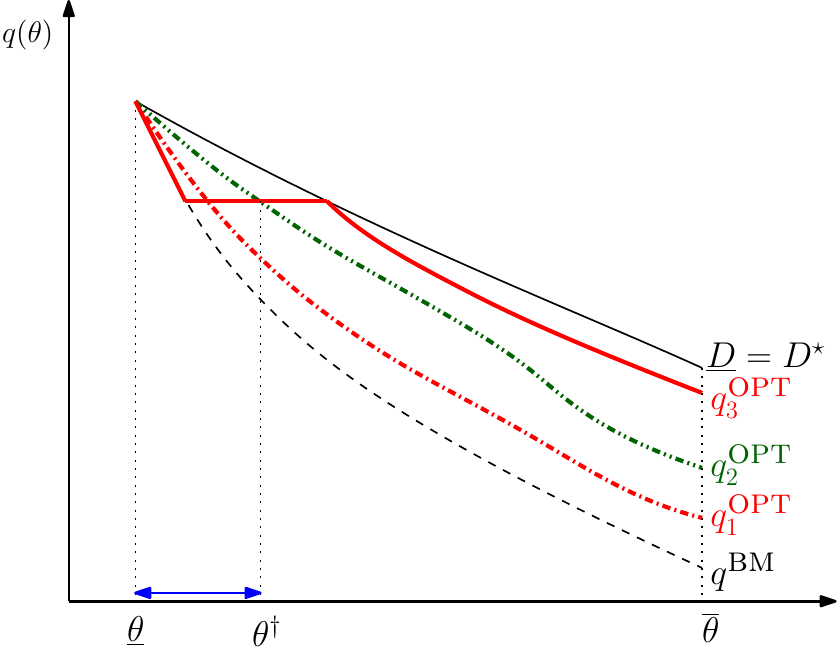}
\caption{Illustration of Proposition \ref{prop:changes-cost-uncertainty}.}
\label{fig:prop_pessimism} 
\end{figure}

Figure \ref{fig:prop_pessimism} illustrates the result in Proposition
\ref{prop:changes-cost-uncertainty}. For any $\theta\in[\underline{\theta},\theta^{\dag}]$,
as the lowest technology changes from $\underline{F}_{1}$ to $\underline{F}_{2}$,
the quantity procured increases. In fact, the robustly optimal quantity
schedule changes from the dash-dotted line to the dash-double-dotted
line. Note that both $\underline{F}_{1}$ to $\underline{F}_{2}$
have support $\Theta$; a reduction in the inverse of reverse hazard
rate then implies a reduction in the value of containing the rents
for the most efficient types, leading to an increase in the output
procured under the optimal mechanism. When the lowest technology changes
from $\underline{F}_{2}$ to $\underline{F}_{3}$, the robustly optimal
quantity schedule changes from the dash-double-dotted line to the
solid line and the quantity procured from types in the range $[\underline{\theta},\theta^{\dag}]$
is reduced. This is because the support of $\underline{F}_{3}$ no
longer contains low-cost types. The buyer can then afford to procure
less output from these types without jeopardizing robustness. Thus,
the quantity procured from types in the range $[\underline{\theta},\theta^{\dag}]$
is not monotone in $n$, equivalently, in the worst possible technology.\medskip{}

\noindent\textbf{Proof of Proposition \ref{prop:changes-cost-uncertainty}.}
Each part below establishes the corresponding part in the proposition.

\noindent\textbf{Part 1}. For any $\theta\in(\underline{\theta},\overline{\theta})$,
let $n(\theta)$ be the largest $n\ge\overline{n}$ such that $\underline{\theta}_{n}\le\theta<\underline{\theta}_{n+1}$.
Existence of such an $n(\theta)$ is guaranteed by Condition (b) in
the definition of the sequence $(\underline{F}_{n})$, which implies
that $\underline{\theta}_{n}\le\underline{\theta}_{n+1}<\overline{\theta}$
and $\lim_{n\rightarrow\infty}\underline{\theta}_{n}=\overline{\theta}$.
For any $n\le n(\theta)-1$, we have that $\theta\in[\underline{\theta}_{n},\overline{\theta}]$
and $\theta\in[\underline{\theta}_{n+1},\overline{\theta}]$. As established
in \nameref{Sec:OS}, $\qopt_{n}(\theta)=\underline{q}^{\textsc{BM}}_{n}(\theta)$
and $\qopt_{n+1}(\theta)=\underline{q}^{\textsc{BM}}_{n+1}(\theta)$.
Condition (\ref{eq:rhr}) in turn implies that $\underline{q}^{\textsc{BM}}_{n}(\theta)\le\underline{q}^{\textsc{BM}}_{n+1}(\theta)$,
that is, $\qopt_{n}(\theta)$ is weakly increasing in $n$ for $n\le n(\theta)-1$.

For any $n>n(\theta)$, $\theta<\underline{\theta}_{n}$, and therefore,
$\qopt_{n}(\theta)=\max\{\qbm(\theta),\underline{D}(\underline{\theta}_{n})\}$.
The quantity $\underline{D}(\underline{\theta}_{n})$ is weakly decreasing
in $n$ because $\underline{\theta}_{n}\le\underline{\theta}_{n+1}<\overline{\theta}$
for every $n$. Consequently, $\qopt_{n}(\theta)$ is also weakly
decreasing in $n$.

\noindent\textbf{Part 2.} To establish the second part of the proposition,
it suffices to exhibit a pair $j,k\in\mathbb{N}$, with $j<k$, such
that $\qopt_{j}(\theta)>\qopt_{k}(\theta)$. To do so, consider the
following two cases.

\noindent\textbf{Case 1.} Suppose $\qbm(\theta)\geq\underline{D}(\underline{\theta}_{n(\theta)+1})$.
Then let $j=n(\theta)$ and $k=n(\theta)+1$, and observe that 
\[
\qopt_{j}(\theta)=\underline{q}^{\textsc{BM}}_{j}(\theta)=\underline{D}\left(\theta+\frac{\underline{F}_{j}(\theta)}{\underline{f}_{j}(\theta)}\right)>\underline{D}\left(\theta+\frac{F^{\star}(\theta)}{f^{\star}(\theta)}\right)=\qbm(\theta)=\qopt_{k}(\theta),
\]
where the inequality follows from (\ref{eq:rhr-star}).

\noindent\textbf{Case 2.} Suppose $\qbm(\theta)<\underline{D}(\underline{\theta}_{n(\theta)+1})$.
Then let $j=n(\theta)+1$ and let $k$ be such that $\underline{\theta}_{k}>\underline{\theta}_{j}$.
Existence of such an $k$ is ensured by Condition (b) in the definition
of the sequence $(\underline{F}_{n})$. Then

\noindent
\[
\qopt_{k}(\theta)=\max\{\qbm(\theta),\underline{D}(\underline{\theta}_{k})\}<\underline{D}(\underline{\theta}_{j})=\qopt_{j}(\theta).
\]
To see this, observe that $\underline{\theta}_{k}>\underline{\theta}_{j}$
implies that $\underline{D}(\underline{\theta}_{k})<\underline{D}(\underline{\theta}_{j})$.
Hence, if $\qopt_{k}(\theta)=\text{\ensuremath{\underline{D}}(\ensuremath{\underline{\theta}_{k}})}$,
then $\qopt_{k}(\theta)=\underline{D}(\underline{\theta}_{k})<\underline{D}(\underline{\theta}_{j})=\qopt_{j}(\theta)$.
If, instead, $\qopt_{k}(\theta)=\qbm(\theta)$, the result follows
from the fact that, by assumption, $\qbm(\theta)<\underline{D}(\underline{\theta}_{j})$.
$\hfill Q.E.D.$

\end{document}